\documentclass[a4paper]{article}

\usepackage{graphicx}
\usepackage{wrapfig}
\usepackage{subcaption}
\usepackage{amsmath, amssymb, physics, amsthm, mathtools, dsfont}
\usepackage{multicol}
\usepackage{multirow}
\usepackage{hyperref}
\usepackage{geometry}
\usepackage{wrapfig}
\usepackage{etoolbox}
\usepackage{authblk}

\usepackage{todonotes}
\todostyle{holl}{color=green!30, size=\footnotesize}
\todostyle{sas}{color=red!30, size=\footnotesize}

\usepackage[inline]{enumitem}
\usepackage{tabularx}

\usepackage{framed}

\usepackage{xcolor}
\definecolor{quantumviolet}{HTML}{53257F}
\definecolor{quantumgray}{HTML}{555555}
\hypersetup{
    colorlinks,
    citecolor=magenta,
    linkcolor=cyan,
    urlcolor=magenta
    }

\providetoggle{darkmode}
\settoggle{darkmode}{false}
\iftoggle{darkmode}{
    \pagecolor{darkbackground}
    \color{white}
    \hypersetup{
    colorlinks,
    citecolor=Red,
    linkcolor=Red,
    urlcolor=Red
    }
}

\usepackage{hyperref}
\usepackage{cleveref}

\newcommand{\Acc}{\mathsf{Acc}}
\newcommand{\Rej}{\mathsf{Rej}}

\renewcommand{\b}{\mathbf b}
\renewcommand{\a}{\mathbf a}
\renewcommand{\r}{\mathbf r}
\renewcommand{\c}{\mathtt{cheat}}
\newcommand{\z}{\mathbf z}
\newcommand{\x}{\mathbf x}
\newcommand{\rate}{r}

\newcommand{\Q}{\mathsf Q}

\newcommand{\X}{\mathsf{X}}
\newcommand{\Y}{\mathsf{Y}}
\newcommand{\Z}{\mathsf{Z}}
\newcommand{\E}{\mathsf{E}}
\newcommand{\Pauli}{\mathcal{P}}
\renewcommand{\P}{\mathsf{P}}
\newcommand{\Clifford}{\mathcal{C}}
\newcommand{\Enc}[2]{\X^{#1}\Z^{#2}}
\newcommand{\Dec}[2]{\Z^{#2}\X^{#1}}

\newcommand{\naturals}{\mathbb N}

\newcommand{\A}{\mathtt{A}}
\newcommand{\Alabel}{\A}
\newcommand{\Tlabel}{\mathtt{T}}
\newcommand{\Xlabel}{\mathtt{X}}
\newcommand{\Ylabel}{\mathtt{Y}}
\newcommand{\Zlabel}{\mathtt{Z}}
\newcommand{\Plabel}{\mathtt{P}}
\newcommand{\C}{\mathsf{C}}

\newcommand{\Apattern}{\A}

\newcommand{\class}{\mathfrak C}
\newcommand{\classA}{\mathcal A}
\newcommand{\BQP}{\textsf{BQP}}

\newcommand{\bqp}{c}
\newcommand{\F}{\mathsf{F}}
\newcommand{\D}{\mathsf{D}}
\newcommand{\G}{\mathsf{G}}

\renewcommand{\H}{\mathsf{H}}

\newcommand{\CNOT}{\mathsf{CNOT}}
\newcommand{\SWAP}{\mathsf{SWAP}}
\newcommand{\U}{\mathsf{U}}
\newcommand{\T}{\mathsf{T}}
\newcommand{\CR}{\mathsf{CR}}
\newcommand{\stab}{\hat{\mathcal{S}}}

\newcommand{\id}{\mathds{1}}

\newcommand{\QOTP}{\textsf{Q-OTP}}
\newcommand{\Resource}{\mathcal R}

\newcommand{\Protocol}{\pi}
\newcommand{\filter}{\vdash}
\newcommand{\qchannel}{\mathcal Q}

\newcommand{\perm}{\sigma}

\usepackage{thmtools}
\theoremstyle{plain}

\newtheorem{definition}{Definition}
\newtheorem{theorem}{Theorem}

\newtheorem{lemma}{Lemma}

\newtheoremstyle{upright}
  {10pt}{10pt}
  {\normalfont}
  {}
  {\bfseries}
  {}
  {0.5em}
  {
    \rule{\linewidth}{0.4pt}\\%
    \thmname{#1}\thmnumber{ #2}\thmnote{ (#3)}\\%
    \rule[1.3ex]{\linewidth}{0.4pt}
  }
\theoremstyle{upright}
\newtheorem{protocol}{Protocol}
\newtheorem{resource}{Resource}
\newtheorem{simulator}{Simulator}

\let\oldprotocol\protocol
\let\endoldprotocol\endprotocol

\let\oldresource\resource
\let\endoldresource\endresource

\let\oldsimulator\simulator
\let\endoldsimulator\endsimulator

\renewenvironment{protocol}
  {\oldprotocol}
  {\par\hrule\endoldprotocol}

\renewenvironment{resource}
  {\oldresource}
  {\par\hrule\endoldresource}

\renewenvironment{simulator}
  {\oldsimulator}
  {\par\hrule\endoldsimulator}

\usepackage [
n,
advantage,
operators,
sets,
adversary,
landau,
probability,
notions,
logic,
ff,
mm,
primitives,
events,
complexity,
asymptotics,
keys
] {cryptocode}

\newcommand{\cmp}[1]{\complclass{#1}}
\renewcommand{\poly}{\mathrm{poly}}

\usepackage{algpseudocodex}
\algrenewcommand\algorithmicprocedure{$\triangleright$}

\usepackage{float}
\floatstyle{ruled}
\newfloat{resourcee}{htb!}{idf}
\floatname{resourcee}{Resource}

\newfloat{simulatorr}{htb!}{idf}
\floatname{simulatorr}{Simulator}

\newfloat{protocoll}{htb!}{idf}
\floatname{protocoll}{Protocol}

\newfloat{informalprotocol}{htb!}{idf}
\floatname{informalprotocol}{Informal Protocol}

\newfloat{informalresource}{htb!}{idf}
\floatname{informalresource}{Informal Resource}
\newfloat{reduction}{htb!}{idf}
\floatname{reduction}{Reduction}

\newfloat{routine}{htb!}{idf}
\floatname{routine}{Routine}

\newfloat{algo}{htb!}{idf}

\title{Composable Verification in the Circuit-Model via Magic-Blindness}

\author[]{Sami Abdul Sater, Harold Ollivier} \affil[]{DI-ENS, Ecole Normale Supérieure (Université PSL, CNRS, INRIA), \\
  45 rue d’Ulm, Paris 75005, France.}

\begin{document}
\maketitle

\begin{abstract}
  As quantum computing machines move towards the utility regime, it is essential that users are able to verify their delegated quantum computations with security guarantees that are (i) robust to noise (ii) composable with other secure protocols and (iii) exponentially stronger as the number of resources dedicated to security increases. Previous works that achieve these guarantees are expressed in the Measurement-Based Quantum Computation (MBQC) model and benefit from a modular framework of verification protocols. This leaves architectures based on the circuit-model---in particular those using the Magic State Injection (MSI)---with fewer options to verify their computations or with the need to compile their circuits in MBQC which leads to overheads.

    This paper introduces a family of noise robust, composable and efficient verification protocols for Clifford + MSI circuits that are secure against arbitrary malicious behavior. This family contains the verification protocol of Broadbent (2018, ToC), extends its security guarantees while also bridging the modularity gap between protocols for MBQC and those for the circuit-model, and reducing quantum communication costs. As a result, it opens the prospect of rapid implementation tailored to near-term quantum devices.

    Our technique is based on a refined notion of blindness, called magic-blindness, which hides only the injected magic states---the sole source of non-Clifford computational power. This enables verification by randomly interleaving computation rounds with classically simulable, magic-free test rounds, leading to a trap-based framework for circuit verification. As a result, circuit-based quantum verification attains the same level of security and robustness previously known only in MBQC. It also reduces the quantum communication cost as transmitted qubits are required only at the locations of state injection.

  \end{abstract}

\section{Introduction}
\subsection{Context and motivation}
The advent of quantum computing as a service and the nearing of the quantum-utility regime---where quantum computers will be able to tackle useful problems beyond the reach of classical computers---call for more trust to be built into the ecosystem~\cite{BKBH25grand}. In essence, this challenge calls for ensuring that the results of quantum computations cannot be spoofed easily while recognizing that checking the results of quantum computations cannot be done by recomputing them using classical means. Indeed, even reproducing the computation on a different quantum computer is not of much help, as current machines require tailoring the computation to their hardware specifications---e.g., with noise-model-aware error mitigation techniques---in a way that would prevent meaningful comparisons.

The question of verification is formalized through Quantum Prover Interactive Proof (QPIP) systems where a \cmp{BPP}---for Bounded Probabilistic Polynomial time---verifier checks the results of a \cmp{BQP} computation provided by an all-powerful quantum prover.
We find this situation in particular in the context of Delegated Quantum Computing (DQC), where a Client delegates a quantum computation\footnote{For $\BQP$ computations, the input to the computation is a classical bitstring, and the output is a decision bit.} to a Server. The Client plays the role of the verifier and the Server is the prover.
Several kinds of protocols achieving verification of quantum computations exist \cite{ABE08interactive,ABEM17interactive,FK17unconditionally,B18how,M18classical} and have essentially settled the question from a theoretical perspective. Yet, from a practical one, the situation is less satisfactory. All these protocols have overheads that impose an untenable trade-off between computational power and security as machines---especially those in the pre-utility regime---are strongly resource-constrained. Yet, verification of quantum computation is a worthwhile task in the pre-utility regime. The reason is that it provides a (the only?) rigorous way to ascertain claims made on the computational capabilities of quantum devices.

This has motivated the quest for efficient strategies to verify quantum computations. This article contributes to this broad line of research by introducing a large class of verification protocols well adapted to circuit-model architectures with magic state injection, which complements and modularizes the few strategies available in the circuit-model~\cite{B18how,BN25noise}.

\subsection{Related Works}
\label{subsection:related works}
Verification was recognized early on as a major tool for establishing trust in quantum computation~\cite{A07scott} as well as to tackle the more fundamental question of falsifiability of quantum mechanics~\cite{V07conference,AV12is}. This set off several lines of research to uncover protocols achieving verification of problems in \cmp{BQP}, with two different broad classes of protocols. The first one grants the Client the additional ability to operate on a single-qubit register as well as to use a quantum communication channel. It offers statistical security \cite{ABE08interactive,ABEM17interactive,FK17unconditionally,B18how,LMKO21verifying,KKLM22unifying}. In contrast, the work of \cite{M18classical} and subsequent works \cite{ACGH20non,BKLM22succinct} do not require a Client with any quantum ability, but come at the expense of security based on computational hardness assumptions.

Among the first ones, the protocols introduced in~\cite{ABE08interactive,ABEM17interactive} rely on quantum authentication schemes. The one of~\cite{FK17unconditionally} exploits specifics of the Measurement-Based Quantum Computing (MBQC) model to insert traps in the computation being delegated. While these two achieve statistical security, their practicality is however limited. This is because they incur a large overhead needed to guarantee their security. In spite of their differences, they rely on the generic idea that verification is achieved by making parts of the computation inaccessible to a malicious Server by embedding it into a much larger Hilbert space---hence the overhead. Then, the verification proceeds by running statistical checks on parts of the larger Hilbert space not used for computation, allowing one to probe the honesty of the Server.

The only protocols achieving a zero space-overhead are based on the Test/Computation paradigm, represented by \cite{B18how} and \cite{LMKO21verifying}.
It consists of randomly interleaving computation rounds and test rounds, where the latter are easy---classically simulable---quantum computations yielding deterministic outcomes that can be used to detect a cheating Server. In this paradigm, rounds have to be delegated indistinguishably, which is made possible by hiding information in qubits prepared and sent by the Client. In \cite{B18how}, this is made possible by compiling the original circuit into three possible instances (called "runs"). One of them has the net effect of performing the target computation, and the two others apply the identity on specific configurations of input states: these two are thus used as tests. The compilation is derived from Quantum Computing on Encrypted Data (QCED) \cite{B15delegating}, ensuring that the Server computes on encrypted data and thus that the three types of runs are statistically indistinguishable from one another.
In \cite{LMKO21verifying}, test runs are built by using the concept of traps introduced by the earlier \cite{FK17unconditionally}, thus again exploiting specifics of MBQC. They derive $k$ types of test runs where $k$ is the chromatic number of the underlying graph of the MBQC computation. The initial computation can be compiled on a bipartite graph state (like the brickwork state) with $k=2$, which echoes the two types of test runs of \cite{B18how}. All the rounds are delegated blindly using the Universal Blind Quantum Computing (UBQC) protocol of \cite{BFK09universal}, again ensuring that the types of runs are indistinguishable from one another. The two main advantages of \cite{LMKO21verifying} over \cite{B18how} are
composability and noise-robustness. The latter is intuitive: the protocol of \cite{LMKO21verifying} can tolerate circuit-level, server-side, noise under a threshold. The recent work of \cite{BN25noise} showed that circuit-model verification can also, like MBQC, be made noise-robust, with an extension of the initial \cite{B18how}. The former—composable security—is key for applications, as developed hereunder.

Much of the recent progress toward hardware-optimized verification has relied on the MBQC approach together with the composable-security analysis of \cite{FK17unconditionally} within the Abstract Cryptography (AC) framework~\cite{MR11abstract}. Establishing composable security for the underlying UBQC delegation protocol required blindness to be formulated as a precise cryptographic resource. Once phrased at that level of abstraction, verification could be decoupled from the original construction, revealing that trappification is not unique and enabling the modularization of verification protocols in \cite{KKLM22unifying}. This, in turn, allowed different components of MBQC-based verification schemes to be optimized independently and adapted to hardware constraints, leading to several concrete protocols~\cite{KLMO24verification,GLMO25composably,YKO25verifiable} and proof-of-concept experiments~\cite{DNMN23verifiable,GLMM24chip,DIVF24design, BWMS25designing}. This structural flexibility stems from the strength of the UBQC blindness notion: when executing a delegated computation, the Server only learns the underlying graph structure of the MBQC computation, and the order of measurements. Such strong blindness, however, comes at a cost, as the quantum communication scales with the size of the computation, requiring the Client to prepare a qubit for each node of the graph. In contrast, in the circuit-model no comparable standalone blindness notion has been identified. What plays its role in protocols such as \cite{B18how} arises from a compilation procedure rooted in QCED \cite{B15delegating}, with additional adaptations for the non-Clifford $\T$-gate and subsequently the Hadamard gate and Phase gate $\P$ ($\Z$-rotation of $\pi/2$ angle). The resulting communication cost scales essentially linearly with the number of single-qubit gates. However, here blindness remains tied to a specific construction rather than supporting a broader class of interchangeable verification protocols.

While the MBQC protocols achieve zero space-overhead when compared to the unprotected computation, their reliance on the MBQC model makes them less adapted to architectures that are close to the circuit-model for quantum computation---in particular those based on the Clifford + Magic-State Injection (MSI) model. This is because compiling a circuit into a measurement pattern already introduces space overhead. In fact, protocols~\cite{B18how,BN25noise} are the only ones known to optimize space-overhead in that model and to be good candidates for pre-utility implementations. But even so, they lack useful characteristics such as composability and, most importantly, the modularity required for further optimization.

There is thus a stark contrast between MBQC and circuit-model verification. On the one hand, MBQC approaches are proven to be composable and robust to noise. Also, they are based on a previously clearly identified notion of blindness (UBQC) that has helped shape a family of verification protocols (a modular framework) rather than exhibiting a single isolated protocol.
On the other hand, no analogous abstraction has been identified in the circuit-model where since \cite{B18how}, only noise-robustness was provided with \cite{BN25noise}.

\subsection{Contributions}

The contrast between MBQC and circuit-model verification raises several questions:
\begin{itemize}
    \item Is there a blindness concept underlying circuit-model verification like UBQC does for MBQC?
    \item If so, can this be leveraged to build a verification protocol in the circuit-model with composable security and noise-robustness?
    \item Can this be further leveraged to derive a family of protocols? In other words, does there exist a modular framework for verification in the circuit-model?

\end{itemize}

We answer all three of the above in the affirmative, as summarized in Table~\ref{tab: comparison table}. This work establishes the missing cryptographic and conceptual steps for verification of delegated quantum computations in the circuit-model and answers the above questions in a constructive manner. Rather than proposing a single new protocol, we identify and formalize the primitives that make composable and noise-robust trap-based verification possible in that model, with a modular trap design. Central among these is a new blindness notion for Clifford+MSI circuits, which we call Magic-Blindness. Our contributions are thus threefold.

\begin{table}
\centering
\resizebox{\linewidth}{!}{
\begin{tabular}{l|cc|ccc}

Computation model
  & \multicolumn{2}{c|}{MBQC} & \multicolumn{3}{c}{Circuit-model} \\
 Ref.
 & \cite{LMKO21verifying}
 & \cite{KKLM22unifying}
 & \cite{B18how}
 & \cite{BN25noise}
 & This work \\
\hline

Number of rounds
& $O(\log(1/\epsilon))$
& $O(\log(1/\epsilon))$
& $O(\poly(1/\epsilon))$
& $O(\log(1/\epsilon))$
& $O(\log(1/\epsilon))$ \\

Blindness
& UBQC
& UBQC
& —
& —
& Magic-Blindness \\

Qubits sent per round
& $O(\poly(|C|))$
& $O(\poly(|C|))$
& $O(|C|)$
& $O(|C|)$
& $O(n + t)$ \\

Robustness
& Yes
& Yes
& No
& Yes
& Yes \\

Composability
& Yes
& Yes
& No
& No
& Yes \\

Modularity
& No
& Yes
& No
& No
& Yes \\

\end{tabular}
}

\caption{\textbf{Comparison of verification protocols in the Test/Computation paradigm.} The table reports the round complexity required to achieve security level $\epsilon$, the quantum communication per round from the client, and whether the protocols provide robustness, composability, and modularity (i.e., whether they consist of a single fixed protocol or a family of protocols). Here, $|C|$ denotes the circuit size for MBQC, while for the circuit-model it is $n+t+6h+2p$ where $n$ is the number of qubits in the input to the circuit, $t$ is the count of $\T$-gates, $h$ of $\H$, $p$ of $\P$ ($\pi/2$ rotation).
}
\label{tab: comparison table}
\end{table}

\paragraph{Composable magic-blind delegation in the circuit-model.} We introduce a composable delegation protocol with a blindness property that is sufficient for the verification of Clifford+MSI circuits. More precisely, working in the AC-security framework, we isolate an ideal resource that allows a Client to delegate a computation to a Server while hiding whether it is the original computation or one in which the magic states have been replaced by stabilizer states, thus yielding a classically simulable computation. Quite naturally, we call such a resource \emph{Magic-Blind Delegated Quantum Computation} (Resource~\ref{resource:mblind_DQC}).

Crucially, this abstraction implies that only the input qubits and the injected resource states need to be hidden: the public Clifford structure can be executed directly by the Server. As a result, the Client’s quantum communication scales with $n+t$, where $n$ is the input size and $t$ is the number of state injections (i.e., the $\T$-count), rather than with the size of the entire circuit as in UBQC that blinds the entire computation. Here, magic-blindness leaks the Clifford structure of the computation, but this does not prevent ensuring blindness between test and computation runs.

It is constructed from intermediate functionalities using composable security: one that allows a Client to delegate a Clifford layer and a state injection while hiding which state is injected (Resource~\ref{resource:blind-gate}), and one that allows delegation of a Clifford layer followed by Pauli measurements while hiding the underlying quantum state (Resource~\ref{resource:blind-meas}). Those constructions are described in Section~\ref{section:blindness}. While clearly analogous to the gadget of~\cite{B18how}, these resources are more atomic and allow greater freedom when generating test runs.

\setcounter{theorem}{4}
\begin{theorem}[Security of Magic-Blind DQC, informal]
    For any computation specified by a sequence of Clifford and state injection layers, the Magic-Blind DQC Protocol~\ref{informal:protocol:mblind_DQC} allows a Client to delegate the computation while revealing to the Server only the Clifford structure (gates and location) and not the injected states.
    The Client’s quantum communication consists of $n+t$ qubits, where $n$ is the input size and $t$ is the number of injected states.
    \end{theorem}

\paragraph{Efficient, noise-robust, and composable verification in the circuit-model.}
Building on the above, we provide a composable verification protocol in the circuit-model.
We first prove that the protocol is correct, \textit{i.e.}, that it performs the intended verified DQC functionality when both parties behave honestly and operate perfectly. We also prove its robustness to circuit-level noise, that is, that the ideal functionality can still be achieved when one of the honest parties (the Server) is affected by circuit-level noise. Finally, we prove its security against arbitrary behavior when the Server stops being honest. Overall, the protocol is efficient because it approximates an ideal functionality up to a negligible construction error.

The protocol presented in Section~\ref{section:verification} works by using the Magic-Blind DQC to interleave test and computation rounds blindly. Each test is designed as a computation with the same Clifford structure, but with only stabilizer states injected, making the measurement outcome of a chosen qubit deterministic and thus efficiently checkable. It is then possible to show that the security error decreases exponentially with the number of rounds, while the quantum communication per round scales linearly in the $\T$-count.

\setcounter{theorem}{5}
\begin{theorem}[Security of Verified DQC, informal]
    For any computation described in the Clifford+MSI model,
    the Verified Delegated Quantum Computation Protocol~\ref{protocol:verification}
    is \emph{composably secure}  and
    ensures that, with probability \emph{exponentially close to one} a dishonest Server is caught and an \emph{honest but noisy Server is accepted}.
  \end{theorem}

\paragraph{A trap-based verification framework in the circuit-model.}
Taking a step back, we show that the same construction offers more than just a protocol. Indeed, magic-blind delegation induces a natural class of indistinguishable computations that comprises the circuits that share the same Clifford skeleton but differ in their injected states. The trap design used in Section~\ref{subsection:framework single qubit} is a particular choice, but in Section~\ref{subsection:framework generalized} we generalize it and show it can be made modular. As a result, we provide a versatile \emph{trap-based verification framework in the circuit-model}. In particular, Broadbent’s verification protocol~\cite{B18how} (and the subsequent \cite{BN25noise}) can be cast into this framework. It also makes explicit the structural origin of circuit-model traps and enables the systematic construction of families of verification protocols with identical security guarantees, in direct analogy with the trappification framework for MBQC~\cite{KKLM22unifying}.

\subsection{Technical overview}
In this section we briefly summarize the technical ingredients that allow us to reach the main goal of the work: composable verification of quantum computations in the circuit-model. We start by describing the Test/Computation paradigm---in which our work fits---while explaining the ingredients that make it work: essentially the ability to execute test computations that are indistinguishable from the target computation. For this to hold, the rounds must be delegated blindly, and we formalize the blindness requirement in the Clifford+MSI model. Finally, we leverage this to build a verification protocol based on traps and show that the approach is in fact modular.

The whole focus of the paper is to provide constructions for verified delegated quantum computing. As we aim for composability, this requires us to define the ideal behavior of verified DQC: this is captured by the informal Verified DQC Resource~\ref{informalresource: VDQC}. It captures two modes for the Server: to cheat or not to cheat. It is defined more rigorously later in the paper, in Section~\ref{section:verification}.
\begin{informalresource}
    \caption{Verified Delegated Quantum Computation}
    \label{informalresource: VDQC}
    \begin{algorithmic}[0]
        \State \textbf{Client's Inputs:} quantum computation on a classical input.
        \State \textbf{Server's Inputs:} a cheating bit $c\in\bin$, with $0$ indicating an honest behavior.
        \Procedure{Computation by the Resource}{}
            \If{$c=0$}
            \State Resource performs the computation, measures the output qubits and sends the classical outcomes to the Client.
            \Else
            \State Resource sends an $\mathsf{Abort}$ message to the Client.
            \EndIf
        \EndProcedure
    \end{algorithmic}
\end{informalresource}
The aim of the constructions we present along the way is to be able, eventually, to approximate this ideal behavior up to a negligible distance in diamond norm through a protocol. In this technical overview, we lay out the typical rationale of the Verification protocols, exposing their requirements, and present the ideal resources and protocols we introduce to meet them.

\paragraph{The Test/Computation paradigm.}
The approach for verification that we propose follows the same paradigm as \cite{B18how,LMKO21verifying}, that we describe here. The Test/Computation paradigm uses repetition as the only overhead to verify a delegated quantum computation. The Client has an initial quantum computation $C$, described in a \emph{computation model}  that she wants to delegate to a quantum Server. She uses a \emph{blind delegation protocol} that embeds her initial computation in a \emph{class of indistinguishable quantum computations}. This prevents the Server from knowing if the instance requested by the Client is the initial computation or any other one in the class.

Among the class, there are instances that are classically simulable and yield deterministic measurement outcomes on some output qubits. These are called \emph{traps}: configurations of the input state to make a chosen output qubit's measurement outcome deterministic.
Then, for a chosen number of rounds, the Client delegates either the initial computation or one of the instances with traps, and rejects the overall outcome if the number of triggered traps exceeds a certain threshold. This is the essence of the \emph{Test/Computation paradigm} and it is depicted in Figure~\ref{fig:verif-canvas}.

\begin{figure}[h]
  \centering
  \begin{subfigure}[b]{0.5\linewidth}
    \centering
    \includegraphics[width=\linewidth]{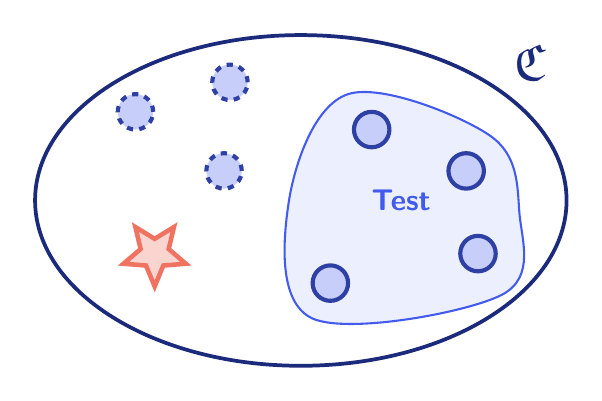}
    \caption{Under a \textbf{blindness} protocol, the target computation is embedded in a class $\class$ containing the initial computation (red star) and classically simulable instances (blue circles), all indistinguishable when delegated. Those yielding deterministic measurement outcomes are named \emph{traps}.}
    \label{fig:verif-canvas:blindness}
  \end{subfigure}
  \hfill
  \begin{subfigure}[b]{0.45\linewidth}
    \centering
    \includegraphics[width=\linewidth]{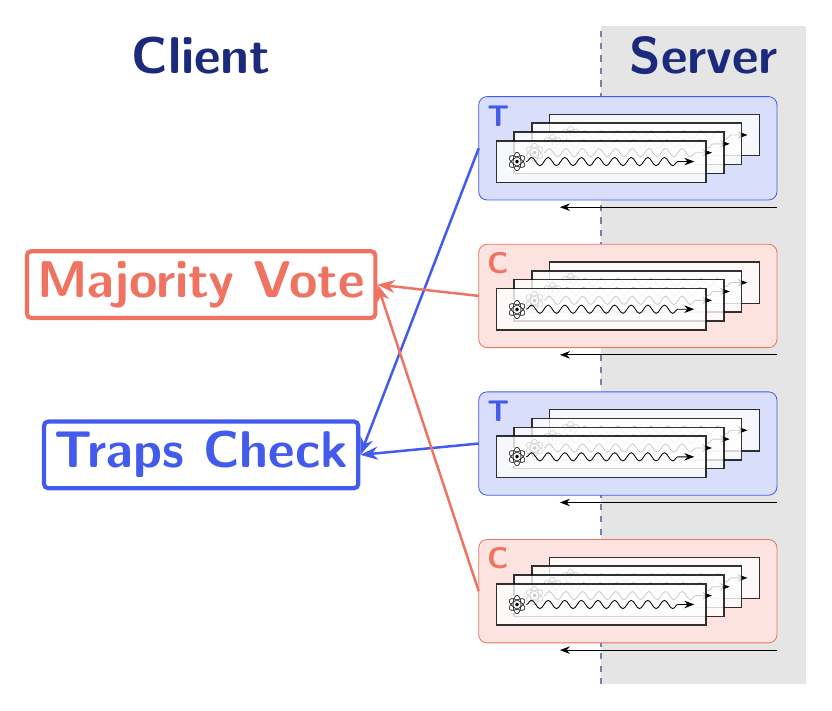}
    \caption{In the \textbf{Verification} protocol, test (blue boxes) and computation rounds (red boxes) are blindly delegated to the Server. Trap outcomes are checked by the Client to yield the acceptance/rejection decision; computation outcomes are aggregated via a majority vote.}
    \label{fig:verif-canvas:verification}
  \end{subfigure}
  \caption{The Test/Computation paradigm.}
  \label{fig:verif-canvas}
\end{figure}

Our work essentially introduces and formalizes constructions for the Test/Computation paradigm in the Clifford+MSI model. This representation captures both standard circuit-model architectures and the measurement-based implementations where qubits measured in a non-Pauli basis are instead prepared in a magic state and measured in an adaptive Pauli basis \cite{SDKO07direct}.

\subsubsection{Blindness (Section~\ref{section:blindness}).}
The Test/Computation paradigm does not require the computation to be entirely blinded when delegated: it only requires enough blindness to make the target computation indistinguishable from classically simulable computations that can then be used as tests.
The central idea is to exploit the fact that, in this model, the only potential source of computational non-classicality is carried by the injected qubits when they are in a magic state, while the surrounding Clifford structure is efficiently simulable and can be made public.
\paragraph{Blinding the Magic only.}
By hiding which single-qubit states are injected, the Client can make different computations indistinguishable to the Server: genuine quantum computations using magic states, and classically simulable test computations obtained by injecting only stabilizer states. This is formalized through a blindness concept that we call \emph{magic-blindness}, captured by Resource~\ref{informalresource:mblind_DQC}, that we present here informally.
\begin{informalresource}
    \caption{Magic-Blind DQC}
    \label{informalresource:mblind_DQC}
    \begin{algorithmic}[0]
        \State \textbf{Client's Inputs:} computation on $n$ qubits and $t$ ancillas.
        \State \textbf{Server's Inputs:} cheating bit $c\in\bin$, set to $0$ if honest.
        \State \textbf{Public Information:} Clifford parts of the computation.
        \Procedure{Computation by the Resource}{}
        \If{$c=0$ (honest behavior)}
        \State Resource performs the intended computation and measures the qubits
        \State Resource returns the classical bit-string of measurement outcomes to the Client.
        \Else{(malicious behavior)}
        \State Server provides a state and a CPTP map to perform, based on the public information only.
        \State Resource applies the CPTP map and returns the outcomes to the Client.
        \EndIf
        \EndProcedure
    \end{algorithmic}
\end{informalresource}
The Magic-Blind Delegated Quantum Computation Resource allows a Client to delegate a computation on $n$ qubits and $t$ ancillas, while hiding whether the injected ancillas are magic states, since only the content of the Clifford gates of the computation is leaked. By definition of this resource, any Server cheating behavior does not allow it to infer this information.
\begin{center}
    \textit{
        Magic-blindness is the ingredient that creates the indistinguishability required to embed efficiently simulable computations in the Clifford+MSI model.
    }
\end{center}

\paragraph{Composing Blind State Injection and Blind Measurements.}
The core idea of the protocol we present to realize this Resource is for the Client to apply a Pauli encryption to the states before sending them, thereby perfectly blinding their content through a quantum one-time pad, and to keep track of the evolution of the Pauli frame to later undo this encryption by decoding the measurement outcomes returned by the Server. Since all operations that the Server is supposed to apply are Clifford, this tracking is efficient. Intuitively, if the Server is not honest, it might return a bad result---we deal with this in the verification part---but it never learns the content of the state, which is what matters for our purposes.
\setcounter{informalprotocol}{3}
\begin{informalprotocol}
    \caption{Magic-Blind Delegated Quantum Computation}
    \label{informal:protocol:mblind_DQC}
    \begin{algorithmic}[0]
        \State Client sends an encrypted input state to the Server.
        \For{$i \le t$}
        \State Server applies public Clifford layer $\C_i$, and
        \State Client and Server perform a blind state-injection gadget.
        \Comment{Protocol~\ref{protocol:blind-gate}}
        \EndFor
        \State Server applies final Clifford $\C_{t+1}$ and measures all qubits, Client decodes. \Comment{Protocol~\ref{protocol:blind-meas}}
    \end{algorithmic}
\end{informalprotocol}

This protocol is built by composing sub-protocols, leveraging composability for its security. The first one makes the Server do a Clifford layer followed by a state injection, where all the qubits are blinded. It is used for as many layers as the computation has. Then a final Clifford layer followed by measurements, where all the qubits are again blinded. Those ideal behaviors are respectively implemented by a Blind State-Injection Protocol~\ref{protocol:blind-gate} and Blind Measurements Protocol~\ref{protocol:blind-meas}, all detailed in Section~\ref{section:blindness}.

\paragraph{Consequence: Clifford and non-Clifford delegated indistinguishably}
The main consequence of using the above protocol is what we aimed for: stabilizer states can be injected instead of magic states, and the computation becomes entirely Clifford on $n+t$ qubits, instead of the initial non-Clifford computation $C$ on $n$ qubits (tracing out the $t$ ancillas after they are injected). Throughout the work, when no magic states are injected, we refer to the resulting $n+t$-qubit Clifford as $\G$: it is the circuit obtained by interleaving the Clifford parts of the initial computation with $\CNOT$ on the ancillas to perform state injection. The same cannot be said when magic states are injected since conditioned Clifford operations need to be added to ensure that the state is being transformed correctly. Indeed, Magic-State Injection requires $\CNOT$, measurement, and conditioned Clifford to implement a $\T$-gate. When injecting a stabilizer state instead, the conditioned Clifford can be dropped, and the measurements can thus be delayed: the evolution of the entire system on $n+t$ qubits can be simulated classically and is described by the $n+t$-qubit Clifford circuit $\G$ (see Equation~\ref{eq:blindness:G} and Figure~\ref{fig:blindness:G}).

Ultimately, Magic-Blindness embeds the initial computation $C$ in a class $\class$ of computations that are delegated indistinguishably under Protocol~\ref{informal:protocol:mblind_DQC}: they all consist of the same Clifford parts, and differ by the states that are injected. The class $\class$ contains the initial computation $C$ as well as other instances that perform the Clifford $\G$ on stabilizer inputs.

\paragraph{Consequence on malicious behavior.}
Since Protocol~\ref{informal:protocol:mblind_DQC} is proven composably secure, \emph{i.e.}, it constructs Resource~\ref{informalresource:mblind_DQC}, any adversarial deviation is captured by a single strategy at the level of the resource. Magic-blindness then guarantees that all computations in the class are indistinguishable to the Server. Therefore, a malicious Server cannot condition its attack on a particular instance: the same deviation acts on every computation in the class, a property that will be crucial in the verification analysis.

\begin{center}
    \textit{
        Because the Server cannot distinguish instances in the class $\class$, any malicious deviation must act uniformly across all instances, including the target computation and the efficiently simulable ones used for tests.
    }
\end{center}

\subsubsection{Verification (Section~\ref{section:verification}).}

\paragraph{Leveraging Magic-Blindness}
We leverage the Magic-Blind DQC Protocol above to realize the ideal Verification Resource with Protocol~\ref{informal:protocol:verification} that follows the Test/Computation paradigm presented in Figure~\ref{fig:verif-canvas}. It interleaves test and computation rounds blindly using Protocol~\ref{informal:protocol:mblind_DQC}, where test rounds are instances of $\class$ that are efficiently simulable and yield deterministic measurement outcomes—\emph{traps}—allowing the Client to catch a cheating Server. This protocol is designed for $\BQP$ computations, for which the input is a classical bitstring and the output is a decision bit.

\setcounter{informalprotocol}{4}
\begin{informalprotocol}
    \caption{Verified Delegated Quantum Computation}
\label{informal:protocol:verification}
\begin{algorithmic}[0]
  \State \textbf{Client inputs:} quantum computation $C$, input bitstring $\x$, parameters $d, s, w$.

  \State Client chooses a random partition of $d$ computation and $s$ test rounds.
  \Comment{Private
  set-up}
  \State Client delegates the rounds blindly using Informal Protocol~\ref{informal:protocol:mblind_DQC}.
  \Comment{Blind delegation}

  \State If more than $w$ test rounds had trap failures, abort; else go to next step. \Comment{Traps check}
  \State If output $z$ has more than $d/2$ occurrences, output $z$ ; else $z\oplus 1$. \Comment{Majority vote}

\end{algorithmic}
\end{informalprotocol}

\paragraph{Designing traps for the Clifford+MSI model.}
Traps must be designed to detect any Server deviation that might harm the computation. However, in Section~\ref{subsection:blind:Pauli dev} we show that as a consequence of using MB-DQC at each round, which contains a Pauli encryption of the entire register on which the Server operates (potentially maliciously), any deviation with respect to the instructed behavior can be reduced to a convex combination of Pauli deviations, because of a Pauli Twirl. This yields a classification of Pauli deviations into \emph{harmful} (flipping at least one measurement outcome) and \emph{harmless} (not flipping any measurement outcome) deviations.

We use this to introduce, in Section~\ref{subsection:framework single qubit},
a way to detect any harmful deviation while being insensitive to harmless ones.
It is based on the following intuition: because deviations are only Pauli, they can only deterministically flip deterministic Pauli measurements. Hence we introduce the concept of \emph{traps} for Clifford+MSI computation: computations in $\class$ yielding efficiently simulable and deterministic measurement outcomes that can be used to check the presence of (harmful) deviations.
Intuitively, because such computations are described by the Clifford evolution $\G$, a trap can be obtained by injecting stabilizer states according to the following logic: in order to make the measurement of qubit $i$ deterministic, the output state must be stabilized by $\Z_i$, thus the input state must be stabilized by $\G^\dagger \Z_i\G$. Preparing a $+1$-eigenstate of this operator and injecting it to the computation, qubit by qubit, and checking outcome of qubit $i$, is thus a way to detect if a Pauli deviation was present on that qubit. Having one trap per output qubit is the construction we propose in Section~\ref{subsection:framework single qubit}, building a protocol on top of it in Section~\ref{subsection:verification:protocol} and proving its security. Later, in Section~\ref{subsection:framework generalized}, we show that this approach can be generalized into a more modular approach, building traps for subsets of qubits. This freedom to design traps while still detecting any harmful deviation is thus at the core of a verification framework in the Clifford+MSI model (depicted on Figure~\ref{fig:canvas-framework}).
\begin{figure}[h]
  \centering
  \begin{subfigure}[c]{0.40\linewidth}
    \centering
    \includegraphics[width=\linewidth]{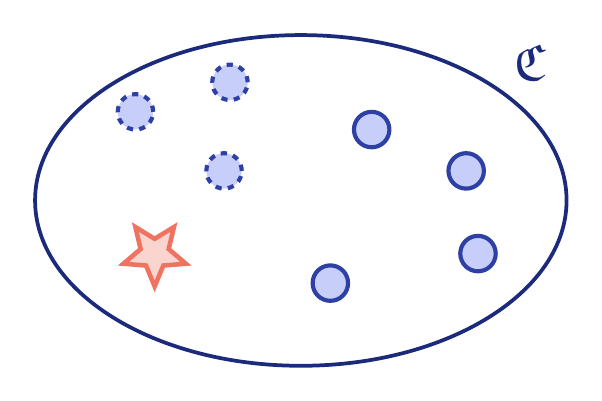}
    \caption{Blindness.}
    \label{fig:canvas-framework:blindness}
  \end{subfigure}
  \hfill
  \begin{subfigure}[c]{0.18\linewidth}
    \centering
    \includegraphics[width=\linewidth]{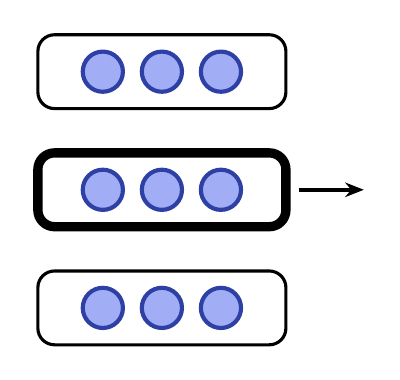}
    \caption{Trap design.}
    \label{fig:canvas-framework:traps}
  \end{subfigure}
  \hfill
  \begin{subfigure}[c]{0.38\linewidth}
    \centering
    \includegraphics[width=\linewidth]{tikz-figures/canvas/verification.pdf}
    \caption{Verification.}
    \label{fig:canvas-framework:verification}
  \end{subfigure}
  \caption{\textbf{The modular trap-based verification framework.}
    \emph{(\subref{fig:canvas-framework:blindness}) Blindness.}
    Magic-blindness embeds the target computation $C$ in a class $\class$ containing it (denoted by a red star) alongside multiple families of efficiently simulable computations (blue circles), some of them yielding deterministic measurement outcomes (solid line)—those are named \emph{traps}.
    \emph{(\subref{fig:canvas-framework:traps}) Trap design.}
    The Client selects one trap family (highlighted row) from the available options; any family that detects all Server deviations suffices.
    \emph{(\subref{fig:canvas-framework:verification}) Verification.}
    Using the chosen family, the protocol proceeds exactly as in Figure~\ref{fig:verif-canvas:verification}: test and computation rounds are blindly delegated, trap outcomes are checked, and computation outcomes are aggregated by majority vote.}
  \label{fig:canvas-framework}
\end{figure}

\paragraph{Main result: efficient, composable, noise-robust verification in the circuit-model.}

\begin{itemize}
    \item \textbf{Composability:} we prove the composable security of Protocol~\ref{informal:protocol:verification}, \emph{i.e.}, we show that it implements Resource~\ref{informalresource: VDQC} up to a construction error $\epsilon$ that corresponds to the distinguishing probability between the Protocol and the Resource for any unbounded adversary.
    \item \textbf{Exponential security:} we show that in the presence of a malicious Server, $\epsilon$ corresponds to the maximum probability to trigger less than $w$ test rounds out of $s$ while corrupting more than $d/2$ computation rounds, enough to alter the outcome of the majority vote. In Lemma~\ref{lemma:security error}, we show that this quantity is negligible in $d, s$, as long as $w$ is chosen within a specific range depending on the number of types of tests and the BQP error of the computation. It thus shows that the security error decreases exponentially with the number of rounds involved in the protocol.
    \item \textbf{Noise-robustness:} we show that when the Server is honest but operates on a device with a circuit-level noise that affects all rounds with probability less than $p_{err}$,
    traps are not triggered as long as $p_{err}< w/s$. As a consequence, the level of noise that is tolerated depends on where the test rounds threshold is set, which depends on the desired security.
\end{itemize}

\begin{center}
    \textit{
        Magic-blindness creates an indistinguishable computation class; this class enables trap constructions that detect all deviations, yielding composable verified delegation with exponential security and correctness robust to circuit-level noise.
    }
\end{center}

\paragraph{Organization of the paper.} The rest of the paper is organized as follows. We start with preliminaries in Section~\ref{section:prelims}. In Section~\ref{section:blindness}, we present the constructions and results relevant for magic-blindness, and in Section~\ref{section:verification}, we show how to leverage this to build a composably secure verification protocol following the Test/Computation paradigm. Finally, we conclude in Section~\ref{section:Discussion}.

\section{Preliminaries}
\label{section:prelims}
\setcounter{theorem}{0}
\subsection{Notations}
\begin{itemize}
    \item Bit strings $\a$ are in bold, with subscript $a_i$ denoting the $i$-th bit, and $|\a|$ denoting their length. Concatenation is denoted by $||$.
    \item Sans-serif font is for unitaries like $\C$, and the application to a quantum state is $\C[\rho] = \C\rho\C^\dagger$; compositions are written with $\circ$.
    \item Typewriter font is for state labels, like the $\ket\Xlabel = \ket+$ state because $\ket+$ is a $+1$-eigenstate of the Pauli operator $\X$. We also label the magic state for the $\T$-gate as $\ket\Tlabel=\T\ket+$ (see later).
    \item Quantum circuits are sequences of instructions, specified by composition of unitaries like $C=\C_2\circ\C_1$, and end with computational basis measurement of all the qubits.
    \item
    A unitary with an index means that it applies the identity on the unspecified indices. For instance, $\X_n$ is the Pauli $\X$ applied on the $n$-th qubit. An exception in the paper holds for $\C_i$ and $\F_i$, where the index $i$ refers to layers in the circuit, numbered from $1$ to $t$. This will be made explicit later.
    \item For a tensor product $\E$ of $n$ unitaries, let $[\E]_i$ denote the $i$-th unitary of the tensor product ($i\leq n$).
    \item We write the set of angle multiples of $\pi/2$ as $\Theta=\{0, \pi/2, \pi, 3\pi/2\}$.
    \item For a set $S$, writing $s\gets\$ S$ means that $s$ is sampled uniformly at random from $S$.
\end{itemize}

\subsection{Setting}
This work is in the Prepare-and-Send setting, where the Client has the ability to prepare single-qubit states in the set of states $\classA = \{\ket\Tlabel, \pm\ket\Xlabel, \pm\ket\Ylabel, \pm\ket\Zlabel\}$. These are the six single-qubit stabilizer states, and the magic state. The Client is assumed to be perfect. For the Server, three scenarios are examined (see Section~\ref{section:verification}): honest and perfect, honest but noisy, or arbitrarily malicious. In this work we consider quantum computations in the $\BQP$ class (see below), for which the input is a classical bitstring and the output is a decision bit.

\subsection{Quantum Computations}
\label{subsection:prelims:QC}

\paragraph{Bounded-error quantum computation ($\BQP$).}
A language $L$ is in $\BQP$ if there exists a uniform family of polynomial-size quantum circuits $\{C_n\}_{n \in \mathbb{N}}$ and a constant $\bqp < 1/2$ such that for every input $x \in \{0,1\}^n$, the circuit $C_n$ outputs the correct decision bit $z^\star$ for whether $x \in L$ with probability at least $1-\bqp$. This naturally captures general quantum computations with classical inputs and outputs.
Indeed, any $\BQP$ computation resulting in a classical bitstring can be reduced to a polynomial number of decision problems, meaning this definition encompasses the full range of problems solvable by a quantum computer.

In the delegated setting considered in this work, it is more convenient to separate the "input preparation" step from the actual computation machinery. We therefore write $C(\x)$ for a computation that consists of running circuit $C$ on input the input $\ket\x$ provided by the Client. Also, we assume all qubits are measured in the computational basis but only the first one $y=C(\ket\x)$ is kept. Under this convention, correctness requires that
\[
\Pr[y = z^\star] \geq 1-\bqp
\]
for every $\x \in \{0,1\}^n$.

\paragraph{The Pauli and Clifford groups.}
In this work, we write $\mathcal P_n$ for the Pauli group on $n$ qubits, defined as
\[
\mathcal P_n=\{\pm 1,\pm i\}\times \{\id,\X,\Y,\Z\}^{\otimes n}.
\]
By definition, the Clifford group $\mathcal{C}_n$ over $n$-qubits is a subgroup of the unitary group that normalizes the Pauli group: for any $\C\in\Clifford_n$ and any $\P\in\Pauli_n$, we have $\C\P\C^\dagger\in\Pauli_n$. The Clifford group is generated by the single-qubit gates $\H$ and $\P$ (where $\P$ is the $\pi/2$ $\Z$-rotation) together with the two-qubit gate $\CNOT$.

\paragraph{Clifford and $\T$ gates.}
Let $C$ be an $n$-qubit quantum circuit. It is a standard result that the gate set formed by the $n$-qubit Clifford group $\Clifford_n$ and the $\T$-gate, $\T = \mathrm{diag}(1, e^{i\pi/4})$, is universal for quantum computation: any polynomial-size quantum circuit can be efficiently approximated by a circuit composed of these gates \cite{NC00quantum}. Any approximation error arising from this compilation can be absorbed into the overall bounded-error parameter $\bqp$.

Since Clifford gates can permute qubits, we may assume without loss of generality that every $\T$-gate acts on the $n$-th qubit.
In this work, we thus consider the Clifford+T decomposition of any quantum circuit, meaning that for any  $\BQP$ circuit $C$, there exists $t\in\naturals$ such that $C$ can be decomposed as:
\begin{equation}
    C = \C_{t+1} \circ \T_n \circ \C_{t} \circ \cdots \circ \T_n \circ \C_{1},
\end{equation}
where  $\C_i \in \Clifford_n$ are Clifford gates. We refer to the sequence $\{\C_i\}_{i\leq t+1}$ as the \emph{Clifford structure} of the computation, and we say that $C$ is a $(n,t)$-Clifford+$\T$ computation.
This is depicted in Figure~\ref{fig:prelims:C-T}, for $n=3$ and $t=2$ for a readable example. Throughout the paper, we will stick with this example but the results hold for any $n,t$.
\begin{figure}[h]
    \centering
    \includegraphics[]{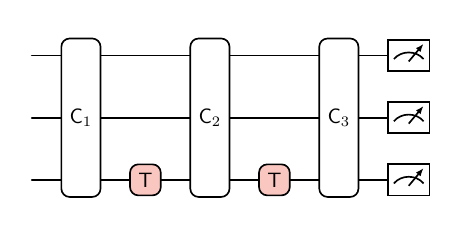}
    \caption{Representation of a $(n=3,t=2)$-Clifford+$\T$ quantum computation.}
    \label{fig:prelims:C-T}
\end{figure}

\paragraph{Measurements.} Without loss of generality, we assume in this paper that measurements are performed in the computational basis. The outcomes are the eigenvalues of the observable $(\id - \Z)/2$, namely $0$ and $1$.

\paragraph{Magic-State Injection \cite{BK04universal}.} In most quantum computing architectures, $\T$-gates are implemented by the use of a Clifford circuit, the injection of an ancillary qubit prepared in a $\T\ket+$ state, a Pauli measurement, and a $\Z$-rotation of angle $\pi/2$ and Pauli $\X$ correction conditioned on the measurement outcome.
Since this specific quantum state allows one to perform the $\T$-gate using Clifford operations only (and a measurement in a Pauli basis, standard in all implementations), it is often referred to in the literature as a \emph{Magic State}. By convention, we write $\ket\Tlabel=\T\ket+$. We also write $\F$ for the circuit consisting of a $\CNOT$ (where the target is the qubit on which the $\T$-gate is intended) followed by a $\SWAP$, as depicted in Figure~\ref{fig:prelims:MSI}.
\begin{figure}[h]
    \centering
    \includegraphics{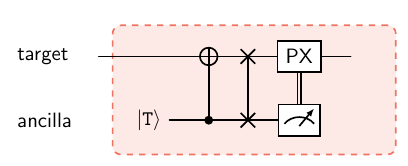}
    \caption{Magic-State Injection, to implement a $\T$-gate.}
    \label{fig:prelims:MSI}
\end{figure}
In the considered model of computation, the $\T$-gate only acts on the $n$-th qubit, and the $i$-th one uses an ancilla labeled as qubit $n+i$, so we formally define $\F_i=\SWAP_{n+i, n}\circ \CNOT_{n+i, n}$. The circuit of Figure~\ref{fig:prelims:C-T}, on $n$ qubits, becomes the circuit of Figure~\ref{fig:prelims:C-T-MSI}.
\begin{figure}[h]
    \centering
    \includegraphics[]{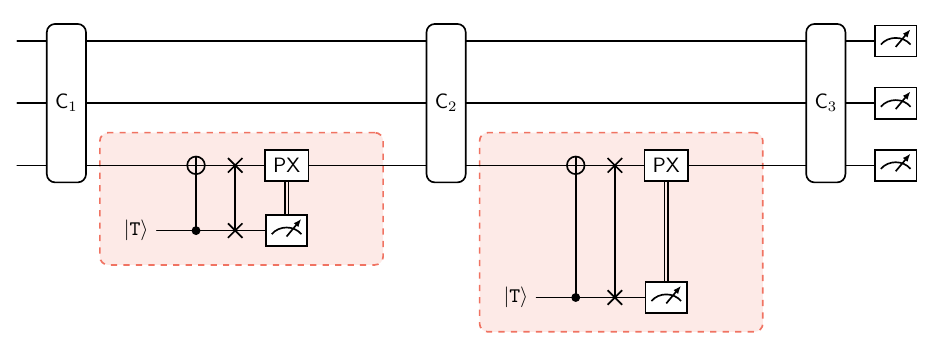}
    \caption{
    Representation of a $(3,2)$-Clifford+MSI quantum computation. It is the same as Figure \ref{fig:prelims:MSI} but the $\T$-gates have been implemented \textit{via} Magic-State Injection.}
    \label{fig:prelims:C-T-MSI}
\end{figure}
\paragraph{Stabilizer States.} We say that a state $\ket\psi$ is stabilized by operator $\P$ if $\ket\psi$ is a $+1$ eigenstate of $\P$, meaning $\P\ket\psi=\ket\psi$.
\emph{Stabilizer states} are stabilized by a maximal Abelian subgroup of $\mathcal{P}_n$. For instance, $\ket0$ is stabilized by $\Z$, while $\ket1$ is stabilized by $-\Z$. There are $6$ single-qubit stabilizer states.
By convention, for any $\P\in\Pauli_1$ we will write its $+1$-eigenstate as $\ket {\Plabel}$.

\paragraph{Classical simulation of Quantum Computations.}
By the Gottesman–Knill theorem~\cite{G98heisenberg}, any Clifford circuit $\C\in \Clifford_n$ acting on a $n$-qubit stabilizer input state, followed by Pauli measurements, can be efficiently simulated classically, including the full output distribution.

Furthermore, a Pauli measurement is deterministic if and only if the measured Pauli operator belongs to the stabilizer of the state. In particular, measuring qubit $i$ in the computational basis yields outcome $0$ with certainty iff the output state is stabilized by $\Z_i$, equivalently iff the input state is stabilized by $\C^\dagger \Z_i \C$.
\subsection{Pauli Encryption and Decryption}
\label{prelim:Clifford conjugation}

\paragraph{Bitstring representation of Pauli string up to a phase.}
For any $\P\in\Pauli_n$, we can associate two $n$-bit strings $\a, \r$ such that, up to a global phase: $\P\in \{\pm 1, \pm i\} \times \X^\a\Z^\r$, where $\X^\a = \bigotimes_{i\leq n}\X^{a_i}$.
 Then, the conjugation of $\P$ by a Clifford $\C$ is a mapping from $\a, \r$ to some other bitstrings $\a', \r'$. We write $(\a', \r')=\C(\a, \r)$ the bitstrings such that $\C (\Enc{\a}{\r})\C^\dagger = \Enc{\a'}{\r'}$.

\paragraph{The Quantum One-Time-Pad ($\QOTP$).}
In quantum cryptography, it is common to consider $\a, \r$ as secret keys and encrypt a quantum state $\rho$ by applying the Pauli $\Enc{\a}{\r}$, then decrypt with $\Dec{\a}{\r}$ (or with other keys if the operations between encryption and decryption have changed the keys). This is known as the Quantum One-Time Pad \cite{C05secure}, since from the point of view of a receiver who does not know $\a, \r$, the received state is a probabilistic mixture of the different $\Enc{\a}{\r}[\rho]$ for all the possible values of $\a, \r$. When these keys are sampled uniformly, this yields the following equality:
\begin{equation}
    \sum_{\a, \r\in\{0,1\}^n} \Enc{\a}{\r}[\rho] = \sum_{\P\in\Pauli_n} \P[\rho] = \id_n/2^n.
\end{equation}
This equality is easy to do by hand for $n=1$, then trivial to generalize to the multi-qubit case.

\paragraph{$\QOTP$ via EPR pairs and measurements.}
A $\QOTP$ followed by a quantum communication channel can be replaced by creating an EPR pair, using a quantum channel to send half of it, and applying a Bell measurement. This equivalence is shown using Shor-Preskill's reduction~\cite{SP00simple} with a delayed measurement; see Figure~\ref{fig:QOTP-EPR}.
Throughout the paper, we refer to this replacement as an EPR reduction of the initial protocol.
\begin{figure}[h]
    \centering
    \includegraphics[width=0.5\linewidth]{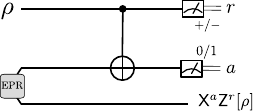}
    \caption{$\QOTP$ via EPR pairs and measurements. The "EPR" box denotes the creation of an EPR pair, the $+/-$ measurement is in the Hadamard basis, the other is in the computational basis. The outcomes $a, r$ are uniformly distributed.}
    \label{fig:QOTP-EPR}
\end{figure}

\paragraph{Pauli Twirl.}
Applying a random Pauli operator on a quantum state is a recurrent procedure in quantum cryptography or learning since it benefits from the following Twirling lemma.
\begin{lemma}[Pauli Twirl
    ~\cite{DCEL09exact}]
    \label{lemma:Pauli Twirl}
    Let $\rho_n$ be a $n$-qubit quantum state, and let $\E, \E'\in\Pauli_n$. Then:
    \begin{equation}
        \sum_{\Q\in\Pauli_n}
        \Q^\dagger
        \E
        \Q
        \;\rho\;
        \Q^\dagger
        \E^{\prime \dagger}
        \Q
        =0 \; \text{if}\; \E\neq \E'.
    \end{equation}
\end{lemma}
\subsection{Abstract Cryptography Framework}
The Abstract Cryptography (AC) security framework \cite{MR11abstract,M12constructive} used in this work follows the \emph{ideal-world/real-world paradigm}. A protocol is considered secure if it is a good approximation of an \emph{ideal resource} that is secure by design. Its main interest is that protocols that are AC-secure are inherently composable, in the sense that if an AC-secure protocol is used inside a larger protocol, the security of the former does not need to be reproved in the context of the latter. In other words it benefits from sequential and parallel composability.  Ref.~\cite{DFPR14composable} provides an introduction to the topic tailored to verification of quantum computation.

In this framework, the purpose of a secure protocol $\pi$ is, given a number of available resources $\mathcal{R}$, to construct a new resource -- written as $\pi \mathcal{R}$. This new resource can itself be reused in a future protocol. A resource~$\mathcal{R}$ is described as a sequence of CPTP maps with an internal state. It has \emph{input and output interfaces} describing which party may exchange states with it. It works by having each party send it a state (quantum or classical) at one of its input interfaces, applying the specified CPTP map after all input interfaces have been initialized, and then outputting the resulting state at its output interfaces in a specified order. An interface is said to be \emph{filtered} if it is only accessible by a dishonest player. The actions of an honest player $i$ in a given protocol are also represented as a sequence of efficient CPTP maps $\pi_i$ -- called the \emph{converter} of party~$i$ -- acting on their internal and communication registers. We focus here on the two-party Client-Server setting, in which case $\pi = (\pi_C, \pi_S)$. Note that all our protocols are built from quantum channels, due to the Client's requirements to prepare and send single-qubit states.

\paragraph{Indistinguishability of Resources.}

In order to define the security of a protocol, we need to give a pseudo-metric on the space of resources.  We consider for that purpose a special type of converter called a \emph{distinguisher}, whose aim is to discriminate between two resources $\mathcal{R}_1$ and $\mathcal{R}_2$, each having the same number of input and output interfaces.  It chooses the input, interacts with one of the resources according to its own (possibly adaptive) strategy, and guesses which resource it interacted with by outputting a single bit.  The Distinguisher has access to all of the resource's interfaces. Two resources are said to be indistinguishable if no distinguisher can guess correctly with good probability, captured by Definition~\ref{def:ind-res}.

\begin{definition}[Statistical Indistinguishability of Resources]
  \label{def:ind-res}
  Let $\epsilon > 0$, and let $\mathcal{R}_1$ and $\mathcal{R}_2$ be two resources with same input and output interfaces.  The resources are \emph{$\epsilon$-statistically-indistinguishable} if, for all unbounded distinguishers $\mathcal{D}$ the distinguishing probability $p_d$ is bounded by $\epsilon$, meaning if
  \begin{equation}
    \label{eq:dist}
    p_d := \Bigl\lvert\Pr[b = 1 \mid b \leftarrow \mathcal{D}\mathcal{R}_1] - \Pr[b = 1 \mid b \leftarrow \mathcal{D}\mathcal{R}_2]\Bigr\rvert \leq \epsilon.
  \end{equation}
  We then write $\mathcal{R}_1 \underset{\epsilon}{\approx} \mathcal{R}_2$.
\end{definition}

\paragraph{Construction of Resources.}
The construction of a given resource $\mathcal{R}$ by the application of protocol $\pi$ to resource $\mathcal{S}$ can then be expressed as the indistinguishability between resources $\mathcal{R}$ and $\pi \mathcal{S}$.  More specifically, this captures the correctness of the protocol.  The security is captured by the fact that the resources remain indistinguishable if we allow some parties to deviate in the sense that they are no longer forced to use the converters defined in the protocol but can use any other CPTP maps instead.  This is done by removing the converters for those parties in Equation~\ref{eq:dist} while keeping only $\pi_H = \prod_{i \in H} \pi_i$ where $H$ is the set of honest parties.
Since this work only considers an honest-Client / (potentially-)malicious-Server setting, there is only one honest party: the Client, with converter $\pi=\pi_C$, and one potentially malicious party.
The security is formalized as in Definition~\ref{def:ac-sec} in this case, and depicted on Figure~\ref{fig:prelims:AC}.

\begin{definition}[Construction of Resources]\label{def:ac-sec}
    Let $\epsilon > 0$. We say that a two-party protocol $\pi$ $\epsilon$-statistically-constructs resource $\mathcal{R}$ from resource $\mathcal{S}$ if:
    \begin{enumerate}
        \item It is correct: $\pi \mathcal{S} \underset{\epsilon}{\approx} \mathcal{R}\filter$, where $\filter$ prevents malicious behavior (from the potentially cheating party);
        \item It is secure against malicious party $P$: there exists a \emph{simulator} (converter) $\sigma$ such that $\pi\mathcal{S} \underset{\epsilon}{\approx} \mathcal{R} \sigma$.
        \label{item:prelim:AC-security}
    \end{enumerate}
\end{definition}

\begin{figure}[h]
    \centering
    \includegraphics[width=\linewidth]{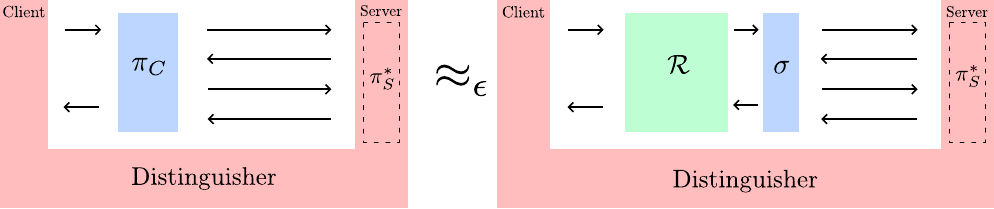}
    \caption{Security in Abstract Cryptography: indistinguishability between the Real World (left picture) and the Ideal World (right picture) up to distance $\epsilon$.
    In each scenario, the Distinguisher (red box) has two interfaces: Client and Server.
    Its inputs are fixed in both scenarios: the Distinguisher chooses the protocol's inputs from the Client interface, and chooses an arbitrary malicious behavior from the Server interface, captured by a CPTP map $\pi_S^*$, that respects the interface of the Client protocol $\pi_C$.
    The protocol is $\epsilon$-secure if the Distinguisher guesses correctly with probability at most $\epsilon$ if it was interacting with the Client Protocol $\pi_C$ (blue box, left picture) in the Real World or the Simulator $\sigma$ (blue box, right picture) plugged in the Resource $\mathcal R$ in the Ideal World.
    }
    \label{fig:prelims:AC}
\end{figure}

\paragraph{Composition of Resources.}
Using the definitions above, we can state the following general composition theorem~\cite{MR11abstract} that guarantees the additive accumulation of distinguishing advantage when composing two statistically secure protocols.
\begin{theorem}[General Composition of Resources~\cite{MR11abstract}]
  Let $\mathcal{R}$, $\mathcal{S}$ and $\mathcal{T}$ be resources, $\alpha, \beta$ and $\mathsf{id}$ be protocols, where protocol $\mathsf{id}$ does not modify the resource it is applied to.  Let $\circ$ and $|$ denote the sequential and parallel composition of protocols and resources, respectively.  Then the following implications hold:
  \begin{itemize}
  \item Sequential composability: if $\alpha \mathcal{R} \approx_{\epsilon_{\alpha}} \mathcal{S}$ and $\beta\mathcal{S} \approx_{\epsilon_{\beta}} \mathcal{T}$, then $\left(\beta \circ\alpha\right) \mathcal{R} \approx_{\epsilon_{\alpha}+\epsilon_{\beta}}\mathcal{T}$.
  \item Parallel composability (context insensitivity): if $\alpha \mathcal{R} \approx_{\epsilon_{\alpha}} \mathcal{S}$, then $\left(\alpha \mid \mathrm{id}\right)\left(\mathcal{R} \mid \mathcal{T}\right) \approx_{\epsilon_{\alpha}} \left(\mathcal{S} \mid \mathcal{T}\right)$.
  \end{itemize}
  Combining these two properties yields the composability of protocols.
\end{theorem}

\section{Blinding Magic in Delegated Quantum Computations}

\label{section:blindness}

In this section, we develop the constructions (Resources, Protocols, Simulators) necessary to blind the potential magic in a quantum computation expressed in the Clifford+MSI model. As already stated in the introduction, this is a fundamental requirement to derive classically simulable instances that could be delegated indistinguishably to the Server, and hence serve as tests that the Client can use in a verification protocol. The section is organized as follows:
\begin{itemize}
    \item In~\ref{subsection:blind:gadget}, we show how to blind the magic on a single Clifford+MSI layer, introducing the Hidden-Magic Resource~\ref{resource:blind-gate} and the Blind State Injection Protocol~\ref{protocol:blind-gate} that implements it, by applying a Pauli encryption on the input state and the injected ancilla, and compensating it later.
    \item After repeating the above protocol for each layer in the computation, what is missing to finish the computation is to perform a last Clifford layer and measurements on a state that is Pauli-encrypted. Thus, in~\ref{subsection:blind:meas}, we present the Blind Measurements Resource~\ref{resource:blind-meas} and Protocol~\ref{protocol:blind-meas} that implements it.
    \item In~\ref{subsection:blind:MB-DQC}, we compose the above protocols to implement the Magic-Blind Delegated Quantum Computing Resource~\ref{resource:mblind_DQC}. Simply composing the protocols, yielding Protocol~\ref{protocol:mblind_DQC-bf}, has a practical caveat: it implies back-and-forth communication of encrypted qubits between the Client and the Server. We show how we can safely remove this requirement and obtain the MB-DQC Protocol~\ref{protocol:mblind_DQC}, that we use throughout the rest of the paper.
    \item Lastly, in~\ref{subsection:blind:Pauli dev}, we show the powerful result that any malicious behavior from the Server—meaning applying any CPTP map instead of performing the protocol honestly—can be reduced to a convex combination of Pauli operators applied after an honest execution of the protocol, before the measurements. This arrives as a consequence of the Pauli encryption of the qubits, allowing us to perform a Pauli Twirl of any malicious behavior.
\end{itemize}
\subsection{Blinding Magic on a single layer}
\label{subsection:blind:gadget}

In the Clifford+MSI model, the only non-Clifford ingredient is the injected single-qubit state together with a simple measurement-dependent correction. This makes state injection the natural ``locus'' where a Client can hide potential magic while delegating only Clifford operations.

We therefore isolate the following elementary delegation task. The Client holds an $n$-qubit input state\footnote{Note that this differs from the typical classical inputs considered in $\BQP$ computations. This is not a problem here because in this section we do not describe a $\BQP$ computation, but quantum evolution on states of $n$ qubits.} $\rho$, chooses a public Clifford layer $\C\in\Clifford_n$, and chooses a state to inject via a label $\A$.
The goal is to implement the corresponding transformation on $\rho$ while hiding from the Server both the data $\rho$ and the choice of injection type $\A$.
This is captured by the Hidden-Magic Gate Resource~\ref{resource:blind-gate}, which we view as the circuit-model analogue of performing the Clifford layer followed by a ``blind gadget'' implementing one compiled (potentially) non-Clifford layer. The Resource supports two regimes:
\begin{enumerate}
    \item \textbf{Magic regime.} When $\Alabel=\Tlabel$, the resource performs a Magic-State Injection.
    \item \textbf{Stabilizer regime.} When $\Alabel\neq \Tlabel$ and $\rho$ is a stabilizer input, the resource returns the result of injecting state $\ket\Alabel$, including the (classically simulable) measurement outcome.
\end{enumerate}

In both regimes, the Server learns only the public information (in particular $\C$ and the system size), but gains no information about $\rho$ nor about the Client's choice $\A$. It only learns the set of possible choices $\classA=\{\Tlabel, \pm\Xlabel,  \pm\Ylabel,  \pm\Zlabel\}$.
\begin{resourcee}[h]
    \caption{Hidden-Magic Gate}
    \label{resource:blind-gate}
    \begin{algorithmic}[0]
        \State \textbf{Public Information:} $\C\in\Clifford_n$, $\classA$.
        \State \textbf{Client's Input:} $\C\in\Clifford_n, \Alabel\in\classA$, and a classical description of a quantum state $\rho$ on $n$ qubits.
        \State \textbf{Server's Input:} $\c\in\bin$, set to $0$ if honest. If $\c=1$, the Server will have the opportunity to provide more inputs.

        \Procedure{Computation by the Resource}{}
        \If{$\c=0$}
            \If{$\Alabel=\Tlabel$}
            \State Output $\T_n\circ \C[\rho]$ at the Client interface.
            \Else
            \State Resource samples $b$ with probability $\bra b
            \Tr_{1, ..., n}[\F\circ \C[\rho\otimes \rho_\Alabel]]
            \ket b$
            \State Output
            $
            (\id_n\otimes \ketbra b)
            \circ
            \F\circ \C[\rho\otimes\rho_\Alabel]
            $ at the Client interface.
            \EndIf
        \Else
            \State The Server holds a system whose reduced state in register $S$ is $\rho_S$: it sends $\rho_S$, alongside the instructions to perform a CPTP map $\mathrm D$ on $\rho$ and $\rho_S$. The Resource returns $\Tr_S[\mathrm D[\rho\otimes \rho_\Alabel\otimes \rho_S]]$ to the Client.
        \EndIf
        \EndProcedure
    \end{algorithmic}
\end{resourcee}
To implement this ideal behavior, we propose the Blind State Injection Protocol~\ref{protocol:blind-gate}, in which the goal of the Client is to have the Server perform a Clifford layer and an unknown state injection.
The way blindness works is indeed very similar to how blindness appears in UBQC \cite{BFK09universal}: states are encrypted and sent to the Server alongside the instruction to perform a rotation with a specific angle that compensates for the encryption.

To do so, the Client first encrypts all the qubits using $\QOTP$ yielding secret keys $\a, \r$, and sends them to the Server. Then she
 sends instructions to perform the Clifford layer and inject the state ($\CNOT$,
$\SWAP$, and measurement), in plain. Finally she needs to have the Server perform the correction (conditioned $\pi/2$ rotation) if $\A=\Tlabel$. This is equivalent to sending to the Server the value $\phi$ computed from the measurement outcome $b\in\bin$ and the injection type $\A\in\classA$ as
\begin{equation}
    \label{eq:blind:phi}
    \phi(\A, b) = \begin{cases}
        b\pi/2 &\text{if}\; \Alabel=\Tlabel\\
        0&\text{else}\;
    \end{cases}.
\end{equation}

Sending $\phi$ in plain would leak information about the choice of $\A$, so to avoid that, the angle needs to be \emph{blinded}: padded with a uniformly sampled angle $\theta\gets \$\Theta$, yielding a blinded angle $\delta$. The ancilla needs to be pre-rotated with the same $\theta$ so the encryption is canceled and has no impact on the final state (it only ensures the Server is blind).
The blinded angle is thus $\phi+\theta$ with a global sign that comes from the Pauli $\X$ encryption since $\Z(\alpha)\circ\X^a = \X^a\circ\Z((-1)^a\alpha)$ for any angle $\alpha$ and bit $a\in\bin$.
Hence, the logic to compute the angle $\delta$ from the plain angle $\phi\in\Theta$, secret $\theta\in\Theta$, and secret $\X$ encryption bit $a\in\bin$ is
\begin{equation}
    \label{eq:blinded angle}
    \delta(\phi, a, \theta)=(-1)^a (\phi+\theta).
\end{equation}
In practice, since the $\Z^\dagger(\delta)$ rotation is to be performed on the $n$-th qubit in the Protocol, Equation~\ref{eq:blinded angle} is used with the updated $\X$ encryption key on the $n$-th qubit, $a_n$.
Finally, Theorem~\ref{thm:blind-gate} states that the Protocol perfectly implements the ideal behavior, as established by its correctness and security proofs.

\begin{protocoll}
    \caption{Blind State Injection}
    \label{protocol:blind-gate}
    \begin{algorithmic}[1]
        \State \textbf{Public Information:} $\C\in\Clifford_n$, $\classA$.
        \State \textbf{Client's Input:} $\C\in\Clifford_n, \A\in\classA$, and a quantum state $\rho$ on $n$ qubits.

        \Procedure{Client - Encrypt and Send}{}
        \State Sample $\a, \r \gets\$\{0,1\}^n$, send $\Enc{\a}{\r}[\rho]$ to the Server.
        \State Sample $x, z \gets\$\{0,1\}$, $\theta\gets\$\Theta$, send $\Enc{x}{z}\Z(\theta)\ket{\Alabel}$ to Server register $n+1$.
        \State Set $\a \gets (\a ||x), \r\gets (\r||z)$. \Comment{Append existing keys}
        \EndProcedure

        \Procedure{Server - Clifford and State Injection}{}
        \State Apply $\C$ on the first $n$ qubits, and $\F=\SWAP_{n+1, n}\circ \CNOT_{n+1, n}$ on the total system.
        \State Measure qubit $n+1$ in the computational basis, send outcome $b$ to the Client.
        \EndProcedure

        \Procedure{Client - Computing and Blinding rotation angle}{}
        \State Compute updated keys $\a,\r \gets \F\circ\C(\a,\r)$.
        \State Decode the measurement outcome: store $b\gets b\oplus a_{n+1}$.
        \State Compute appropriate rotation angle $\phi(\Alabel, b)$ according to Equation~\ref{eq:blind:phi}.
        \State Compute blind angle $\delta(\phi, a_n, \theta)$, and send it to the Server.
        \label{protocolstep:blind angle}
        \EndProcedure

        \Procedure{Server - Blind Rotation}{}
        \State Apply rotation $\Z^\dagger(\delta)$ on the $n$-th qubit.
        \EndProcedure

        \Procedure{Client - Key update, Receive and Decrypt}{}
        \State Receive $n$ qubits from Server.
        \If{$\Alabel=\Tlabel$}
            \State Update $a_n\gets a_n\oplus b$ \Comment{Pauli correction required for MSI}
        \Else
            \State Store $b$ as classical output.
        \EndIf
        \State Truncate $\a, \r$ to the first $n$ bits, and apply $\Dec{\a}{\r}$ on the received qubits.
        \EndProcedure

    \end{algorithmic}
\end{protocoll}

\begin{theorem}
\label{thm:blind-gate}
The Blind State Injection Protocol~\ref{protocol:blind-gate} perfectly constructs the Hidden-Magic Gate Resource~\ref{resource:blind-gate}.
\end{theorem}

\begin{proof}[Proof of correctness (sketch, formal in~\ref{appendix: correctness of blind-gate})]
    The Protocol consists of applying the Clifford circuit, performing a measurement, and the appropriate correction (rotation, and Pauli in the magic regime) on a state that is initially encrypted by a $\QOTP$ by the Client, which is a Pauli encryption. Since all the operations that are performed are Clifford, the evolution of the pad can be tracked efficiently by the Client so the state and measurement outcome can be correctly un-padded. Lastly, the $\Z^\dagger(\delta)$ rotation is performed on the $n$-th qubit; by Equation~\ref{eq:blinded angle}, this exactly implements the required correction while compensating for both the Pauli $\X$ encryption and the ancilla pre-rotation. The presence of encryption and the fact that the ancilla is pre-rotated are therefore fully accounted for in Equation~\ref{eq:blinded angle}.
\end{proof}

\begin{proof}[Security proof]
  For the security proof, we are interested in proving the existence of a Simulator such that any malicious behavior by the Server in the Real World in Protocol~\ref{protocol:blind-gate} can be reproduced by the Simulator interacting with Resource~\ref{resource:blind-gate} in the Ideal World. Formally, both scenarios must be indistinguishable from the point of view of an unbounded Distinguisher that controls both the Client inputs and the Server deviation.
    This means that the \textit{transcript}---the total quantum state perceived by the distinguisher in one scenario for a chosen input and choice of malicious behavior---generated in both scenarios must be statistically indistinguishable.

  Figure~\ref{subfigure:secproof:initial} captures the interactions in the Real World. Playing both the Client and the Server roles, the Distinguisher chooses the input (quantum state and state label) and a cheating behavior that respects the interfaces of the Protocol, meaning that it sends a bit $b$ to the Client.
      \begin{figure}[h]
            \centering
            \includegraphics{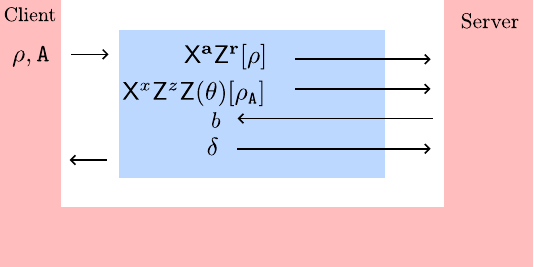}
            \caption{Representation of the interactions between a Client and a potentially malicious Server in Protocol~\ref{protocol:blind-gate}, using the color code of Figure~\ref{fig:prelims:AC}.
            The Client protocol blinds the input state, prepares and blinds the ancilla, and sends the resulting states to the Server. When receiving a bit $b$ from the Server, the blind angle $\delta$ is computed and sent to the Server. Finally, the Client receives $n$ qubits from the Server.
            }
            \label{subfigure:secproof:initial}
    \end{figure}

  To prove the existence of a Simulator that can be plugged in Resource~\ref{resource:blind-gate} to reproduce this behavior, we do the following.
   We reproduce the steps of \cite{DFPR14composable}, and present Reduction~\ref{protocol:blind-gate-reduction}: it is an EPR reduction of Protocol~\ref{protocol:blind-gate} (using the terminology introduced in the preliminaries), in which the Client uses the fact that measurements are delayed to choose and send a uniformly distributed rotation angle $\delta$ and compute $\theta$ simply by inverting Equation~\ref{eq:blinded angle}:
  \begin{equation}
    \label{eq:blinded angle inverted}
    \theta = (-1)^{a}\delta - \phi.
  \end{equation}
  Then, we show that this Reduction can instead be performed by a Simulator plugged in the Resource.
  On Figure~\ref{subfigure:secproof:epr}, we represent the first part of the reduction: replacing $\QOTP$ by EPR-encryption in the Client part of the Protocol. The transcript (\textit{i.e} the overall quantum state perceived by the Distinguisher: the $n$ qubit state, the ancilla qubit, and the rotation) generated in this version is clearly identical to the one generated by the initial protocol, by correctness of this EPR-realization of the $\QOTP$ (see \cite{SP00simple} and preliminaries). At this stage, when measurement outcomes are $\a, \r$, $x, z$, according to the notations of Figure~\ref{subfigure:secproof:epr}, the Server holds the state $\Enc{\a}{\r}[\rho]\otimes \Enc{x}{z}\circ\Z(\theta)[\rho_\Alabel]$, identical to what it holds in Protocol~\ref{protocol:blind-gate}
  for the same key bits obtained by uniform sampling.
      \begin{figure}[h]
            \centering
            \includegraphics{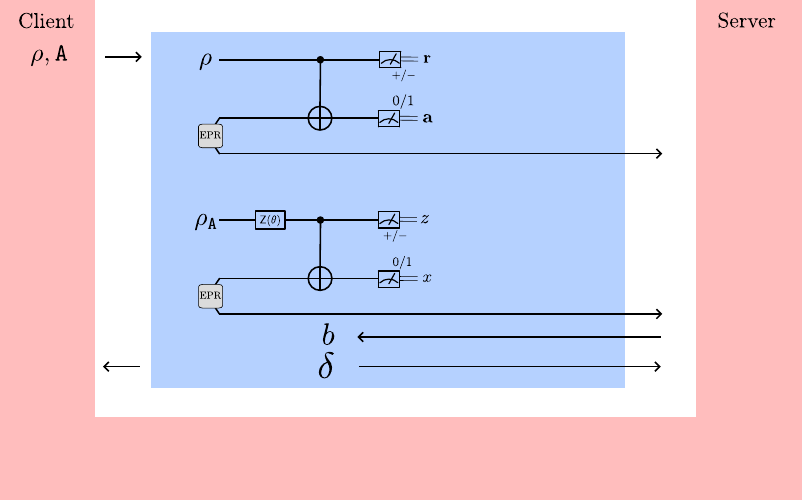}
            \caption{Same as Figure~\ref{subfigure:secproof:initial}, but the Client part of the Protocol used EPR-encryption to perform the $\QOTP$ on the $n$-qubit state $\rho$, and on the single-qubit, ancilla pre-rotated with $\Z(\theta)$.
            }
            \label{subfigure:secproof:epr}
    \end{figure}

    Then, to reach Reduction~\ref{protocol:blind-gate-reduction}, we use a delayed-measurement trick to show that instead of choosing $\theta$ at random and computing $\delta$ from it using Equation~\ref{eq:blinded angle}, the Client can choose $\delta$ at random and compute $\theta$ from it using Equation~\ref{eq:blinded angle inverted}. Indeed, on Figure~\ref{subfigure:secproof:epr}, the $\Z(\theta)$ commutes with the control: it can thus be delayed and performed right before the Hadamard-basis measurement that yields outcome $z$, which is then stored as $r_{n+1}$.
    Now we make the following two observations, captured in Lemma~\ref{lemma:blind-gate:key dependency}.
    The first one is that $a'_n$ (the bit used in the rotation angle logic) does not depend on $r_{n+1}$. The other one is that $a'_{n+1}$ (the bit used to decode the measurement outcome) neither.
    Therefore, since $r_{n+1}$ is not needed to compute $a'_n$ nor $a'_{n+1}$, this measurement can be delayed to after the interaction with the Server.

    \begin{lemma}
        \label{lemma:blind-gate:key dependency}
        Let there be a $n+1$-qubit system, $\C\in\Clifford_n$ and $\F=\SWAP_{n+1,n}\circ\CNOT_{n+1,n}$. Let $\a, \r\in\bin^{n+1}$ Pauli encryption keys and $\a', \r'=\F\circ \C(\a, \r)$. Then, neither $a'_n$ nor $a'_{n+1}$ depend on $r_{n+1}$.
    \end{lemma}
    \begin{proof}
        This can be proven just by analyzing how the $\Z$-encryption on the $n+1$-th qubit commutes through $\F\circ \C$. Since $\C$ acts on the first $n$ qubits, it commutes trivially. Then, through $\F$, it commutes through the control of the $\CNOT$ and ends as a $\Z$ encryption on the $n$-th qubit after the $\SWAP$. Hence, we conclude that $r_{n+1}$ does not influence $a'_{n}$ nor $a'_{n+1}$.
    \end{proof}

    Thus, in Reduction~\ref{protocol:blind-gate-reduction}, the Client first chooses $\delta$ uniformly at random, and computes $\theta$ from Equation~\ref{eq:blinded angle inverted}, without changing the distribution of the overall state.
    Finally, the Client performs the delayed $\Z(\theta)$ rotation and the Hadamard basis measurement to obtain $r_{n+1}$ and re-compute the updated keys to decrypt the qubits identically as in the original version.
    This is depicted on Figure~\ref{secproof:delayedrot}: for the same choice of inputs and malicious behavior, this interaction produces the same transcript as the initial Protocol in Figure~\ref{subfigure:secproof:initial}.
            \begin{figure}[h]
            \centering
            \includegraphics{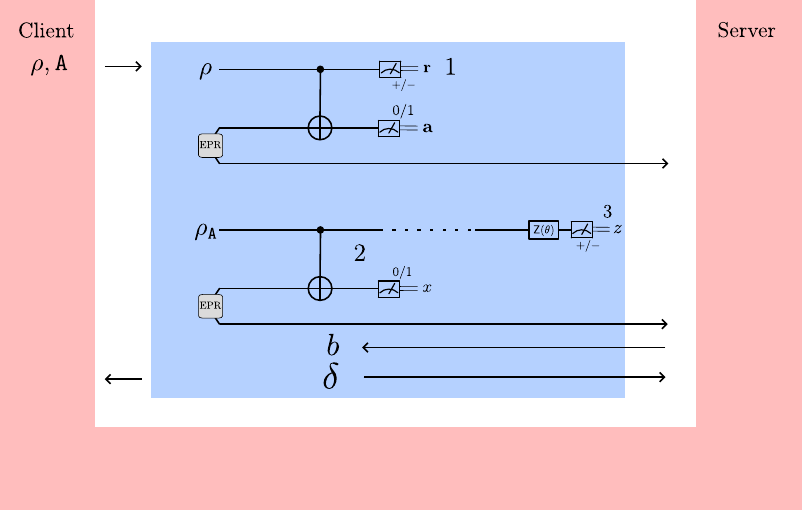}
            \caption{Illustration of the interactions in Reduction~\ref{protocol:blind-gate-reduction}, using the same color code as in Figure~\ref{subfigure:secproof:initial}. Here, the difference with Figure~\ref{subfigure:secproof:epr} is that $\delta$ is sampled uniformly and sent to the Server before the $n+1$-th encryption bit $z$ (equivalently $r_{n+1}$) is obtained, and $\theta$ is computed from $\delta$ according to Equation~\ref{eq:blinded angle inverted}.}
            \label{secproof:delayedrot}
    \end{figure}

    Finally, we introduce Simulator~\ref{simulator:blind-gate}: it deals with the Server interactions the same way as the Client in Reduction~\ref{protocol:blind-gate-reduction}, and lets the input-dependent operations be performed by Resource~\ref{resource:blind-gate}, according to its description (in the malicious scenario $c=1$, the Resource receives a quantum state and a CPTP map).
    The interaction generates the same transcript, since the same operations are being performed, just by different entities. In Figure~\ref{secproof:simulator}, the same operations as in Figure~\ref{secproof:delayedrot} are carried out, but some by the Simulator---preparing EPR pairs and sending random $\delta$---while the rest are carried out by a CPTP map that the Simulator asks the Resource to perform, according to the description of Resource~\ref{resource:blind-gate}.
    Therefore, for a given Client input and Server cheating behavior, the same transformation as in the initial protocol is being implemented, but by the Simulator plugged into the Ideal Resource, and hence the two worlds are indistinguishable from the point of view of the Distinguisher, which concludes the proof.

    \begin{figure}[h]
            \centering
            \includegraphics{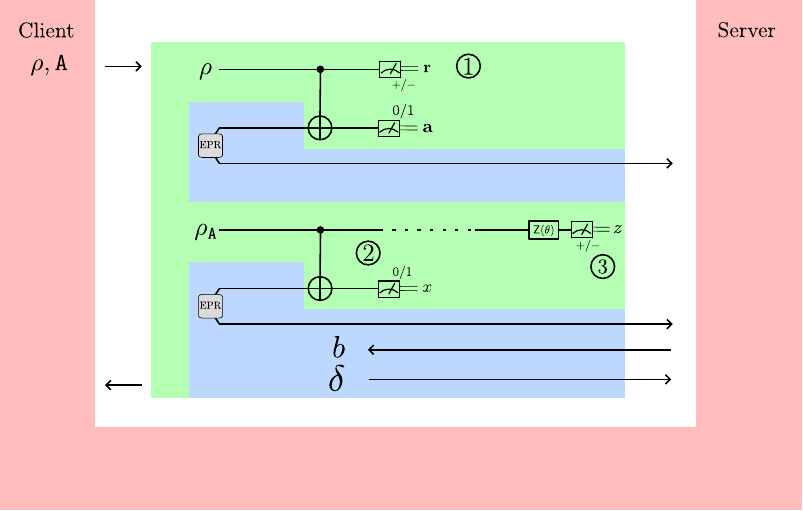}
            \caption{Same interaction as Figure~\ref{secproof:delayedrot}, but now the blue box represents \textbf{Simulator}~\ref{simulator:blind-gate}, and the green box represents the CPTP map performed by the Ideal Resource: step $1$ is the encryption of $\rho$, step $2$ is the EPR-encryption of $\rho_\A$ part 1, and step $3$ is part 2. Compared to Figure~\ref{subfigure:secproof:initial}, the execution of Protocol~\ref{protocol:blind-gate} has been replaced by Simulator~\ref{simulator:blind-gate} plugged into Resource~\ref{resource:blind-gate} with the required instructions to reproduce the deviation chosen by the Server interface of the Distinguisher. Since the transcript generated by the interactions in Figure~\ref{subfigure:secproof:initial} and the current figure are indistinguishable (given a choice of inputs/deviation from the Distinguisher), the two scenarios are indistinguishable and therefore the composable security is proven.}
            \label{secproof:simulator}
    \end{figure}

\begin{reduction}
    \caption{Blind State Injection, EPR-Client, malicious Server (reduction)}
    \label{protocol:blind-gate-reduction}
    \begin{algorithmic}[1]
        \State $\triangleright$ \textsc{Client - } Prepare $n$ EPR-pairs and send half of each pair to the Server.
        \Procedure{Client - EPR-encryption of $\rho$}{}
        \For{$1\leq i\leq n$}
            \State Apply $\CNOT$ on $i$-th qubit of $\rho$ (as control) and its $i$-th EPR-half (as target).
            \State Measure the $i$-th EPR half in computational basis, label the outcome $a_i$.
            \State Measure $i$-th qubit of $\rho$ in the Hadamard basis, label the outcome $r_i$.
        \EndFor
        \State Write the resulting $n$-bit strings $\a, \r$.
        \EndProcedure

        \State$\triangleright$ \textsc{Client - } Prepare an EPR-pair and send half to the Server.
        \Procedure{Client - EPR-encryption of ancilla, part 1}{}
        \State Prepare $\rho_\A=\ketbra{\A}$.
        \State Apply $\CNOT$ on $\rho_\A$ (as control) and the EPR-half (as target).
        \State Measure the EPR-half in the computational basis, and label the outcome $x$.
        \State Set $\a' \gets (\a ||x), \r'\gets (\r||0)$, as $z$ is obtained only in part 2.
        \EndProcedure

        \State$\triangleright$ \textsc{Client - } receives bit $b$ from Malicious Server.

        \Procedure{Client - Computing and Blinding rotation angle}{}
        \State Compute $\a', \r' \gets \F\circ \C(\a', \r')$. Note that $z$ does not affect $a'_{n+1}$ nor $a'_n$, which are the only ones we need for this step.
        \State Decode the measurement outcome: receive $b$, store $b\gets b\oplus a'_{n+1}$.
        \State Compute appropriate measurement angle $\phi(\A, b)$.
        \EndProcedure
        \State $\triangleright$\textsc{Client - } Sample random rotation angle uniformly. $\delta\gets\$\Theta$ and send it to the Server.

        \Procedure{Malicious Server - Blind Rotation}{}
        \EndProcedure

        \Procedure{Client - EPR-encryption of ancilla, part 2}{}
        \State Compute appropriate pre-rotation angle $\theta(\phi, a'_n, \delta)=(-1)^{a'_n}\delta - \phi$.
        \State Apply $\Z(\theta)$ on the ancillary qubit, measure in the Hadamard basis, label the outcome $z$.
        \State Set $\a \gets (\a ||x), \r\gets (\r||z)$ and re-compute key-update: $\a, \r = \F\circ \C(\a, \r)$.
        \EndProcedure

        \Procedure{Client - Key update, Receive and Decrypt}{id. Protocol~\ref{protocol:blind-gate}}
        \EndProcedure

    \end{algorithmic}
\end{reduction}

\begin{simulatorr}
    \caption{Blind State Injection}
    \label{simulator:blind-gate}
    \begin{algorithmic}[1]
        \Procedure{Simulator - Emulates Client interaction with Server}{}
        \State Prepare $n+1$ EPR Pairs, and send each half to the Server.
        \State Receive bit $b$ from the Server.
        \State Sample $\delta \gets \$\Theta$ and send it to the Server.
        \EndProcedure

        \Procedure{Simulator - Make Resource reproduce the Server behavior}{}
        \State Send the remaining $n+1$ EPR halves, $b$, and $\delta$ to Ideal Resource
        \State Send the below instructions as a CPTP map to the Resource to perform.
        \EndProcedure
        \Procedure{CPTP map sent to the Resource}{$\A, \rho$ from Client ; $n+1$ EPR-halves, $\delta, b$ from Simulator}
        \State EPR-encryption of $\rho$
        \State EPR-encryption of $\rho_\A$, part 1
        \State Compute rotation angle
        \State EPR-encryption of ancilla, part 2
        \State Key update, Receive and Decrypt
        \EndProcedure

    \end{algorithmic}
\end{simulatorr}
 \end{proof}
\subsection{Blind measurements}
\label{subsection:blind:meas}
In this section, we introduce the ideal functionality allowing a Client to delegate the execution of a Clifford circuit on $n$ qubits followed by a measurement of the $n$ qubits in the computational basis, with the security property that the Server performs these steps blindly on the data sent by the Client. Indeed, Resource~\ref{resource:blind-meas} captures this security property: the Server learns the intended Clifford circuit $\C$ and the size of the system, but not $\rho$ the $n$-qubit input of the Client, while the Client receives $\z$ sampled by the $\Z$-basis measurement distribution of the state $\C[\rho]$.
\begin{resourcee}
    \caption{Blind measurements}
    \label{resource:blind-meas}
    \begin{algorithmic}[1]
        \State \textbf{Public Information:} $n, \C$ and basis of measurements (computational basis).
        \State \textbf{Client's Input:} $n$-qubit quantum state $\rho$, Clifford circuit
        \State \textbf{Server's Input:} $\c\in\bin$, set to $0$ if honest. If $\c=1$, the Server will have the opportunity to provide more inputs.
        \Procedure{Computation by the Resource}{}
        \If{$\c=0$}
            \State Sample $\z$ with probability $\bra \z \C[\rho]\ket \z$
            \State Output $ \ketbra \z \circ \C[\rho]$
        \Else
            \State The Server holds a system whose reduced state in register $S$ is $\rho_S$: it sends $\rho_S$, alongside the instructions to perform a CPTP map $\mathrm D$ on $\rho$ and $\rho_S$. The Resource returns $\Tr_S[\mathrm D[\rho\otimes \rho_S]]$ to the Client.
        \EndIf
        \EndProcedure
    \end{algorithmic}
\end{resourcee}
We introduce Protocol~\ref{protocol:blind-meas} to implement this resource, which makes the simple use of a Quantum One-Time Pad applied on $\rho$ before being sent to the Server, and a decoding of the measurement outcome by a Pauli operator resulting of conjugating the initial Pauli (used for encryption) by the Clifford circuit performed by the Server.

\begin{protocoll}
    \caption{Blind measurements}
    \label{protocol:blind-meas}
    \begin{algorithmic}[1]
        \State \textbf{Public Information:} $n, \C$ and basis of measurements (computational basis).
        \State \textbf{Client's Input:} $n$-qubit quantum state $\rho$, Clifford circuit $\C\in\Clifford_n$.

        \Procedure{Client - Encrypt and Send}{}
        \State Sample $\a, \r\gets \bin^{n}$ and send $\Enc{\a}{\r}[\rho]$ to Server
        \EndProcedure

        \Procedure{Server - Clifford circuit and measurements}{}
        \State Perform $\C$ on the received qubits.
        \State Measure all the qubits in the computational basis, getting the $n$-bit outcomes $\z$
        \EndProcedure

        \Procedure{Client - Decode measurement outcomes}{}
        \State Client updates keys as $\a, \r\gets \C(\a, \r)$
        \State Store $\z\oplus \a$ as output of the protocol.
        \EndProcedure
    \end{algorithmic}
\end{protocoll}

\begin{theorem}
    \label{thm:blind-meas}
    The Blind Measurements Protocol~\ref{protocol:blind-meas} perfectly constructs Resource~\ref{resource:blind-meas}
    in the Abstract Cryptography framework.
    \end{theorem}

\begin{proof}[Proof of correctness]
    When both parties are honest, the Server receives $\Enc{\a}{\r}[\rho]$, applies $\C$, performs a computational-basis measurement of all qubits, and sends the measurement outcomes in an $n$-bit string $\z$ to the Client. The Client computes the conjugation of the initial encryption by the Clifford circuit, which gives $\a', \r'=\C(\a, \r)$, and stores $\z\oplus\a'$ as the output of the protocol. Formally, the output of the Client can thus be written
    \begin{align}
        \rho_{out, \z\oplus\a'} &=
        \ketbra{\z\oplus\a'}{\z}
        \circ
        \C
        \circ \Enc{\a}{\r}[\rho]
        \\
        &=
        \ketbra{\z\oplus\a'}{\z}
        \circ
        \Enc{\a'}{\r'}
        \circ
        \C
        [\rho].
    \end{align}
    With a change of variables $\z'= \z\oplus \a'$, then $\z$ becomes $\z'\oplus \a'$, and the expression becomes
    \begin{align}
        \rho_{out, \z'}
        &=
        \ketbra{\z'}{\z'\oplus \a'}
        \circ
        \Enc{\a'}{\r'}
        \circ
        \C
        [\rho]
        \\
        &=
        \ketbra{\z'}{\z'}
        \circ
        \X^{\a'}
        \circ
        \Enc{\a'}{\r'}
        \circ
        \C
        [\rho]
        \\
        &=
        \ketbra{\z'}{\z'}
        \circ
        \Dec{\a'}{\r'}
        \circ
        \Enc{\a'}{\r'}
        \circ
        \C
        [\rho]
        \\
        &=
        \ketbra{\z'}{\z'}
        \circ
        \C
        [\rho].
    \end{align}
    Thus, the output of the protocol is a bitstring $\z$ that follows the distribution $p(\z)=\bra\z\C[\rho]\ket\z$, identically to the output of Resource~\ref{resource:blind-meas}.
\end{proof}

\begin{proof}[Security proof (sketch, formal in~\ref{section:blindmeas-secproof-formal})]
     Here, we aim to present a Simulator that uses only the public information and interacts with Resource~\ref{resource:blind-meas} in a way that makes the Real World indistinguishable from the Ideal World for an unbounded distinguisher.
     To do so, we start by carrying out an EPR reduction of Protocol~\ref{protocol:blind-meas}: performing the Pauli encryption using EPR pairs and Bell measurements. Doing so clearly results in the same transformation from the Distinguisher's point of view. Then, we show that the same transformation can be reached by Simulator~\ref{simulator:blind-meas}, which uses only the public information and interacts with Resource~\ref{resource:blind-meas}. The equivalence is straightforward, as Simulator~\ref{simulator:blind-meas} performs the steps of the EPR reduction that do not require the input state, while instructing the Resource to do the rest.
     This concludes the proof, as in the presence of a malicious Server, Protocol~\ref{protocol:blind-meas} is indistinguishable from Simulator~\ref{simulator:blind-meas} plugged in Resource~\ref{resource:blind-meas} from an unbounded Distinguisher.

\begin{simulatorr}
    \caption{Blind Measurements}
    \label{simulator:blind-meas}
    \begin{algorithmic}[1]
        \State \textbf{Simulator Inputs:} $n, \C$ and basis of measurements (computational basis).
        \Procedure{Simulator - Emulates Client interaction with Server}{}
        \State Prepare $n$ EPR Pairs, and send each half to the Server.
        \State Receive $n$-bit string $\z$ from the Server.
        \EndProcedure

        \Procedure{Simulator - Make Resource reproduce the Server behavior}{}
        \State Send the remaining $n$ EPR halves and $\z$ to Ideal Resource.
        \State Send the below instructions as a CPTP map to the Resource to perform.
        \EndProcedure

        \Procedure{CPTP map sent to the Resource}{$\C, \rho$ from Client; $n$ EPR-halves, $\z$ from Simulator}
        \State EPR-encryption of $\rho$
        \State Decode measurement outcomes
        \EndProcedure
    \end{algorithmic}
\end{simulatorr} \end{proof}

\subsection{Magic-Blind Delegated Quantum Computations}
\label{subsection:blind:MB-DQC}
We now complete our goal of providing magic-blindness by composing the previously introduced protocols: we show how to hide the magic of an entire delegated computation $$C \;=\; \C_{t+1}\circ \T_n \circ \C_t \circ \cdots \circ \T_n \circ \C_1$$ followed by computational basis measurements where $\C_i\in\Clifford_n$ are Clifford layers and each $\T_n$ is implemented by a state-injection gadget acting on the $n$-th wire.
\subsubsection{Ideal functionality and induced computation class}
The blindness requirement is captured by the Magic-Blind Delegated Quantum Computation Resource~\ref{resource:mblind_DQC}.
\begin{resourcee}
    \caption{Magic-Blind DQC}
    \label{resource:mblind_DQC}
    \begin{algorithmic}[0]
        \State \textbf{Public Information:} $n, t, \C_1...\C_{t+1}$ and basis of measurements (computational basis).
        \State \textbf{Client's Inputs:} classical input $\x\in\bin^n$, $\C_1...\C_{t+1}\in\Clifford_n$ and $\{\A_i\in\classA\}_{i\leq t}$.
        \State \textbf{Server's Input:} $\c\in\bin$, set to $0$ if honest. If $\c=1$, the Server will have the opportunity to provide more inputs.

        \Procedure{re-naming}{}
        \State  $\star=1$ if $\A_i=\Tlabel$ for all $i$, $\star=0$ if $\A_i\neq \Tlabel$ for all $i$.
        \State  $\G=\C_{t+1}
        \circ \F_t\circ \C_t
        \circ \cdots
        \circ \F_1\circ \C_1$
        ,  $\rho_\x=\ketbra \x$
        ,  $\rho_\A=\bigotimes_{i=1}^t\rho_{\A_i}$
        \EndProcedure

        \Procedure{Computation by the Resource}{}
        \If{$\c=0$}
        \If{$\star$}
        \State Sample $\b$ with probability $\bra \b \C_{t+1}
        \circ
        \T_n\circ \C_t \circ \cdots \T_n\circ \C_1[\rho_\x]\ket \b$
        \State Output $
       \ketbra\b
        \circ
        \C_{t+1}
        \circ
        \T_n\circ \C_t \circ \cdots \T_n\circ \C_1[\rho_\x]
        $.
        \Else{}
        \State Sample $\b$ with probability $\bra \b \G[\rho_\x\otimes \rho_\A] \ket \b$
        \State Output $
        \ketbra\b
        \circ
        \G[\rho_\x\otimes \rho_\A]
        $.
        \EndIf
        \Else
        \State The Server holds a system whose reduced state in register $S$ is $\rho_S$: it sends $\rho_S$, alongside the instructions to perform a CPTP map $\mathrm D$ on the entire system. The Resource returns $\Tr_S[\mathrm D[\rho_\x\otimes \rho_\Alabel\otimes \rho_S]]$ to the Client.
        \EndIf
        \EndProcedure
    \end{algorithmic}
\end{resourcee}
The Client provides an $n$-bit input $\x$ and, for each injection step $i\leq t$, a choice $\A_i\in\classA$ specifying which single-qubit state is injected.
The public input to the resource consists of the Clifford structure $\C_1,\ldots,\C_{t+1}$, which is explicitly leaked to the Server. On the other hand, the inputs of the Client (choice of inputs and injections) are kept hidden from the Server, which constitutes the magic-blindness functionality.
The resource supports two operational modes, captured by $\star$ in Resource~\ref{resource:mblind_DQC}:
\begin{itemize}
    \item \textbf{Computation mode}, where $\A_i=\Tlabel$ for all $i\leq t$.
    In this case, the resource performs the target non-Clifford computation $C$ on input $\rho$ and returns the resulting measurement outcomes to the Client.
    \item \textbf{Magic-free mode}, where $\A_i\neq\Tlabel$ for all $i\leq t$.
    Here, all injections are stabilizer injections and the resulting transformation on the enlarged $n+t$ qubit system is Clifford. The outcomes correspond to a classically simulable computation—cf. Section~\ref{subsection:prelims:QC}, since input states are all stabilizers, transformation is Clifford, and measurements are Pauli—and serve only for testing.
\end{itemize}

In the magic-free mode, it is convenient to make the induced Clifford structure explicit. Writing $\F_i=\SWAP_{n+i,n}\circ\CNOT_{n+i,n}$ for the Clifford circuit implementing the $i$-th injection gadget, the overall transformation applied to the joint input $\rho_\x\otimes\rho_\A$ (where $\rho_\A=\bigotimes_{i=1}^t\rho_{\A_i}$) is the Clifford circuit
\begin{equation}
    \label{eq:blindness:G}
\G \;=\; \C_{t+1}\circ \F_t \circ \C_t \circ \cdots \circ \F_1 \circ \C_1,
\end{equation}
also depicted at Figure~\ref{fig:blindness:G}.
A final layer of computational-basis measurements is then applied to the remaining unmeasured qubits. Crucially, the Server observes the same sequence of Clifford operations in both modes.

\begin{figure}[h]
    \centering
    \includegraphics[]{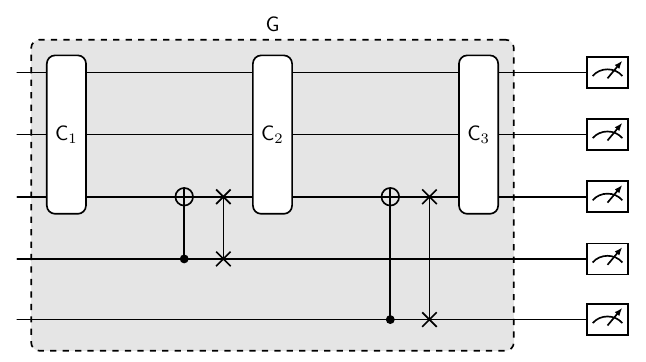}
    \caption{The Clifford transformation $\G$ that is implemented on $n+t$ qubits in the magic-free mode. Rotation angles are replaced by $0$, and no Pauli correction is being made.}
    \label{fig:blindness:G}
\end{figure}

\paragraph{Class of indistinguishable computations.}
By construction, Resource~\ref{resource:mblind_DQC} induces a natural \textit{class of computations that are indistinguishable to the Server under magic-blind delegation}, that we write $\class$.
All computations in this class share the same public Clifford structure $\C_1,\ldots,\C_{t+1}$, but may differ in their injected single-qubit states and in the associated measurement-angle logic.
In particular, $\class$ contains:
\begin{itemize}
    \item the \emph{target computation}, obtained by injecting magic states
    $\rho_\A=\bigotimes_{i\leq t}\T[\ketbra{+}]$ and using the measurement-dependent correction
    angles $\phi_i=b_i\pi/2$;
    \item a family of \emph{test computations}, obtained by injecting stabilizer states
    $\rho_{\A_i}\in\{\pm\ket \Xlabel,\pm\ket \Ylabel,\pm\ket \Zlabel\}$, for which the circuit does not need adaptivity and can use fixed angles
    (here set to $\phi_i=0$).
\end{itemize}

Without surprise, we propose the Magic-Blind Delegated Quantum Computation Protocol~\ref{protocol:mblind_DQC-bf}, a sequential composition of $t$ instances of the Blind State Injection Protocol~\ref{protocol:blind-gate}, adapting the injected ancilla, Clifford circuit, and ancilla index to the layer; then followed by the Blind Measurements Protocol~\ref{protocol:blind-meas}. As for the ideal behavior, it supports two modes: either $\A_i=\Tlabel$ for all $i$ (computation), or for no $i$ (magic-free). This is captured by a $\star$ in Protocol~\ref{protocol:mblind_DQC-bf}.
\begin{protocoll}
    \caption{Magic-Blind Delegated Quantum Computation (MB-DQC)}
    \label{protocol:mblind_DQC-bf}
    \begin{algorithmic}[1]
        \State \textbf{Public Information:} $n, t, \C_1...\C_{t+1}$ and basis of measurements (computational basis).
        \State \textbf{Client's Inputs:} classical input $\x\in\bin^n$, $\C_1...\C_{t+1}$ and $\A_i\in\classA$ for $i\leq t$.

        \State Client sets her computation register to $\rho_\x = \ketbra \x$.
        \Procedure{Clifford and Blind Gadget layers}{}
            \For{$i\leq t$}
                \State Client and Server perform Protocol~\ref{protocol:blind-gate} with the Client's computation register on inputs $\C_i, \A_i$ and ancilla register $n+i$ instead of $n+1$.
                \State
                Client sets the received $n$
                 qubits in her computation register, and if $\A_i\neq\Tlabel$ sets the received classical bit to $b_i$.
            \EndFor
        \EndProcedure

        \Procedure{Clifford and Blind measurements}{}
            \State Client and Server perform Protocol~\ref{protocol:blind-meas} on Clifford $\C_{t+1}$ and the Client's computation register.
            \State Output: Client receives $n$-bit string $\x$.
        \EndProcedure

        \Procedure{Output}{}
        \If{$\star$}
        \State \textbf{Output:} set $\b\gets \x$ and output $\b$.
    \Else
        \State \textbf{Output:} set $\b\gets \x\ \|\ b_1\ \|\ \cdots\ \|\ b_t$ and output $\b$.
    \EndIf
        \EndProcedure
    \end{algorithmic}
\end{protocoll}

\begin{theorem}
    \label{thm:mblind_DQC}
    The Magic-Blind Delegated Quantum Computation Protocol~\ref{protocol:mblind_DQC-bf} perfectly constructs Resource~\ref{resource:mblind_DQC}
    in the Abstract Cryptography framework.
    \end{theorem}
\begin{proof}
    This can be proven straightforwardly using the composability of the protocols presented before, since Resource ${\ref{resource:mblind_DQC}}$ on inputs $\C_1, ..., \C_{t+1}$, and $\A_1, ..., \A_t$ can be realized by using Resource~\ref{resource:blind-gate} on $\C_1, \A_1$ and index $1$, then again on $\C_2, \A_2$, and index $2$, and so on, and finally Resource~\ref{resource:blind-meas} with input $\C_{t+1}$.
      Formally, using the notations of the AC framework and denoting by $\Protocol_{\ref{protocol:blind-gate}}$ the converter of the Blind State Injection Protocol~\ref{protocol:blind-gate}, and by $\Protocol_{\ref{protocol:blind-meas}}$ the converter of the Blind Measurements Protocol~\ref{protocol:blind-meas}, we have:
      \begin{align}
          \Resource_{\ref{resource:mblind_DQC}}(\C_1, ..., \C_{t+1};\A_1, ..., \A_t)
          &=
          \Resource_{\ref{resource:blind-meas}}(\C_{t+1})
          \circ
          \Resource_{\ref{resource:blind-gate}}(\C_t; \A_t; t)
          \circ
          \cdots
          \circ
          \Resource_{\ref{resource:blind-gate}}(\C_1; \A_1; 1)
          \\
          &=
          \left( \Protocol_{\ref{protocol:blind-meas}}(\C_{t+1})
          \circ
          \Protocol_{\ref{protocol:blind-gate}}(\C_t; \A_t; t)
          \circ
          \cdots
          \circ
          \Protocol_{\ref{protocol:blind-gate}}(\C_1; \A_1; 1) \right)\qchannel
          \\
          &=
          \Protocol_{\ref{protocol:mblind_DQC-bf}}(\C_1, ..., \C_{t+1};\A_1, ..., \A_t) \qchannel.
      \end{align}
  \end{proof}

\subsubsection{Removing the back-and-forth communication between Client and Server}
Protocol~\ref{protocol:mblind_DQC-bf} has a practical caveat. It uses a two-way quantum communication channel since it requires the Client, for consecutive calls to Protocol~\ref{protocol:blind-gate}, to receive $n$ qubits, decrypt them, encrypt them again, and re-send them---this is also the case for the last call between Protocol~\ref{protocol:blind-gate} and Protocol~\ref{protocol:blind-meas}. In this section, we show that this "back-and-forth" communication can be removed. The first forward (encrypted) quantum communication is still required. Then, instead of receiving, decrypting, re-encrypting, and re-sending qubits, the Client can simply instruct the Server to keep the qubits while treating the decryption keys as new encryption keys kept hidden from the Server. Doing so yields Protocol~\ref{protocol:mblind_DQC}. Indeed, we replace the textbook execution of Protocols~\ref{protocol:blind-gate} and~\ref{protocol:blind-meas} by equivalent procedures in which the Client neither sends freshly encrypted qubits nor receives qubits to decrypt, since the Server keeps the qubits.
\begin{protocoll}
\caption{MB-DQC (without Back-And-Forth)}
\label{protocol:mblind_DQC}
\begin{algorithmic}[1]
\State \textbf{Public:} $\C_1,\ldots,\C_{t+1}\in\Clifford_n$
\State \textbf{Client input:}  $\x\in\bin^n$, $\A_i\in\classA$ for $i\le t$
\State \textbf{Notation:} $\star$ if $\A_i=\Tlabel$ for all $i\leq t$. Steps with "\textbf{C}" are done by the Client, "\textbf{S}" by the Server.

\Procedure{Blind-State Injection}{$i,\C,\A;\ (\a,\r)$}
  \State \textbf{C:} sample $x,z\gets\$\{0,1\}$, $\theta\gets\$\Theta$.

  \State \textbf{C:} prepare $\Enc{x}{z}\Z(\theta)\ket{\Apattern}$ and send to Server register $n+i$.
  \State \textbf{C:} set $(\a,\r) \gets(\a\|x,\ \r\|z)$.
  \State \textbf{S:} apply $\C$ and then $\F_i := \SWAP_{n+i,n}\circ\CNOT_{n+i,n}$.
  \State \textbf{S:} measure qubit $n+i$ in computational basis; send outcome $b_i$.
  \State \textbf{C:} update keys $(\a,\r)\gets (\F_i\circ \C)(\a,\r)$.
  \State \textbf{C:} decode $b_i \gets b_i \oplus a_{n+i}$.
  \State \textbf{C:} compute appropriate rotation angle $\phi(\A_i, b_i)$.
  \State \textbf{C:} compute blind angle $\delta(\phi, a_n, \theta)$.
  \State \textbf{S:} apply $\Z^\dagger(\delta)$ on qubit $n$.
  \If{$\A=\Tlabel$}
    \State \textbf{C:} $a_n \gets a_n \oplus b_i$ \Comment{Pauli correction for MSI}
  \EndIf
  \State \textbf{Output:} updated $(\a,\r)$ and (if $\A\neq\Tlabel$) the classical bit $b_i$.
\EndProcedure

\Procedure{Blind Measurements}{$\C;\ (\a,\r)$}
  \State \textbf{S:} apply $\C$ and measure all qubits in the computational basis, obtaining $\z$.
  \State \textbf{C:} update keys $(\a,\r)\gets \C(\a,\r)$ and output $\z\oplus \a$.
\EndProcedure

\Statex
\Procedure{Computation}{}
  \State \textbf{C:} sample $(\a,\r)\gets\$\{0,1\}^{n}\times\{0,1\}^{n}$ and send $\Enc{\a}{\r}[\rho]$.
  \For{$i=1, ..., t$}
    \State run \textsc{Blind-State Injection}$(i,\C_i,\A_i;\ (\a,\r))$;
          store classical bit as $b_i$ if $\A_i\neq\Tlabel$.
  \EndFor
  \State Run \textsc{Blind Measurements}$(\C_{t+1};\ (\a,\r))$, output $\z$.
  \If{$\star$}
    \State \textbf{Output:} set $\b\gets \z$ and output $\b$.
  \Else
    \State \textbf{Output:} set $\b\gets \z\ \|\ b_1\ \|\ \cdots\ \|\ b_t$ and output $\b$.
  \EndIf
\EndProcedure
\end{algorithmic}
\end{protocoll}

\begin{theorem}
    \label{thm:mblind_DQC-noBF}
    The Magic-Blind Delegated Quantum Computation Protocol without Back-and-Forth, Protocol~\ref{protocol:mblind_DQC} perfectly constructs Resource~\ref{resource:mblind_DQC} in the Abstract Cryptography framework.
    \end{theorem}

\begin{proof}[Sketch of proof (formal in~\ref{subsection:security proof MB-DQC})]
    We start the proof by reminding that Protocol~\ref{protocol:blind-gate} is composably secure: it is secure in all context in which it is used, as long as the Client follows her part of the protocol, \textit{i.e} the keys used for each call are uniformly distributed and kept hidden from the Server.
    The main point here is that the Client can simply use Protocol~\ref{protocol:blind-gate} repeatedly, decrypt the qubits and re-encrypt them with the same keys, without compromising the security of the protocol. Finally, since the Client is receiving qubits, decrypting them, encrypting them with the same key, and sending them back, this step can be omitted to avoid a waste of quantum communication, and simply let the Server keep the qubits, without compromising the security of the protocol. This results in Protocol~\ref{protocol:mblind_DQC}, in which the Client sends the (encrypted) $n$-qubit state only once at the beginning, then at each call of Protocol~\ref{protocol:blind-gate}, she only sends the encrypted ancilla, and updates the encryption keys.
\end{proof}

\subsection{Reduction to Pauli Deviations}
\label{subsection:blind:Pauli dev}
A central result that emerges from the previously introduced protocol is that any malicious Server behavior can be reduced to Pauli deviations on the qubits before their measurements in the computational basis. This is due to the initial Pauli encryption of the Client's $n+t$ qubits that can be propagated through  each layer of the protocol, and the decryption that takes place inside the Client's part of Protocol~\ref{protocol:mblind_DQC}.
This allows a Pauli Twirl of any Server malicious behavior, even those that aim to exploit the received angles values or an additional fixed-size working register. This is captured in the following lemma.

\begin{restatable}[Reduction to Pauli deviations]{lemma}{paulidev}
    \label{lemma:mblind:Pauli deviations}
    Let $\class$ be the class of computations indistinguishable under magic-blind delegation with Clifford structure $\C_{1}, ..., \C_{t+1}$, let  $C\in \class$, and let $\rho_{out}$ denote the state of the total system after a Client delegates $C$ to a (malicious) Server under Protocol~\ref{protocol:mblind_DQC}. Then, for any Server cheating behavior, there exists a convex combination of Pauli operators with coefficients $\alpha_\E$, where $\sum_{\E\in\Pauli_{n+t}}\abs{\alpha_\E}^2=1$, such that the state of the total system (including bits discarded by the Client) can be written as
    \begin{equation}
        \rho_{out}
        =
        \sum_{\b\in\bin^{n+t}}
        \sum_{\E\in\Pauli_{n+t}}
        \abs{\alpha_\E}^2
        \ketbra\b
        \circ
        \E
        [
            \rho_{cor, \b, C}
        ]
        \otimes \U_\E\left[
            \ketbra0^{\otimes w}
            \right],
    \end{equation}
    where
    \begin{equation}
        \rho_{cor, \b, C}
        =
        \C_{t+1}
        \circ
        \Z_n^\dagger\left(\phi_\b^{(t)}\right)
        \circ
        \F_{t}
        \circ
        \C_{t}
        \circ
        \cdots
        \Z_n^\dagger\left(\phi_\b^{(1)}\right)
        \circ
        \F_{1}
        \circ
        \C_{1}
        \left[\rho \otimes \rho_\Apattern\right]
    \end{equation}
    is the correct state state, meaning after an honest execution of $C$ alongside the branch $\b$.
\end{restatable}

\begin{proof}[Sketch of proof, formal in~\ref{subsection:proof of reduction to Pauli Dev}]
    This lemma can be proven by following the proof technique of \cite{FK17unconditionally,KKLM22unifying}. We first fix a branch of measurement outcomes $\b$ of $n+t$ bits, and denote by $\b'$ the decoded outcomes that the Client stores, which we refer to as the \emph{computation branch}.
    This allows to unitarize the protocol, expressing it as a sequence of unitary operations alongside branch $\b'$ for an honest behavior before a layer of $n+t$ measurements in the computational basis yielding outcomes $\b$ that the Client stores as $\b'$.
    Furthermore, we write any CPTP map performed by the Server as a unitary acting on the qubits and angles sent by the Client and on a private working register. As a result, any Server behavior can be written as an honest-then-malicious unitary behavior acting on the total system.
    Since the sent qubits are encrypted and the measurement outcomes are decrypted, a Pauli Twirl of the deviation can be obtained, hence the announced result.

    This holds because we are able to sandwich the deviation between Pauli operators corresponding respectively to the encryption and the decryption, both unknown to the Server. Indeed, the Client starts by encrypting the total system and ends by decoding measurement outcomes, which is equivalent to decrypting the states before their measurement. In between, there is the honest unitary execution and the Server's unitary deviation on the entire system (including a private working register). The initial Pauli encryption is commuted progressively through the honest part of the protocol: at each layer, it removes the encryption of the rotation angle (\emph{i.e.}, $\delta$ becomes $\phi$) while Clifford operations progressively alter the keys. After all the honest layers of the protocol, the state is thus still perfectly one-time padded, and the encryption keys correspond to the decryption keys of the Client. Hence, the Server's unitary deviation on the entire system is sandwiched between unknown conjugate Pauli operators on the $n+t$-qubit system returned to the Client. As a result, the deviation undergoes a Pauli Twirl on that subsystem, and the result follows.
\end{proof}

\section{Verification \emph{via} Magic-Blindness}
\label{section:verification}
In this section, we use the previous constructions and results to tackle the verification of any $\BQP$ quantum computation on $n$ qubits with inherent error $\bqp$.
Throughout this section, we consider that $C$ is expressed in the Clifford+MSI model with $t$ state injection layers. Namely, with the convention introduced in preliminaries \ref{subsection:prelims:QC}, $C$ is a $(n,t)$-Clifford+MSI computation.
\paragraph{Description of the Ideal Resource.}
We will first define the Verified Delegated Quantum Computation Resource~\ref{resource:verification}. It captures mathematically the verifiability property of a delegated quantum computation, \emph{i.e.}, it is secure by design. It explicitly models the fact that the Server either behaves honestly or forces an abort, but cannot corrupt the result of an accepted computation.
Notably, the resource leaks the Clifford structure of the computation to the Server, showcasing that blindness of the entire circuit is not a requirement. When the Server is honest, the resource outputs the correct decision bit $z^\star \in \{0,1\}$ for the language $L$.
By our definition of $\BQP$, the first measurement outcome $y$ of the computation $C(\ket{\x})$ satisfies:
\(
\Pr[y = z^\star] \ge 1 - c .
\)

\begin{resourcee}
    \caption{Verified Delegated Quantum Computation}
    \label{resource:verification}
    \begin{algorithmic}[0]
        \State \textbf{Client inputs:} $\x \in \{0,1\}^n$, and $(n,t)$-Clifford+MSI computation $C$.
    \State \textbf{Public information:} Clifford structure $\C_{1}, \dots, \C_{t+1} \in \Clifford_n$, parameters of the protocol.
    \State \textbf{Server inputs:} $\c \in \{0,1\}$, set to $0$ if honest.
    \Procedure{Computation by the Resource}{}
        \If{$\c=0$}
            \State Output $\ketbra{\Acc} \otimes \ketbra{z^\star}$ to the Client, where $z^\star$ is the biased output of $C(\ket\x)$.
        \Else
            \State Output $\ketbra{\Rej} \otimes \ketbra{\perp}$ to the Client.
        \EndIf
    \EndProcedure
    \end{algorithmic}
\end{resourcee}
\paragraph{Rationale of the Protocol.}
The chosen design will follow the template of other trap-based protocols consisting of interleaving computation and test rounds, while delegating them to the Server in a blind way using the previously introduced Magic-Blind Delegated QC Protocol~\ref{protocol:mblind_DQC}. This will ensure that the honesty of the Server can be effectively deduced from what is observed on the test rounds alone.

\paragraph{Tolerance to a constant amount of circuit-level noise.}
We also show that the protocol remains relevant in the presence of a constant amount of circuit-level Server-side noise. It also naturally accounts for certain forms of Client-side noise, namely those that can be equivalently interpreted as occurring on the Server side (see the proof of robustness).
The noise is modeled by an error rate $p_{err}<1/2$ quantifying the probability that an honest run produces an incorrect result (other than because of the $\BQP$ error)—the same noise model considered in other robust verification protocols \cite{LMKO21verifying,BN25noise}. We provide an upper bound for the tolerated values of $p_{err}$.

\paragraph{Organization of the section.}
In order to do so, we will first define in~\ref{subsection:framework single qubit} preliminary notions to explain the Protocol. Using the fact that Server deviations are reduced to Pauli deviations (Section~\ref{subsection:blind:Pauli dev}), we identify those that actually harm the computation, and introduce a construction to detect them: it is based on the concept of \emph{traps}.
Second, in Section~\ref{subsection:verification:protocol} we introduce a trap-based verification protocol for Clifford+MSI computations, prove its composable security in implementing Resource~\ref{resource:verification} with a negligible construction error even with circuit-level noise with strength bounded by \(p_{err}\). Finally, we show that the trap design can be generalized into a generic framework for composable verification in the circuit-model in Section~\ref{subsection:framework generalized}.

\subsection{Designing traps}
\label{subsection:framework single qubit}

In this section, we first classify Pauli deviations for a computation in the Clifford+MSI model. Then we design \emph{traps}—classically simulable computations yielding deterministic measurement outcomes—to detect them. Since they will be delegated blindly, traps must be designed
for the class $\class$ of indistinguishable computations under MB-DQC with Clifford structure $\C_1, ..., \C_{t+1}$.

\subsubsection{Classification of deviations}
In Section \ref{subsection:blind:Pauli dev} we saw a strong consequence of delegating computations in $\class$ using the MB-DQC Protocol \ref{protocol:mblind_DQC}: any malicious behavior is reduced to a convex combination of Pauli deviations on $n+t$ qubits applied before the final measurements (Lemma~\ref{lemma:mblind:Pauli deviations}).
Let us analyze the impact of a single Pauli deviation $\E\in\Pauli_{n+t}$.
Intuitively, if $\E$ commutes with the measurements, it does not change the distribution of outcomes and is considered \emph{harmless}. For instance, applying a $\Z$ gate before a computational-basis measurement does not change the result. However, a deviation is problematic if it contains a bit flip ($\X$ or $\Y$) on a wire that contributes to the final result.
In this work, since we are interested in the decision bit (the first wire) and the integrity of the non-Clifford resources (the MSI wires), we first define the set of qubits with a useful measurement outcome:
\begin{definition}[Useful qubits]
    \label{verification:def:useful}
    For any Clifford+MSI computation on $n$ qubits with $t$ MSI steps, we define the set of \emph{useful qubits} as
    \begin{equation}
        Q = \{1, n+1, \dots, n+t\}\; .
    \end{equation}
    These are the qubits whose measurement outcome is relevant for the computation. This set includes the output wire (whose outcome is used as a decision bit) and the ancilla wires whose outcomes are used to perform state-injection.
\end{definition}
Let $W_{\X\Y}(\E) \subseteq \{1, \dots, n+t\}$ be the set of indices where the Pauli operator $\E$ acts as $\X$ or $\Y$.
Based on this, we define harmfulness as follows:
\begin{definition}[Harmful deviations]
    \label{verification:def:harmful}
    Let $\E\in\Pauli_{n+t}$. We say that $\E$ is \emph{harmful} if its bit-flip support intersects with the set of useful output qubits (consistent with Def. \ref{verification:def:useful}):
    $$  W_{\X\Y}(\E) \cap Q \neq \emptyset.$$
\end{definition}
This definition captures the intuition that any deviation that does not affect the output wire (index 1) or the magic state injections (indices $n+1$ to $n+t$) is effectively harmless for the verified $\BQP$ computation. Indeed, as a consequence of the Reduction to Pauli Deviation Lemma \ref{lemma:mblind:Pauli deviations}, any malicious behavior can be pushed until the end of the computation, \emph{i.e} right before the measurements. There, the individual Pauli deviations on wires $2, ..., n$ don't affect the outcomes distribution of qubit $1$ since they are local operations.

\subsubsection{Definitions of traps and detection properties}
Here, we first build traps for the Clifford+MSI model, then relate to the deviations they allow to detect.
The requirements to build a trap are: belonging to $\class$, classical simulability, and determinism in the measurement outcomes.

To be classically simulable, a trap must be in the magic-free subclass of $\class$: it must contain only stabilizer injections.
Furthermore, to force a deterministic measurement outcome, we use the stabilizer formalism as teased in the preliminaries~\ref{subsection:prelims:QC}, using the fact that in that subclass, the transformation from the input qubits to the output qubits (before measurements) is the Clifford circuit $\G \;=\; \C_{t+1}\circ \F_t \circ \C_t \circ \cdots \circ \F_1 \circ \C_1$ (see Figure~\ref{fig:blindness:G}).
Below is a definition that allows us to derive traps that make the measurement outcome of a single qubit deterministic. Consistent with Definitions~\ref{verification:def:harmful}, \ref{verification:def:useful}, we focus on trapping the qubits that have a useful measurement outcome.

\begin{definition}[Trap for a useful qubit]
  \label{def:verification:trap}
  A \emph{trap} for $\class$, parametrized by a qubit index $q \in Q$, is defined as an instance of $\class$ with $\x_0$ and $\Alabel_1, \dots, \Alabel_t$ such that $\ket{\x_0} \otimes \rho_{\Alabel_1} \otimes \dots \otimes \rho_{\Alabel_t}$ is stabilized by $\stab_q = \G^\dagger \Z_{q} \G$. We refer to $C_q$ as a trap for $q$.
\end{definition}

Using the above notation, delegating trap $C_q$ means for the Client to compute an appropriate $+1$-eigenstate, extract its classical representation in terms of $\x_0$, $\Alabel_1, ..., \Alabel_t$ and feed it to Protocol \ref{protocol:mblind_DQC}. On the other hand, to delegate the target computation on input $\x$, the Client uses the Protocol with $\x$ and $\Alabel_i=\Tlabel$ for each $i$.

The following definition captures the Pauli deviations that a trap allows to detect. Intuitively, a Pauli deviation $\E$ \emph{triggers} a trap if it flips the expected measurement outcome. Indeed, for measurements in the $\Z$ basis, the outcome is deterministically $0$ by design if $\E=\id$ or $\Z$, and is flipped if $\E=\X$ or $\Y$.

\begin{definition}[Detection of deviations]
  Let $\E\in\Pauli_{n+t}$ be a Pauli deviation, and let $C_q$ be a trap for $q\in Q$. We say that deviation $\E$ is \emph{detected} by trap $C_q$ if the deviation and $\Z_q$ anti-commute, and that it is \emph{undetected} if they commute:
  \begin{equation}
    \emph{Trap $C_q$ detects $\E$}\; \iff \{\Z_q, \E\} = 0,
    \quad
    \emph{Trap $C_q$ does not detect $\E$}\; \iff [\Z_q, \E] = 0.
  \end{equation}
\end{definition}

\subsubsection{Building a concrete set of traps}

The purpose of trap-based verification is to be able to detect any harmful Pauli deviation. Hence, verification is secure if there are enough traps to detect all the harmful deviations. In this work, we consider the set of traps
\begin{equation}
    \label{eq: set of traps}
    \mathcal Q = \{C_{i} \; \text{for}\; i\in Q\},
\end{equation}
meaning one trap for each useful output qubit.
The following lemma ensures that the set of traps detects all Server deviations.

\begin{lemma}[Harmful deviations are detected]
  \label{lemma:harmful deviations are detected}
  Let $\E\in\Pauli_{n+t}$ be a harmful deviation, and let $\mathcal Q$ be the set of traps of Equation~\ref{eq: set of traps}. Then, there is at least one trap in $\mathcal Q$ that detects $\E$.
\end{lemma}
\begin{proof}
    By Definition~\ref{verification:def:harmful}, a harmful deviation $\E$ satisfies $W_{\X\Y}(\E) \cap Q \neq \emptyset$. Let $i$ be an index in this intersection. Since $[\E]_i \in \{\X, \Y\}$, the Pauli operator anti-commutes with the measurement basis, i.e., $\{\Z_i, \E\} = 0$. Thus, by definition of detection, the trap $C_i \in \mathcal{Q}$ detects $\E$.
\end{proof}

\subsection{A trap-based verification protocol}
\label{subsection:verification:protocol}
We can finally introduce the Verified Delegated Quantum Computation Protocol~\ref{protocol:verification}.
It allows the Client to verify her $(n,t)$-Clifford+MSI computation $C$ with $\BQP$ error $\bqp$.
The verification mechanism in this protocol uses the above trap construction. It defines a trap for each qubit in $Q$---i.e.~it uses the set of traps $\mathcal Q$---thereby ensuring all harmful deviations are detected.
Then, it follows the usual Test/Computation paradigm mentioned in the Introduction, in Figure~\ref{fig:verif-canvas}. The Client samples a random partition of test and computation rounds according to chosen parameters $d$ (number of \emph{computation rounds}) and $s$ (number of \emph{test rounds}), and chooses a threshold $w$ smaller than $s$: the number of tolerated test-round failures. For each test round, she delegates a trap at random. For computation rounds, she delegates the target computation. Both can be simply expressed as using MB-DQC on a different input state and choice of state injection, both of these being perfectly hidden from the Server, by security of Protocol~\ref{protocol:mblind_DQC}. When the delegation phase is over, the Client checks the outcomes of the test rounds to see if more than $w$ test rounds failed. That is, she checks how many traps have yielded outcome $1$. It amounts to summing such events over the whole set of test rounds and rejecting if the total is greater than $w$. When the computation is accepted, the result is obtained by performing a majority vote over the first outcome of each computation round.

\begin{protocoll}
\caption{Verified Delegated Quantum Computation}
\label{protocol:verification}
\begin{algorithmic}[0]
    \State \textbf{Client inputs:} $\x\in\bin^n$, $(n,t)$-Clifford+MSI computation $C$, protocol parameters $d, s$ and $w<s$
    \Comment{Number of computation/test rounds, and threshold}
    \State \textbf{Public Information:} Clifford structure of $C$, set of traps $\mathcal Q$, and $d, s, w$.

  \Procedure{{Client - Sample random permutation of test/computation rounds.}}{}
    \State Sample a random permutation $\perm$ of $[N]$. Then for $i\in \perm$:
    \State If $i\leq d$: set $C^{(i)} \leftarrow C$. \Comment{(Computation round)}
    \State If $i> d$: sample $C_{q_i}$ randomly from $\mathcal Q$. Set $C^{(i)} \gets C_{q_i}$. \Comment{(Test round)}
  \EndProcedure

  \Procedure{{Client-Server - Delegate rounds blindly.}}{}

  \State For $i\in \perm$, delegate $C^{(i)}$ using Magic-Blind DQC Protocol~\ref{protocol:mblind_DQC}, receive $\b^{(i)}$.
    \State Set $\b = \b^{(1)} || ... || \b^{(N)}$
  \EndProcedure

  \Procedure{{Client - Traps Check}}{}
    \State Compute the number of failed test rounds $\sum_{i>d}b^{(i)}_{q_i}$ and aborts if $\geq w$. Else continue.
  \EndProcedure

  \Procedure{{Client - Majority Vote}}{}
  \State If there exists $y\in \bin$ such that $\#\{\, i \leq d \;:\; b^{(i)}_1 = y \,\} > d/2$, then output $y$.
  \State Else, repeat the protocol.
\EndProcedure

\end{algorithmic}
\end{protocoll}

Theorem~\ref{theorem:verification security} is the main result of this paper: it captures the fact that Protocol~\ref{protocol:verification} implements Resource~\ref{resource:verification} up to a negligible distance that depends on the parameters chosen by the Client, essentially the number of test and computation rounds ($s$ and $d$), while giving the range of admissible values for the threshold $w$ (in order for the protocol to be secure). Moreover, it tolerates circuit-level noise with rate $p_{err} < w/s$.

\begin{theorem}
    \label{theorem:verification security}
    Let $C$ be a $\BQP$ computation expressed in the Clifford+MSI model on $n$ qubits with $t$ MSI steps and $\BQP$ error $\bqp$.
    Let Protocol \ref{protocol:verification} be instanciated with $C$ and $d, s\in \naturals$, with $N=d+s$. Also, let  $k=|\mathcal Q|=1+t$ and $\alpha=\tfrac{1-2\bqp}{2-2\bqp}$.
    Then, for any chosen threshold $w\in\naturals$ such that $0\leq \tfrac{w}{s} < \tfrac{\alpha}{k}$,
    Protocol~\ref{protocol:verification} $\epsilon$-constructs Resource~\ref{resource:verification} in the Abstract Cryptography framework with $\epsilon = \max(\epsilon_{cor}, \epsilon_{sec})$ where $\epsilon_{cor}, \epsilon_{sec}$ are negligible in $N$.

In addition, for an honest-but-noisy Server with circuit-level noise rate $p_{err}<w/s$, the correctness error $\epsilon_{rob}$ is negligible in $N$.
\end{theorem}

Below, we proceed with a proof of correctness: it shows that when the Server is honest and perfect, no traps are triggered and there is only a negligible probability that the outcome of the majority vote differs from the output of the Ideal Resource. Then, we give a proof of robustness, analyzing the case of an honest-but-noisy Server: we similarly show that fewer than $w$ out of $s$ test rounds fail as long as $p_{err}< w/s$, which is intuitive given the notion of circuit-level noise\footnote{Note that here we voluntarily do not explore the case where the noise alters enough computation rounds---to corrupt the outcome of the majority vote---without triggering enough test rounds---to be detected. Indeed, dealing with this kind of noise is identical to dealing with a malicious adversary that wants to corrupt the outcomes while staying undetected, so it is taken care of in the security proof.}.
Finally, we prove security against arbitrarily malicious behavior from the Server that might try to corrupt the outcome of the majority vote without triggering more than $w$ test rounds.
This concludes the proof of composability and noise robustness for our verification protocol, and thus the main result of this work.

\begin{proof}[Proof of correctness]
    Assume the Server is honest. By construction of the test rounds, all traps are deterministic
    in honest executions of Protocol~\ref{protocol:mblind_DQC}, hence no trap is triggered and
    the Client reaches the output step.
    It remains to bound the probability that the Client's majority vote on the $d$ computation
    rounds outputs an incorrect decision bit. Let $z^\star \in \{0,1\}$ denote the correct
    decision bit for the delegated $\BQP$ computation $C$ (i.e., $z^\star=1$ if
    $x\in L$ and $z^\star=0$ otherwise).
    Also, for $i\leq d$ let $\b^{(i)} \in \{0,1\}^n$ be the outcome of the measurement of all qubits and $y^{(i)} = b^{(i)}_1$ be the first bit, namely the output bit of computation round $i$.
    By correctness of Protocol~\ref{protocol:mblind_DQC},
    $\b^{(i)}$ follows the distribution induced by $C$, and in particular $y^{(i)}$ follows the marginal distribution of the first qubit.
    Indeed, since $C$ decides the language with error at most $c < 1/2$,
    we have
    \[
    \Pr[y^{(i)} \neq z^\star] \le c.
    \]
    Let $Y := \sum_{i=1}^d \mathbf 1[y^{(i)} \neq z^\star]$ be the number of incorrect outputs.
    Then $Y$ is stochastically dominated by $X\sim \mathrm{Bin}(d,c)$, hence
    \[
    \Pr\!\left(Y > \tfrac d2\right) \le \Pr\!\left(X > \tfrac d2\right).
    \]
    By Hoeffding's inequality (Lemma~\ref{lemma:hoeffding:binomial}),
    \[
    \Pr\!\left(X > \tfrac d2\right)\le \exp\!\left(-2d\left(\tfrac12-c\right)^2\right).
    \]
    Therefore, conditioned on the Server being honest and the trap test passing (which occurs
    with probability~1), the Client's majority vote outputs $z^\star$ except with probability at
    most $\epsilon_{\mathrm{cor}} := \exp\!\left(-2d\left(\tfrac12-c\right)^2\right)$.
    \end{proof}

\begin{proof}[Proof of robustness.]
    In the presence of noise, test rounds might fail even for an honest Server. A rejection is obtained if more than  $w$ test rounds fail.
    Let $Y$ be the random variable that denotes the number of failed test rounds. Because it is stochastically dominated by $X\sim \mathrm{Bin}(s,p_{err})$, the above expressions can be re-used: as long as $p_{err}< \tfrac w s$, the probability that more than $w$ rounds are affected is upper-bounded by $\epsilon_{rob} :=\exp\!\left(-2(p_{err}-\tfrac w s)^2 s\right)$.

    This robustness guarantee holds for any circuit-level noise that is independent of the Client's secret parameters throughout the execution of the MB-DQC protocol. This condition is naturally satisfied for Server-side noise. For Client-side circuit-level noise, the independence from the secret needs to be assumed. In such a case, security is maintained because the Client-side noise could equally be seen as the first deviation of a malicious server. We refer the interested reader to \cite{KLMO25plugging} for managing Client-side noise that depends on the Client's secrets.

\end{proof}

\begin{proof}[Security proof]
    Here, we prove the security of Protocol~\ref{protocol:verification} by introducing Simulator~\ref{simulator:vdqc} below.
\begin{simulatorr}
    \caption{Verified Delegated Quantum Computation}
    \label{simulator:vdqc}
    \begin{algorithmic}[0]

    \Procedure{Simulator emulates Client.}{With same procedures as Protocol~\ref{protocol:verification}}
    \State Choose an arbitrary $n$-bit input string $\x_\emptyset$.
    \State Sample random permutation of test/computation rounds.
    \State Delegate rounds blindly to the Server with $\x^{(i)}=\x_\emptyset$ on computation rounds.
    \State Check Traps. Set $c=0$ if less than $w$ were triggered, $1$ else.
    \EndProcedure

    \Procedure{{Simulator - report to Ideal Resource}}{}
    \State Send $c$ to the Ideal Resource.
    \EndProcedure
    \end{algorithmic}
\end{simulatorr}     The Simulator has access to what is leaked by the protocol: the size of the input $n$, the decomposition of the target computation in Clifford gates $\C_1...\C_{t+1}$, the resulting $n+t$-Clifford map $\G$ to build the input state of traps, and $d, s, w$ the parameters of the protocol.

Its only purpose is to infer the Server's honesty as faithfully as in the Real World, encode that in a bit $c$, and send it to the Ideal Resource. In other words, the Simulator does not need to perform the computation $C$ on input $\x$. As a consequence, it does not need $\x$, and can instead start with any $n$-bit input $\x_\emptyset$. Yet, one must check that this replacement does not affect its ability to detect a malicious Server. This follows directly from the composability of Protocol~\ref{protocol:mblind_DQC}, which ensures that the cases where $\x$ or $\x_\emptyset$ are used as inputs are indistinguishable from the transcripts available to the Server.

    Below, we express the state of the total system after interaction with an arbitrarily malicious Server in Protocol~\ref{protocol:verification}. When the permutation symbol $\perm$ appears as a superscript of a given quantity, it refers to the fact that the quantity is defined for a fixed permutation $\perm$ of computation rounds and test rounds, together with a random assignment of traps to test rounds. In fact, $\perm$ describes the configuration sampled at random by the Client at the start of the protocol. With a slight abuse of notation, we will simply refer to $\perm$ as a choice of permutation, implicitly meaning a choice of permutation and trap configuration.

    \paragraph{Using the Reduction to Pauli Deviations.}
    Since Protocol~\ref{protocol:mblind_DQC} is composably secure in the AC framework, it is in particular composable in parallel. Therefore, since it is used $N$ times in parallel in Protocol~\ref{protocol:verification}, it benefits from the security properties mentioned in Section~\ref{subsection:blind:MB-DQC}, in particular the Reduction to Pauli Deviations Lemma~\ref{lemma:mblind:Pauli deviations} : for any malicious Server deviation on $N$ parallel usages of Protocol~\ref{protocol:mblind_DQC}, there exists a convex combination of Pauli operators on $N\times (n+t)$ qubits such that the total state after interaction with the Server is
    \begin{align}
        \label{eq:verif:after Pauli}
        \rho_{real}^\perm &=
        \sum_{\b}
        \sum_{\E\in\Pauli_{N\times(n+t)}}
        \abs{\alpha_{\E}}^2
        \times
            \ketbra \b
            \circ
            \E
            [\rho_{cor, \b}^{\perm}]
        \otimes \U_{\E}\left[
            \ketbra0^{\otimes w}
        \right],
    \end{align}
    where we refer to the definitions of Lemma~\ref{lemma:mblind:Pauli deviations} for
        $\rho_{cor, \b}^{\perm}
        =
        \bigotimes_{i=1}^N
        \rho_{cor, \b^{(i)}, C^{(i)}}$.

    \paragraph{Output and abort probability analysis.}
    After interaction with a malicious Server, in both worlds, the traps are checked by analyzing the outcomes of test rounds ($i>d$) of the permutation $\perm$.
    In this step, we want to express the acceptance probability.
    Note that Equation~\ref{eq:verif:after Pauli} implies that for a given permutation $\perm$ and for a fixed Pauli deviation $\E$, the distributions of outcomes on test rounds and of those on computation rounds are independent.

    We thus introduce a random variable $Y^{(\sigma)}_\E$ counting the number of failed test rounds when deviation $\E$ is applied and the choice of permutation is $\sigma$.
    In terms of that random variable, the acceptance probability for permutation $\perm$ when a fixed deviation $\E$ is applied can be expressed as
    \begin{align}
        p_{\E}^{(\perm)}
        &=
        \Pr[Y^{(\sigma)}_\E < w].
    \end{align}
    It is identical in both worlds, because of the above comment. It does not depend on the state used in computation rounds ($\rho$ for the Client in the Real World, $\rho_\emptyset$ for the Simulator in the Ideal World).

    Finally, in the Real World, a majority vote is computed on the decision bits
    $y^{(i)} := \b^{(i)}_1$ of the computation rounds $i \le d$.
    A wrong outcome occurs if the majority vote differs from the correct decision
    bit $z^\star \in \{0,1\}$.
    Let us define a random variable $Z^{(\perm)}_\E$ counting the number of such failures on computation rounds, whether they are caused by the deviation, or simply the inherent $\BQP$ error.
    In terms of this random variable, the probability that a bad result is output after a majority vote, for a given permutation $\perm$ and fixed deviation $\E$ is
    \begin{equation}
        q^{(\perm)}_{\E}
        =
        \Pr[Z^{(\sigma)}_\E > \frac d 2].
    \end{equation}

    The final output state in the Real World can thus be written as
    \begin{multline}
        \rho^{\perm}_{out,\mathrm{real}}
        =
        \sum_{\E\in\Pauli_{N\times(n+t)}}
        |\alpha_\E|^2
        \Big[
            p^{(\perm)}_{\E}\ketbra{\Acc}
            \otimes
            \big(
                q^{(\perm)}_{\E}\ketbra{z^\star\oplus 1}
                +
                (1-q^{(\perm)}_{\E})\ketbra{z^\star}
            \big)
        \\
        +
            (1-p^{(\perm)}_{\E})
            \ketbra{\Rej}\otimes\ketbra{\perp}
        \Big]
        \otimes
        \U_\E\!\left[
            \ketbra0^{\otimes w}
        \right].
    \end{multline}

    Regarding the Ideal World, there is no majority vote: conditioned on acceptance, the Resource outputs the correct decision bit $z^\star$, and otherwise outputs
    $\perp$. As already mentioned, the acceptance probability is the same, because of independence of test and computation rounds.
    Hence,
    \begin{multline}
        \rho^{\perm}_{out,\mathrm{ideal}}
        =
        \sum_{\E\in\Pauli_{N\times(n+t)}}
        |\alpha_\E|^2
        \Big[
            p^{(\perm)}_{\E}\ketbra{\Acc}\otimes\ketbra{z^\star}
            +
            (1-p^{(\perm)}_{\E})
            \ketbra{\Rej}\otimes\ketbra{\perp}
        \Big]
        \\
        \otimes
        \U_\E\!\left[
            \ketbra0^{\otimes w}
        \right].
    \end{multline}

    \paragraph{Reducing the Server deviation to a single $N\times(n+t)$-qubit Pauli.}
    Now, without loss of generality, we make the following two reductions: first, we observe that the working register is always unentangled from the rest of the state that constitutes the output of the protocol. It therefore does not contribute to any distinguishing advantage and can be traced out. Second, since the output contains a convex combination of possible Pauli deviations, it is enough to consider the case where the Server applies a single Pauli deviation $\E$ for which the distinguishing probability is maximal. Together, for a choice of permutation $\perm$ and a fixed deviation $\E$, we obtain
        \begin{multline}
            \rho_{out, real}^\perm
            =
            p_{\E}^{(\perm)}\ketbra\Acc
                \otimes\big(
                    q^{(\perm)}_{\E}\ketbra{z^\star\oplus 1}
                    +
                    (1-q^{(\perm)}_{\E})\ketbra{z^\star}
                \big)
                \\
                +
                (1-p_{\E}^{(\perm)})
                \ketbra\Rej\otimes \ketbra \perp
        \end{multline}
        and
        \begin{equation}
            \rho_{out, ideal}^\perm
            =
                p_{\E}^{(\perm)}\ketbra\Acc
                \otimes\ketbra{z^\star}
                +
                (1-p_{\E}^{(\perm)})
                \ketbra\Rej\otimes \ketbra \perp.
        \end{equation}

         \paragraph{From the point of view of the Distinguisher,} the choice of the permutation and trap configurations $\perm$ is not known. The state is thus a probabilistic mixture of all possible permutations $\perm$ uniformly. Denoting $S$ the set of permutations of $[N]$ and trap configurations, we have
        \begin{equation}
            \rho_{out, ideal}=
            \frac{1}{|S|}
            \sum_{\perm \in S} \rho_{out, ideal}^\perm
        \end{equation}
        and
        \begin{equation}
            \rho_{out, real}=
            \frac{1}{|S|}
            \sum_{\perm \in S} \rho_{out, real}^\perm.
        \end{equation}

    \paragraph{Evaluating the distinguishing advantage.}
    The distinguishing advantage is the maximum probability that the Distinguisher discriminates the correct scenario by observing the transcripts $\rho_{out, ideal}$ and $\rho_{out, real}$, with a maximization taken over the choice of cheating strategies.
    This can thus be captured by the trace distance between the two transcripts, maximized over all the possible Pauli deviations\footnote{The maximization is only over the Server choice of deviation because the working register has no impact on the output state, see the above comment. Otherwise, it would have been a maximization over the deviation and the content of the working register.}. It is indeed equivalent to the concept of the diamond norm. Hence, we can write the distinguishing advantage $p_d$ as
    \begin{align}
        p_d
        &= \max_{\E\in\Pauli_{N\times(n+t)}}
        \norm{\rho_{out, real}-\rho_{out, ideal}}_{\Tr\strut}
        \\
        &=
        \max_{\E\in\Pauli_{N\times(n+t)}}
        \frac{1}{|S|}
        \norm{\sum_{\perm \in S} \rho_{out, real}^\sigma-\rho_{out, ideal}^\sigma}_{\Tr\strut}
        \\
        \label{verif:proof:triangle ineq.}
        &\leq
        \max_{\E\in\Pauli_{N\times(n+t)}}
        \frac{1}{|S|}
        \sum_{\perm \in S}
        \norm{
             \rho_{out, real}^\sigma-\rho_{out, ideal}^\sigma}_{\Tr\strut}
        \\
        \label{verif:proof:final p_d}
        &=
        \max_{\E\in\Pauli_{N\times(n+t)}}
        \frac{1}{|S|}
        \sum_{\perm \in S}
            p_{\E}^{(\perm)}\times q_{\E}^{(\perm)},
    \end{align}
    where eq.~\ref{verif:proof:triangle ineq.} follows from the triangle inequality, and eq.~\ref{verif:proof:final p_d} follows from the definition of $p_\E^{(\perm)}$ and $q_\E^{(\perm)}$.
    We can re-write the above expression as
    $p_d = \max_\E \Pr[Y < w \wedge Z > d/2]$, where $Y$ is the random variable counting
    the number of triggered test rounds and $Z$ counts the number of computation rounds
    whose decision bit is incorrect, under a fixed deviation $\E$ and a uniformly random
    permutation $\perm$.
    Using Lemma~\ref{lemma:security error} proved in Section~\ref{section: verification proofs}, this is negligible in $s$ and $d$ as long as $\tfrac w s< \tfrac\alpha k$, where $\alpha = (1-2\bqp)/(2-2\bqp)$ and $k=|\mathcal Q|=1+t$. This proves the exponential security of the protocol.
\end{proof}

\begin{restatable}[Negligible security error]{lemma}{securityerror}
    \label{lemma:security error}
    Let $C$ be a computation with $\BQP$ error $\bqp$, and $\alpha=(1-2\bqp)/(2-2\bqp)$.
    Let $\mathcal Q$ be a set of traps that detects any harmful deviation, and let $k=|\mathcal Q|$.
    Then, using the notations of Protocol~\ref{protocol:verification}, as long as $\tfrac w s< \tfrac\alpha k$,
     for any Pauli deviation chosen by the Server the
    probability that the deviation triggers less than $w$ rounds and affects more than $d/2$ computation rounds is negligible in $d, s$.
\end{restatable}

\subsection{Towards a verification framework in the circuit-model
}
\label{subsection:framework generalized}
Note that the concepts of Section~\ref{subsection:framework single qubit} can naturally be generalized to constitute a modular framework for composable verification in the circuit-model, similar to \cite{KKLM22unifying} for the measurement-based setting. We hereby outline this generalization, which will be explored more formally and in more detail in future work.

\subsubsection{From single-qubit traps to multi-qubit traps}
A \emph{trap} in this work is simply a qubit index $q\in \mathcal Q$ aiming to yield a deterministic outcome for qubit $q$. A \emph{generalized trap} is a subset of qubit indices $Q\subset \mathcal Q$ aiming to yield a deterministic result for the parity of measurement outcomes of qubits in $Q$.
\begin{definition}[Generalized trap]
    A \emph{generalized trap} $Q$ is a subset of qubit indices $Q\subset \mathcal Q$. An input state for trap $Q$ is a $+1$ eigenstate of $\stab_Q=\G^\dagger \Z_Q\G$ where $\Z_Q = \prod_{q\in Q} \Z_{q}$.
\end{definition}

A generalized trap $Q$ detects Pauli deviations that anti-commute with $\Z_Q$.
\begin{definition}[Detection property of generalized traps]
    Using the same notations as Section~\ref{subsection:framework single qubit}, trap $Q$ detects deviation $\E$ if $\{\Z_Q, \E\}=0$, and does not detect it if $[\Z_Q, \E]=0$.
\end{definition}
Here, it appears that the traps introduced in~\ref{subsection:framework single qubit} are single-qubit traps, where $|Q|=1$.
\paragraph{Adapting the traps check.}
Following this generalization, the \textsc{Traps Check} from Protocol~\ref{protocol:verification} operation can be re-written. Indeed, introduce $\tau_Q(\b)=\bigoplus_{q\in Q} b_q$. Then, if we write $Q_i$ the trap for test round $i$, then the sum in \textsc{Traps Check} becomes $\sum_{i>d} \left(
        \tau_Q(\b^{(i)})
    \right)$.

\paragraph{Adapting the security guarantees.}
Given a set of traps, the same security properties hold if Lemma~\ref{lemma:harmful deviations are detected} can apply, meaning if any harmful deviation is detected by at least one trap of the set. We can thus state the following theorem, informal, adapted version of Theorem 6.

\begin{theorem}[Informal]
    Let $R$ be a set of generalized traps $Q_1, ..., Q_r$. Let Protocol $\ref{protocol:verification}'$ be a variant of Protocol~\ref{protocol:verification} where, in test rounds, the Client chooses a trap $Q_i$ randomly from $R$, delegates it, and uses the adapted trap-check operations above. Then, if the traps $Q_1, ..., Q_r$ detect all harmful deviations, the same security guarantees as in Theorem 6 apply.
\end{theorem}

\subsubsection{Combining compatible traps in a single test run}

Furthermore, two traps $Q_1$ and $Q_2$ can have a compatible input state, meaning a tensor product of $n+t$ single-qubit states stabilized by both $\stab_{Q_1}$ and $\stab_{Q_2}$; then the traps can be \emph{merged}: the Client delegates both traps in the same test run, and has to check that \emph{none} was triggered. Indeed, let $\rho\otimes\rho_\A$ be an input state for both traps (which exists by assumption). The Client can delegate the test run $\C_1...\C_{t+1},\A_1...\A_t$ as usual.

Note that this is not always possible. If the stabilizers of traps $Q_1$ and $Q_2$ do not allow it, then they have to be delegated in different test runs. The following definition captures trap compatibility.

\begin{definition}[Compatible traps]
    Generalized traps $Q_1, Q_2$ are compatible if their stabilizers $\stab_{Q_1}, \stab_{Q_2}$ commute on each index. Indeed, only then can there exist an input state consisting of a tensor product of $n+t$ single-qubit states stabilized by both.
\end{definition}
\paragraph{Adapting the traps check.}
Now, let $R=Q_1, ..., Q_r$ be a set of compatible traps. Then, define $\tau_R(\b)= \bigwedge_{Q\in R} \tau_Q(\b)$. The sum in the \textsc{Traps Check} operation can be written
$\sum_{i>d}
\left(
    \tau_R(\b^{(i)})
\right)
$. As a consequence, the same security properties hold if the initial set of traps (without merging) detects all harmful deviations, since such a deviation would trigger at least one of the traps.

\paragraph{Merging traps is equivalent to graph coloring.}
Quite naturally, if $Q_1$ and $Q_2$ are compatible, and $Q_2$ and $Q_3$ are compatible, then $Q_1$ and $Q_3$ are compatible, so the three traps can be done in a single test run. This merging procedure thus reduces the number of types of test runs, and finding the optimal merging procedure can be reduced to a graph coloring problem \cite{BNZ25sampling,VYI19measurement}—finding the minimum number of test runs is thus equivalent to finding the chromatic number of a graph, which is NP-hard.

\paragraph{Connection with \cite{B18how}}
This allows us to showcase \cite{B18how} as an instance of Protocol~\ref{protocol:mblind_DQC}, where the initial circuit was compiled by writing Hadamard gates as $\H=\H\T\T\H\T\T\H\T\T\H$, $\P=\T\T$. It compiles the initial circuit by introducing even more magic-state injection steps. As a direct consequence, it increases the number of qubits that the Client has to send at each round of MB-DQC. Another direct, yet naive, consequence is that it would increase the number of types of test runs. This is true if we stick to the vanilla definition of our traps in Section~\ref{subsection:framework single qubit}. But by exploiting the fact that \emph{traps can be merged}, we can show that this compilation allows the traps to be merged into \emph{only two different types of test runs}. Indeed, after Broadbent's compilation trick, only Hadamard gates and $\CNOT$ remain in the circuit.
It can be shown that the graph corresponding to the traps merging here is a bipartite graph, with a chromatic number of $2$, which explains the protocol with two types of test runs initially obtained by Broadbent in \cite{B18how}.

\section{Discussion}
\label{section:Discussion}

\paragraph{Verification in the Clifford+MSI model.}
The central contribution of this paper is to provide a framework for verification protocols that simultaneously achieves composable security, exponentially small soundness error against malicious behavior, and robustness to circuit-level noise, while being tailored to the Clifford+MSI model.
Compared to the recent result of \cite{BN25noise}, which focuses on noise robustness for verification in the circuit-model, our work additionally provides a composable and modular construction for verification protocols.
This closes the conceptual and practical gaps that had been widening between circuit- and MBQC-model-based verification strategies.

\paragraph{Constructions for composable delegation.}
Along the way, we introduced a delegation protocol that is sufficient to hide whether a computation or a test is being delegated, with full composable security inherited from its sub-components. Working in the Clifford+MSI model, we derived a blind state-injection protocol that hides the potential magic injected by the Client at each layer, together with a protocol that delegates a Clifford circuit and measurements on a blinded quantum state. Since these primitives are composably secure by construction, they can be reused in other protocols and optimized independently, both theoretically and in practical implementations.

\paragraph{A framework for verification in the circuit-model.}
By generalizing the trap design, we have shown that there exists a unified stabilizer-based formalism for circuit-model verification. It allows trap designs on subsets of qubits (rather than single qubits) and supports the combination of test rounds by merging compatible traps. This modularity exposes a whole family of verification protocols—including \cite{B18how,BN25noise}—all enjoying the same security guarantees. This has the potential to enable Clients to derive trap-based protocols optimized for noise-robustness (see paragraph below) and to tailor verification schemes to hardware constraints, in direct analogy with dummyless verification in MBQC \cite{KKLM23asymmetric}.

\paragraph{Noise-robustness and trap engineering.}
We highlight that the verification framework presented here achieves noise-robustness for circuit-level noise, adopting the same noise model as established in \cite{LMKO21verifying, BN25noise}. In our construction, the maximum tolerated noise rate is bounded by $\alpha \cdot \rate$, where $\alpha$ is a constant related to the $\BQP$ error and $\rate$ represents the \textit{trap detection rate} (i.e., the probability that a harmful Pauli-error is detected by at least one of the traps). In the current case, which is the simplest, the protocol chooses uniformly between $k$ types of test rounds, and thus this rate is $\rate = 1/k$. For the specific construction of Protocol \ref{protocol:verification}, we have $\rate = 1/(1+t)$, where $t$ is the number of state injection layers, so the maximum tolerated noise rate is bounded by $\alpha/(1+t)$. On the other hand, the protocols in \cite{BN25noise} and \cite{LMKO21verifying} use two types of test runs only, so they achieve $k=2$, meaning $\rate=1/2$, and tolerate noise up to $\alpha/2$.

This contrast suggests that noise-robustness can be significantly improved by engineering traps to increase $\rate$. While a direct way to improve the rate is to reduce the number of test types $k$---a task made possible by the flexibility of our framework as teased in Section \ref{subsection:framework generalized}---our modular approach enables even more sophisticated optimizations.
Specifically, the detection rate $\rate$ does not need to be restricted to the naive $1/k$ lower bound derived from a uniform choice of test rounds. As showcased in \cite{KKLM22unifying} for the MBQC model, within the stabilizer formalism $\rate$ can be determined through a more refined analysis of the trap-based construction's detection capabilities. By applying similar analytical tools to our circuit-model framework, one could potentially improve noise-robust verification to meet the $\rate=1/2$ detection rate of \cite{BN25noise} without relying on the compilation trick of \cite{B18how} that introduces significant ancilla overhead. As a crucial consequence, this would achieve the same level of noise-tolerance more efficiently.

\paragraph{Magic-blindness is sufficient for verification. Is it necessary?}
We established \emph{magic-blindness} as a sufficient notion of blindness for
 verification in the circuit-model. This raises a more fundamental question: \emph{must a verification protocol necessarily hide some quantum computational resource from the prover?} In our construction, the hidden resource is \emph{magic} (non-stabilizerness), because it is what enables quantum advantage in Clifford+MSI architectures, but there is no a priori reason why this must be the only possibility. This reframes verification as a \emph{resource-hiding task}: is there a different, or minimal, quantum resource whose blindness is required and sufficient to make malicious deviations detectable? If such a minimal resource exists, it would imply fundamental lower bounds on verification overhead. More generally, this question opens the door to a \emph{resource theory of verification}.

\paragraph{Acknowledgements}
Authors acknowledge Elham Kashefi, Dominik Leichtle, and Luka Music for fruitful discussions throughout the work, and Rajarsi Pal for pointing out the link between optimal traps merging procedures and graph coloring. Authors acknowledge funding from the Hybrid Quantum Initiative (HQI) supported by France 2030 under ANR grant ANR-22-PNCQ-0002.

\bibliography{qubib.bib} \bibliographystyle{linksen-sorted}

@Misc{A07scott,
  author       = {Scott Aaronson},
  howpublished = {\url{http://www.scottaaronson.com/blog/?p=284}},
  month        = {October},
  note         = {Accessed: Jan. 30 2025},
  title        = {{T}he {S}cott {A}aronson 25.00\$ {P}rize},
  year         = {2007},
  source       = {Luka Music, 2021-09-16},
}

@Article{ABE08interactive,
  author  = {Aharonov, Dorit and Ben-Or, Michael and Eban, Elad},
  journal = {arxiv:0810.5375},
  title   = {Interactive Proofs For Quantum Computations},
  year    = {2008},
  note    = {Presented at ICS 2010},
}

@Article{ABEM17interactive,
  author  = {Aharonov, Dorit and Ben-Or, Michael and Eban, Elad and Mahadev, Urmila},
  journal = {arxiv:1704.04487},
  title   = {Interactive Proofs For Quantum Computations},
  year    = {2017},
  note    = {Updated and corrected version of arxiv:0810.5375},
  url     = {https://arxiv.org/pdf/1704.04487},
}

@InProceedings{ACGH20non,
  author    = {Gorjan Alagic and Andrew M. Childs and Alex B. Grilo and Shih-Han Hung},
  booktitle = {Theory of Cryptography. TCC 2020},
  title     = {Non-interactive Classical Verification of Quantum Computation},
  year      = {2020},
  editor    = {Pass R. and Pietrzak K.},
  series    = {Lecture Notes in Computer Science},
  volume    = {12552},
  doi       = {10.1007/978-3-030-64381-2_6},
}

@Misc{AV12is,
  author = {Dorit Aharonov and Umesh Vazirani},
  note   = {Eprint:\href{http://arxiv.org/abs/1206.3686}{arXiv:1206.3686}},
  title  = {Is Quantum Mechanics Falsifiable? A computational perspective on the foundations of Quantum Mechanics},
  year   = {2012},
  source = {Luka Music, 2021-09-16},
}

@Article{B18how,
  author    = {Broadbent, Anne},
  journal   = {Theory of Computing},
  title     = {How to Verify a Quantum Computation},
  year      = {2018},
  number    = {11},
  pages     = {1--37},
  volume    = {14},
  doi       = {10.4086/toc.2018.v014a011},
  publisher = {Theory of Computing},
  url       = {https://theoryofcomputing.org/articles/v014a011},
}

@InProceedings{BFK09universal,
  author    = {Broadbent, Anne and Fitzsimons, Joseph and Kashefi, Elham},
  booktitle = {50th Annual IEEE Symposium on Foundations of Computer Science},
  title     = {Universal blind quantum computation},
  year      = {2009},
  editor    = {IEEE},
}

@Article{C05secure,
  author  = {A. Childs},
  journal = {Quant. Inf. Compt.},
  title   = {Secure assisted quantum computation},
  year    = {2005},
  number  = {6},
  pages   = {456},
  volume  = {5},
  oldkey  = {#C05#},
  source  = {Luka Music, 2021-09-16},
}

@InProceedings{DFPR14composable,
  author    = {Dunjko, Vedran and Fitzsimons, Joseph F. and Portmann, Christopher and Renner, Renato},
  booktitle = {Advances in Cryptology -- ASIACRYPT 2014},
  title     = {Composable Security of Delegated Quantum Computation},
  year      = {2014},
  address   = {Berlin, Heidelberg},
  editor    = {Sarkar, Palash and Iwata, Tetsu},
  pages     = {406--425},
  publisher = {Springer Berlin Heidelberg},
  isbn      = {978-3-662-45608-8},
  source    = {Luka Music, 2021-09-16},
}

@Misc{DNMN23verifiable,
  author        = {Drmota, P. and Nadlinger, D. P. and Main, D. and Nichol, B. C. and Ainley, E. M. and Leichtle, D. and Mantri, A. and Kashefi, E. and Srinivas, R. and Araneda, G. and Ballance, C. J. and Lucas, D. M.},
  title         = {Verifiable blind quantum computing with trapped ions and single photons},
  year          = {2023},
  archiveprefix = {arXiv},
  eprint        = {2305.02936},
  journal       = {arxiv},
  primaryclass  = {quant-ph},
}

@Article{FK17unconditionally,
  author    = {Fitzsimons, Joseph F. and Kashefi, Elham},
  journal   = {Phys. Rev. A},
  title     = {Unconditionally verifiable blind quantum computation},
  year      = {2017},
  month     = {Jul},
  pages     = {012303},
  volume    = {96},
  doi       = {10.1103/PhysRevA.96.012303},
  issue     = {1},
  numpages  = {27},
  publisher = {American Physical Society},
  source    = {Luka Music, 2021-09-16},
  url       = {https://link.aps.org/doi/10.1103/PhysRevA.96.012303},
}

@Article{KKLM22unifying,
  author  = {Kapourniotis, Theodoros and Kashefi, Elham and Leichtle, Dominik and Music, Luka and Ollivier, Harold},
  journal = {arxiv:2206.00631},
  title   = {Unifying Quantum Verification and Error-Detection: Theory and Tools for Optimisations},
  year    = {2022},
}

@Misc{KKLM23asymmetric,
  author       = {Theodoros Kapourniotis and Elham Kashefi and Dominik Leichtle and Luka Music and Harold Ollivier},
  howpublished = {Cryptology ePrint Archive, Paper 2023/379},
  title        = {Asymmetric Quantum Secure Multi-Party Computation With Weak Clients Against Dishonest Majority},
  year         = {2023},
  url          = {https://eprint.iacr.org/2023/379},
}

@Article{LMKO21verifying,
  author  = {Dominik Leichtle and Luka Music and Elham Kashefi and Harold Ollivier},
  journal = {Phys. Rev. X Quantum},
  title   = {Verifying BQP Computations on Noisy Devices with Minimal Overhead},
  year    = {2021},
  number  = {040302},
  volume  = {2},
}

@InProceedings{M12constructive,
  author    = {Maurer, Ueli},
  booktitle = {Theory of Security and Applications},
  title     = {Constructive Cryptography -- A New Paradigm for Security Definitions and Proofs},
  year      = {2012},
  address   = {Berlin, Heidelberg},
  editor    = {M{\"o}dersheim, Sebastian and Palamidessi, Catuscia},
  pages     = {33--56},
  publisher = {Springer Berlin Heidelberg},
  isbn      = {978-3-642-27375-9},
  source    = {Luka Music, 2021-09-16},
}

@InProceedings{MR11abstract,
  author       = {Maurer, Ueli and Renner, Renato},
  booktitle    = {Innovations in Computer Science},
  title        = {Abstract cryptography},
  year         = {2011},
  month        = {jan},
  organization = {Tsinghua University Press},
  pages        = {1 - 21},
  isbn         = {978-7-302-24517-9},
  source       = {Luka Music, 2021-09-16},
  url          = {https://crypto.ethz.ch/publications/files/MauRen11.pdf},
}

@Book{NC00quantum,
  author    = {Nielsen, Michael A. and Chuang, Isaac L.},
  publisher = {Cambridge University Press},
  title     = {Quantum Computation and Quantum Information: 10th Anniversary Edition},
  year      = {2000},
  doi       = {10.1017/CBO9780511976667},
  source    = {Luka Music, 2021-09-16},
}

@Article{SP00simple,
  author        = {Shor, P. W. and Preskill, J.},
  journal       = {Phys. Rev. Lett.},
  title         = {Simple proof of security of {BB84} quantum key distribution protocol},
  year          = {2000},
  pages         = {441-444},
  volume        = {85},
  eprint        = {quant-ph/0003004},
}

@Misc{V07conference,
  author = {Vazirani, Umesh},
  title  = {Conference reported in~\cite{ABE08interactive}},
  year   = {2007},
}

@Article{B15delegating,
  author        = {Broadbent, Anne},
  journal       = {Canadian Journal of Physics, 2015, 93(9): 941-946},
  title         = {Delegating Private Quantum Computations},
  year          = {2015},
  issn          = {1208-6045},
  month         = sep,
  number        = {9},
  pages         = {941--946},
  volume        = {93},
  archiveprefix = {arXiv},
  copyright     = {arXiv.org perpetual, non-exclusive license},
  date          = {2015-06-03},
  doi           = {10.1139/cjp-2015-0030},
  eprint        = {1506.01328},
  primaryclass  = {quant-ph},
  publisher     = {Canadian Science Publishing},
}

@Article{G98heisenberg,
  author        = {Gottesman, Daniel},
  journal       = {Group22: Proceedings of the XXII International Colloquium on Group Theoretical Methods in Physics, eds. S. P. Corney, R. Delbourgo, and P. D. Jarvis, pp. 32-43 (Cambridge, MA, International Press, 1999)},
  title         = {The Heisenberg Representation of Quantum Computers},
  year          = {1998},
  month         = jul,
  archiveprefix = {arXiv},
  copyright     = {Assumed arXiv.org perpetual, non-exclusive license to distribute this article for submissions made before January 2004},
  doi           = {10.48550/ARXIV.QUANT-PH/9807006},
  eprint        = {quant-ph/9807006},
  primaryclass  = {quant-ph},
  publisher     = {arXiv},
}

@Article{M18classical,
  author        = {Mahadev, Urmila},
  title         = {Classical Verification of Quantum Computations},
  year          = {2018},
  month         = apr,
  archiveprefix = {arXiv},
  copyright     = {arXiv.org perpetual, non-exclusive license},
  doi           = {10.48550/ARXIV.1804.01082},
  eprint        = {1804.01082},
  primaryclass  = {quant-ph},
  publisher     = {arXiv},
}

@Article{BKLM22succinct,
  author        = {Bartusek, James and Kalai, Yael Tauman and Lombardi, Alex and Ma, Fermi and Malavolta, Giulio and Vaikuntanathan, Vinod and Vidick, Thomas and Yang, Lisa},
  title         = {Succinct Classical Verification of Quantum Computation},
  year          = {2022},
  month         = jun,
  archiveprefix = {arXiv},
  copyright     = {Creative Commons Attribution Share Alike 4.0 International},
  doi           = {10.48550/ARXIV.2206.14929},
  eprint        = {2206.14929},
  primaryclass  = {quant-ph},
  publisher     = {arXiv},
}

@Article{KLMO24verification,
  author        = {Kashefi, Elham and Leichtle, Dominik and Music, Luka and Ollivier, Harold},
  title         = {Verification of Quantum Computations without Trusted Preparations or Measurements},
  year          = {2024},
  month         = mar,
  archiveprefix = {arXiv},
  copyright     = {arXiv.org perpetual, non-exclusive license},
  doi           = {10.48550/ARXIV.2403.10464},
  eprint        = {2403.10464},
  primaryclass  = {quant-ph},
  publisher     = {arXiv},
}

@Article{SDKO07direct,
  author        = {Marcus Silva and Vincent Danos and Elham Kashefi and Harold Ollivier},
  journal       = {New J. Phys. 9 192 (2007)},
  title         = {A direct approach to fault-tolerance in measurement-based quantum computation via teleportation},
  year          = {2007},
  issn          = {1367-2630},
  month         = jun,
  number        = {6},
  pages         = {192--192},
  volume        = {9},
  archiveprefix = {arXiv},
  date          = {2006-11-28},
  doi           = {10.1088/1367-2630/9/6/192},
  eprint        = {quant-ph/0611273},
  primaryclass  = {quant-ph},
  publisher     = {IOP Publishing},
}

@Article{BNZ25sampling,
  author        = {Barberà-Rodríguez, Júlia and Navarro, Mariana and Zambrano, Leonardo},
  title         = {Sampling Groups of Pauli Operators to Enhance Direct Fidelity Estimation},
  year          = {2025},
  month         = jan,
  archiveprefix = {arXiv},
  copyright     = {Creative Commons Attribution 4.0 International},
  doi           = {10.48550/ARXIV.2501.19228},
  eprint        = {2501.19228},
  primaryclass  = {quant-ph},
  publisher     = {arXiv},
}

@Article{VYI19measurement,
  author        = {Verteletskyi, Vladyslav and Yen, Tzu-Ching and Izmaylov, Artur F.},
  journal       = {The Journal of Chemical Physics},
  title         = {Measurement Optimization in the Variational Quantum Eigensolver Using a Minimum Clique Cover},
  year          = {2019},
  issn          = {1089-7690},
  month         = mar,
  number        = {12},
  volume        = {152},
  archiveprefix = {arXiv},
  copyright     = {arXiv.org perpetual, non-exclusive license},
  date          = {2019-07-07},
  doi           = {10.1063/1.5141458},
  eprint        = {1907.03358},
  primaryclass  = {quant-ph},
  publisher     = {AIP Publishing},
}

@Article{BK04universal,
  author        = {Bravyi, Sergei and Kitaev, Alexei},
  journal       = {Phys. Rev. A 71, 022316 (2005)},
  title         = {Universal Quantum Computation with ideal Clifford gates and noisy ancillas},
  year          = {2004},
  issn          = {1094-1622},
  month         = feb,
  number        = {2},
  pages         = {022316},
  volume        = {71},
  archiveprefix = {arXiv},
  copyright     = {Assumed arXiv.org perpetual, non-exclusive license to distribute this article for submissions made before January 2004},
  date          = {2004-03-03},
  doi           = {10.1103/physreva.71.022316},
  eprint        = {quant-ph/0403025},
  primaryclass  = {quant-ph},
  publisher     = {American Physical Society (APS)},
}

@Article{BN25noise,
  author        = {Broadbent, Anne and Nevin, Joshua},
  title         = {Noise-Robustness for Delegated Quantum Computation in the Circuit Model},
  year          = {2025},
  month         = nov,
  archiveprefix = {arXiv},
  copyright     = {Creative Commons Zero v1.0 Universal},
  doi           = {10.48550/ARXIV.2511.22844},
  eprint        = {2511.22844},
  primaryclass  = {quant-ph},
  publisher     = {arXiv},
}

@Article{GLMO25composably,
  author        = {Garnier, Maxime and Leichtle, Dominik and Music, Luka and Ollivier, Harold},
  title         = {Composably Secure Delegated Quantum Computation with Weak Coherent Pulses},
  year          = {2025},
  month         = mar,
  archiveprefix = {arXiv},
  copyright     = {arXiv.org perpetual, non-exclusive license},
  doi           = {10.48550/ARXIV.2503.08559},
  eprint        = {2503.08559},
  primaryclass  = {quant-ph},
  publisher     = {arXiv},
}

@Article{KLMO25plugging,
  author        = {Kapourniotis, Theodoros and Leichtle, Dominik and Music, Luka and Ollivier, Harold},
  title         = {Plugging Leaks in Fault-Tolerant Quantum Computation and Verification},
  year          = {2025},
  month         = oct,
  archiveprefix = {arXiv},
  copyright     = {arXiv.org perpetual, non-exclusive license},
  doi           = {10.48550/ARXIV.2510.03227},
  eprint        = {2510.03227},
  primaryclass  = {quant-ph},
  publisher     = {arXiv},
}

@Article{YKO25verifiable,
  author        = {Yang, Bo and Kashefi, Elham and Ollivier, Harold},
  title         = {Verifiable blind observable estimation: A composably secure protocol for near-term quantum advantage tasks},
  year          = {2025},
  month         = oct,
  archiveprefix = {arXiv},
  copyright     = {arXiv.org perpetual, non-exclusive license},
  doi           = {10.48550/ARXIV.2510.08548},
  eprint        = {2510.08548},
  primaryclass  = {quant-ph},
  publisher     = {arXiv},
}

@Article{GLMM24chip,
  author        = {Gustiani, Cica and Leichtle, Dominik and Mills, Daniel and Miller, Jonathan and Grassie, Ross and Kashefi, Elham},
  title         = {On-Chip Verified Quantum Computation with an Ion-Trap Quantum Processing Unit},
  year          = {2024},
  month         = oct,
  archiveprefix = {arXiv},
  copyright     = {Creative Commons Attribution 4.0 International},
  doi           = {10.48550/ARXIV.2410.24133},
  eprint        = {2410.24133},
  primaryclass  = {quant-ph},
  publisher     = {arXiv},
}

@Article{BKBH25grand,
  author        = {Babbush, Ryan and King, Robbie and Boixo, Sergio and Huggins, William and Khattar, Tanuj and Low, Guang Hao and McClean, Jarrod R. and O'Brien, Thomas and Rubin, Nicholas C.},
  title         = {The Grand Challenge of Quantum Applications},
  year          = {2025},
  month         = nov,
  archiveprefix = {arXiv},
  copyright     = {Creative Commons Attribution 4.0 International},
  doi           = {10.48550/ARXIV.2511.09124},
  eprint        = {2511.09124},
  primaryclass  = {quant-ph},
  publisher     = {arXiv},
}

@Article{DIVF24design,
  author        = {Donne, Carlo Delle and Iuliano, Mariagrazia and van der Vecht, Bart and Ferreira, Guilherme Maciel and Jirovská, Hana and van der Steenhoven, Thom and Dahlberg, Axel and Skrzypczyk, Matt and Fioretto, Dario and Teller, Markus and Filippov, Pavel and Montblanch, Alejandro Rodríguez-Pardo and Fischer, Julius and van Ommen, Benjamin and Demetriou, Nicolas and Leichtle, Dominik and Music, Luka and Ollivier, Harold and Raa, Ingmar te and Kozlowski, Wojciech and Taminiau, Tim and Pawełczak, Przemysław and Northup, Tracy and Hanson, Ronald and Wehner, Stephanie},
  title         = {Design and demonstration of an operating system for executing applications on quantum network nodes},
  year          = {2024},
  month         = jul,
  archiveprefix = {arXiv},
  copyright     = {arXiv.org perpetual, non-exclusive license},
  doi           = {10.48550/ARXIV.2407.18306},
  eprint        = {2407.18306},
  primaryclass  = {quant-ph},
  publisher     = {arXiv},
}

@Article{DCEL09exact,
  author    = {Dankert, Christoph and Cleve, Richard and Emerson, Joseph and Livine, Etera},
  journal   = {Physical Review A},
  title     = {Exact and approximate unitary 2-designs and their application to fidelity estimation},
  year      = {2009},
  issn      = {1094-1622},
  month     = jul,
  number    = {1},
  pages     = {012304},
  volume    = {80},
  doi       = {10.1103/physreva.80.012304},
  publisher = {American Physical Society (APS)},
}

@Article{BWMS25designing,
  author        = {Baranes, Gefen and Wang, Iria W. and Machado, Francisco and Suleymanzade, Aziza and Stas, Pieter-Jan and Wei, Yan-Cheng and Yelin, Susanne F. and Borregaard, Johannes and Lukin, Mikhail D.},
  title         = {Designing Fault-Tolerant Blind Quantum Computation},
  year          = {2025},
  month         = may,
  archiveprefix = {arXiv},
  copyright     = {arXiv.org perpetual, non-exclusive license},
  doi           = {10.48550/ARXIV.2505.21621},
  eprint        = {2505.21621},
  primaryclass  = {quant-ph},
  publisher     = {arXiv},
}

\appendix

\section{Proofs for Blindness}
\subsection{Correctness of Blind-State Injection}
\label{appendix: correctness of blind-gate}
\begin{proof}[Proof of correctness of Protocol~\ref{protocol:blind-gate}]
    When both parties behave honestly, we here show that the output of the protocol is identical to the one of the Ideal Resource.
    \paragraph{State received by the Server}
    Below, we write the state held by the Server after receiving encrypted qubits from the Client, including the ancilla register. The Client samples $n$-bit strings, additional encryption bits on the $n+1$-th index, and simply writes the total $n+1$-bit encryption keys $\a, \r$.
    Taking into account that the ancillary state (in register $n+1$) is pre-rotated by $\theta$, the Server receives
    \begin{align}
        \rho_{in}
        &=
        \X^{\a}
        \Z^{\r}
        \circ
        \Z_{n+1}(\theta)
        [
            \rho\otimes \rho_\A
        ]
    \end{align}
    \paragraph{State Injection by the Server.}
    The Server starts by performing the Clifford circuit $\C$ on the first $n$ qubits, and applies $\F$ (CNOT followed by a SWAP) on the total system.
    \begin{align}
        \rho_{inj} &=
        \F
            \circ
            \C
            [
                \rho_{in}
            ]
    \end{align}

    \paragraph{Measurement and Rotation.}
    The Server then measures qubit $n+1$ in the computational basis. Assuming it yielded outcome $b$, the Client computes $b'$ as mentioned in the protocol, computes the appropriate angle and blinds it, and the Server receives $\delta_{b'}$. In total, from the Server's point of view we can write that the system is in the state $\Z^{\dagger}_n(\delta_{b'})[\rho_{inj}]$ on which we apply the projector $\ketbra b$ on the $n+1$-th qubit. However since the Client stores the output bit as $b'$, the total system is actually in the state
    \begin{align}
        \rho_{dec}
        &=
        \ketbra{b'}{b}_{n+1}
        \circ
        \Z_n^\dagger(\delta_{b'})
        \left[
            \rho_{inj}
            \right]
        \end{align}

        \paragraph{Commuting the encryption.}
        We now re-write the above state by commuting the encryption.
        In the first line, we simply expand the above expression. In the second line, we commute the Pauli encryption through the Clifford circuit $\C$, then $\F$, on $n+1$ qubits. Note that this gives the Pauli $\Enc{\a'}{\r'}$ that the Client computes as well, in the protocol. Also, the $\theta$ rotation acts on the $n+1$-th bit, so it commutes trivially through $\C$ (which acts on the $n$ first qubits only). Then, through $\F$ it commutes with the $\CNOT$ since it acts on qubit $n+1$ which is the controlled qubit, and after the $\SWAP$ it ends up on the $n$-th qubit.
        This results in
        \begin{align}
            \rho_{dec}
            &=
            \ketbra{b'}{b}_{n+1}
            \circ
            \Z_n^\dagger(\delta_{b'})
            \circ
            \F
            \circ
            \C
            \circ
            \X^{\a}
            \Z^{\r}
            \circ
            \Z_{n+1}(\theta)
        [
            \rho\otimes \rho_\Apattern
        ]
        \\
            &=
            \ketbra{b'}{b}_{n+1}
            \circ
            \Z_n^\dagger(\delta_{b'})
            \circ
            \X^{\a'}
            \Z^{\r'}
            \circ
            \Z_{n}(\theta)
            \circ
            \F
            \circ
            \C
        [
            \rho\otimes \rho_\Apattern
        ]
        \end{align}
        And now, we are interested in commuting the encryption through the $\delta$ rotation on the $n$-th qubit. In the first line we just use the fact that this rotation acts on the $n$-qubit only, so the encryption on the other qubits commute trivially. On the second line, we use the fact that $\Z$ commutes through $\Z$-rotation and $\X$ commute up to a sign flip of the angle, and in the third line we just combine the operators again to have a compact notation. Finally in the last three lines we use the fact that $\Z(\alpha)\circ \Z(\beta) = \Z(\alpha+\beta)$, and replaced $\delta$ by its definition, which cancels out the $(-1)^{a_n'}$ sign introduced by commuting the encryption.
        \begin{align}
            \Z_n^{\dagger}(\delta) \circ \X^{\a'}\Z^{\r'} \circ \Z_{n}(\theta)
            &=
            \left(\bigotimes_{i\neq n}\X^{a'_i}\Z^{r'_i}\right)
            \otimes
            \left(\Z^{\dagger}(\delta) \circ \X^{a_n'}\Z^{r_n'}\right)
            \circ \Z_{n}(\theta)
            \\
            &=
            \left(\bigotimes_{i\neq n}\X^{a'_i}\Z^{r'_i}\right)
            \otimes
            \left(
                \X^{a_n'}\Z^{r_n'}
                \circ
                \Z^{\dagger}((-1)^{a_n'}\delta)
            \right)
            \circ \Z_{n}(\theta)
            \\
            &=
            \Enc{\a'}{\r'}
            \circ
            \Z_n^{\dagger}((-1)^{a_n'}\delta)
            \circ \Z_{n}(\theta)
            \\
            &=
            \Enc{\a'}{\r'}
            \circ
            \Z_n^{\dagger}((-1)^{a'_n}\times(-1)^{a_n'}(\phi_{b'}+\theta)-\theta)
            \\
            &=
            \Enc{\a'}{\r'}
            \circ
            \Z_n^{\dagger}(\phi_{b'})
        \end{align}

        Altogether, we can thus write
        \begin{align}
            \rho_{dec}
            &=
            \ketbra{b'}{b}_{n+1}
            \circ
            \X^{\a'}
            \Z^{\r'}
            \circ
            \Z_n^\dagger(\phi_{b'})
            \circ
            \F
            \circ
            \C
        [
            \rho\otimes \rho_\Apattern
        ]\; .
        \end{align}
        The next re-writing step is to interpret the Client's decoding of the measurement outcome as a decryption of the quantum state before measurement. This can be done by changing variables for $b$, as $b\mapsto b\oplus a_{n+1}'$. The consequence is that $b'$ becomes $b$ and $\ketbra{b'}{b}$ becomes $\ketbra{b}{b'} = \ketbra{b} \circ \X^{a_{n+1}'}$: flipping the measurement outcome is equivalent to applying a Pauli $\X$ upon measurement in computational basis. Furthermore, for the sake of homogeneity, we note that we can write it as a full decryption $\Dec{a'_{n+1}}{r'_{n+1}}$ as $\ketbra{b} \circ \X^{a_{n+1}'} = \ketbra{b} \circ \Z^{r'_{n+1}}\X^{a_{n+1}'}$ since up to a global and irrelevant phase, $\bra b \Z^r = \bra b$. Thus, the state of the total system, before decryption by the Client, is
        \begin{align}
            \rho_{dec, b}
            &=
            \left(
            \ketbra{b}{b}
            \circ
                \Dec{a'_{n+1}}{r'_{n+1}}\right)_{n+1}
            \circ
            \Enc{\a'}{\r'}
            \circ
            \Z_n^\dagger(\phi_{b})
            \circ
            \F
            \circ
            \C
        [
            \rho\otimes \rho_\A
        ]
    \end{align}
    \paragraph{Client decryption.}
    Finally, the Client applies $\Dec{\tilde\a'}{\tilde\r'}$ defined as $\Dec{\a'}{\r'}$ if $\A\neq \Tlabel$, and $\X^{b}_n\circ\Dec{\a'}{\r'}$ if $\A=\Tlabel$. In both cases, the additional correction is Pauli and acts on the $n$-th qubit, so it commutes with $\ketbra b_{n+1}$.
    \begin{itemize}
        \item In the $\A\neq \Tlabel$ case, then it is straightforward that
        \begin{align}
            \rho_{out, b}
            &=
            \ketbra b_{n+1}
            \circ
            \underbrace{(\bigotimes_{i\leq n}\Dec{a'_i}{r'_i})
            \circ
            (\Dec{a'_{n+1}}{r'_{n+1}})_{n+1}}_{\Dec{\a'}{\r'}}
            \circ
            \Enc{\a'}{\r'}
            \circ
            \Z_n^\dagger(\phi_{b})
            \circ
            \F
            \circ
            \C
        [
            \rho\otimes \rho_\A
        ]
        \nonumber
        \\
        &=
            \ketbra b_{n+1}
            \circ
            \underbrace{\Dec{\a'}{\r'}
            \circ
            \Enc{\a'}{\r'}}_{\id_{n+1}}
            \circ
            \Z_n^\dagger(\phi_{b})
            \circ
            \F
            \circ
            \C
        [
            \rho\otimes \rho_\A
        ]
        \\
        &=
            \ketbra b_{n+1}
            \circ
            \Z_n^\dagger(\phi_{b})
            \circ
            \F
            \circ
            \C
        [
            \rho\otimes \rho_\A
        ]
        \\
        &=
            (\id_n\otimes \ketbra b)
            \circ
            \F
            \circ
            \C
        [
            \rho\otimes \rho_\A
        ]
    \end{align}
    which is the output of the Ideal Resource.
    \item In the $\A=\Tlabel$ case, we have something similar, except the Pauli correction is being performed by the Client as part of the decryption process. Hence the net transformation that will be implemented on $\rho$ is $\T_n\circ \C$. Indeed, for the same reason as above but introducing the required correction, by replacing $\phi_b=b\times\tfrac\pi2$ (which is the appropriate value for $\A=\Tlabel$):
        \begin{align}
            \rho_{out, b}
            &=
            \ketbra b_{n+1}
            \circ
            \X^{b}_{n}
            \circ
            \Z_n^\dagger({b}\times \tfrac{\pi}{2})
            \circ
            \F
            \circ
            \C
        [
            \rho\otimes \rho_\A
        ]
        \\
        &=
            \T_n\circ\C[\rho]
            \otimes
            \ketbra b
    \end{align}
    \end{itemize}

    Since the Client's output of the protocol is the $n$-qubit state $\Apattern_n
        [\rho]$, this concludes the proof of correctness, as this is the output of the Ideal Resource.
\end{proof}

 \subsection{Security of Blind Measurements Protocol}
\label{section:blindmeas-secproof-formal}
\begin{proof} Here, we prove that Simulator~\ref{simulator:blind-meas} allows one to reproduce any Server deviation by requiring knowledge of the Clifford circuit $\C$ only, not the $n$-qubit input state $\rho$.
\paragraph{State sent to the Distinguisher.}
Unsurprisingly, it emulates the application of a quantum one-time pad by teleportation, without learning the Client's inputs. Instead, it creates $n$ EPR-pairs, gives halves to the Server interface of the Distinguisher, the other halves to the Resource, with the instruction to perform a Bell measurement on them and the Client's input state $\rho$, as per the EPR-encryption explained in the preliminaries. The result is that once the measurements are done, the Resource holds the outcomes $\a, \r$ interpreted as keys of the one-time pad, while the Server holds the state
\begin{equation}
    \rho^{\a, \r}_{ideal, in} = \Enc{\a}{\r}[\rho] = \rho^{\a, \r}_{real, in}\; .
\end{equation}
Intuitively, from then on, we will show that any malicious Server interacts with the received state identically in both worlds, hence the proof can be concluded by showing that the received states are themselves initially identical, as we have just shown. The rest of the proof will formalize this.

\paragraph{State after interaction with a malicious Server.}
Without loss of generality, the Server's behavior is a CPTP map on $n$ qubits, followed by computational basis measurements yielding outcomes $\z$. The CPTP map can be written as a unitary $\D$ on the $n$ qubits and a working register of fixed size $w$, initialized at $\ketbra0^{\otimes w}$.
In the Ideal World, the bitstring is returned to the Ideal Resource, which sets $\z\oplus\a$ as the output of the protocol. In the Real World, the Client sets the exact same output.
In total, the system's state can thus be written identically in both worlds as
\begin{equation}
    \rho^{\a, \r}_{real, out, \z\oplus\a}=
    \rho^{\a, \r}_{ideal, out, \z\oplus\a}=
    \rho^{\a, \r}_{out, \z\oplus\a}=
    \left( \ket{\z\oplus\a}\bra{\z}\otimes \id_w \right) \circ \D \circ \left( \Enc{\a}{\r} \otimes \id_w \right)[\rho\otimes \ketbra0^{\otimes w}]
\end{equation}
We apply the same change of variables as in the proof of correctness, to obtain
\begin{align}
    \rho^{\a, \r}_{out, \z}
    &=
    \left( \ket{\z}\bra{\z\oplus\a}\otimes \id_w \right)\circ \D \circ \left( \Enc{\a}{\r} \otimes \id_w \right)[\rho\otimes \ketbra0^{\otimes w}]
    \\
    &=\left( \ket{\z}\bra{\z} \otimes \id_w \right)\circ \left( \X^\a \otimes \id_w \right)\circ \D \circ \left( \Enc{\a}{\r} \otimes \id_w \right)[\rho\otimes \ketbra0^{\otimes w}]
    \\
    &=\left( \ket{\z}\bra{\z} \otimes \id_w \right)\circ \left( \Dec{\a}{\r} \otimes \id_w \right)\circ \D \circ \left( \Enc{\a}{\r} \otimes \id_w \right)[\rho\otimes \ketbra0^{\otimes w}]
\end{align}
where in the third line we used the fact that $\bra z \Z^r = \bra z$ up to an irrelevant global sign for any $z, r \in\{0,1\}$.

\paragraph{From the point of view of the Distinguisher.}
The values of $\a, \r$ are unknown to the Distinguisher, and they follow a uniform distribution. Hence the Distinguisher only knows a statistical mixture of all possible values, identical in both worlds:
\begin{equation}
    \rho_{out, \z} =
    \dfrac{1}{4^n}
    \sum_{\a, \r\in\{0,1\}^n}
    \rho^{\a, \r}_{out, \z}
    =
    \dfrac{1}{4^n}
    \left( \ket{\z}\bra{\z} \otimes \id_w \right)
    \left[ \sum_{\a, \r\in\{0,1\}^n}
    \left( \Dec{\a}{\r} \otimes \id_w \right)\circ \D \circ \left( \Enc{\a}{\r} \otimes \id_w \right)[\rho\otimes \ketbra0^{\otimes w}] \right]
\end{equation}
From this we can deduce that $\rho_{ideal, out, \z} = \rho_{real, out, \z}$ and so $\norm{\rho_{ideal, out, \z}-\rho_{real, out, \z}}_{\Tr\strut} = 0$, which concludes the proof.

\end{proof}

\subsection{Proof of Reduction to Pauli Deviations Lemma}
\label{subsection:proof of reduction to Pauli Dev}
In this section, we aim to prove the important Reduction to Pauli Deviations lemma, that we hereby re-state.
\paulidev*

\begin{proof}
    In this proof, we aim to show that the total system can be described by the state formulated in the lemma,
     after interacting with any unbounded Server. In what follows, $\rho$ represents the classical input $\x$ encoded in a quantum state.

    \paragraph{State sent to the Server.}
    We start by writing the total state sent to the Server alongside computation branch $\b'$. This includes the initial one-time-pad of all qubits with keys $\a, \r$ and secret pre-rotation of ancillas with angles $\boldsymbol{\theta}=(\theta_1, ..., \theta_t)$, as well as the classical registers carrying rotation angles denoted by the $2\times t$-bit string $\boldsymbol{\delta}_{\b'}$. Let $\rho\otimes \rho_\A$ denote the input state tensored with the ancilla qubits, and without loss of generality, let $\ketbra0^{\otimes w}$ be the initial state of a private working register held by the Server. From the point of view of the Client having sampled secrets $\a, \r, \boldsymbol{\theta}$, the description of the total state held by the Server is thus
    \begin{equation}
        \rho_{in, \b'}^{\a, \r, \boldsymbol{\theta}}
        =
        \Enc{\a}{\r}
            \circ
            \Z(\boldsymbol\theta)
            \left[\rho
            \otimes
            \rho_\A
            \otimes
            \ketbra{\boldsymbol\delta_{b'}}
            \otimes
            \ketbra{0}^{\otimes w}\right]\; .
    \end{equation}

    \paragraph{After interaction with a malicious Server.}
    An arbitrary malicious Server can be modelled as performing unitary deviations on the computations qubits and some internal work register at each layer.
    We now use the same trick as \cite{FK17unconditionally,KKLM22unifying}: use the fact that the computation branch is fixed to $\b'$, so we can model $\Z^\dagger(\delta)$ at each layer as classically-controlled gates, controlled by the angles register containing two qubits per layer, sent by the Client. Then, in all generality, an interaction with a malicious Server, on the first layer, writes as follows:
    \begin{itemize}
        \item The Server might apply a deviation $\D_{in}$ when receiving the $n$ qubits and the first ancilla.
        \item Then, we can always assume that the Server performs
        $\F_1\circ \C_1$ followed by a unitary deviation $\D_{pre}$.
        \item Finally we can assume that the Server performs a rotation on the $n$-th qubit, that is classically conditioned by the angles register, which we call $\CR^\dagger$, followed by a final unitary deviation $\D_{fin}$.
    \end{itemize}
        Since they are all unitary operators, we can gather them as one global unitary $\D_1$ applied before the measurement of the ancilla, for the first layer. This allows to then replace the honest $\CR^\dagger$ by a $\Z$ rotation with the first angle sent by the Client.
    \begin{align}
        &=
        \D_{fin}
        \circ
        \CR^\dagger
        \circ
        \D_{pre}
        \circ
        \F_1
        \circ
        \C_1
        \circ
        \D_{in}
        \left[
            \rho_{in, \b'}^{\a, \r, \boldsymbol{\theta}}
            \otimes
            \ketbra0^{\otimes w}
         \right]
         \\
         &=
         \D_1
         \circ
        \CR^\dagger
        \circ
        \F_1
        \circ
        \C_1
        \left[
            \rho_{in, \b'}^{\a, \r, \boldsymbol{\theta}}
            \otimes
            \ketbra0^{\otimes w}
         \right]
         \label{eq:unitarization:avant}
         \\
         &=
         \D_1
         \circ
        \Z^\dagger_n\left( \delta_{\b'}^{(1)} \right)
        \circ
        \F_1
        \circ
        \C_1
        \left[
            \rho_{in, \b'}^{\a, \r, \boldsymbol{\theta}}
            \otimes
            \ketbra0^{\otimes w}
         \right]
         \label{eq:unitarization:fin}
    \end{align}
    where we have set
        \begin{equation}
        \D_1 = \D_{fin}\circ \CR^\dagger \circ \D_{pre}\circ \F \circ \C \circ \D_{in}\circ \C^\dagger \circ \F^\dagger \circ \CR\; .
    \end{equation}
    This could not have been done without unitarizing the protocol, since otherwise commuting the deviation would have carried a dependency on the rotation angle.

    We can apply the same reasoning for all the layers. Overall, this trick allows to extract the honest unitary part from the deviation, followed by a pure deviation term that we write $\D$. The resulting state is represented in Figure~\ref{fig:mb-DQC-secproof-before commutation} for the case $n=3, t=2$, and can thus be expressed as
    \begin{multline}
        \rho_{out, \b'}^{\a, \r, \boldsymbol{\theta}}
        =
       \ketbra{\b'}{\b}
        \circ
        \D
        \circ
        \C_{t+1}
        \circ
        \Z^\dagger_n(\delta^{(t)}_{\b'})
            \circ \F_{t}
            \circ \C_{t}
        \circ
        \cdots
        \circ
        \Z^\dagger_n(\delta^{(1)}_{\b'})
            \circ \F_{1}
            \circ \C_{1}
            \left[
                \rho_{in, \b'}^{\a, \r, \boldsymbol{\theta}}
            \right]\; .
        \nonumber
    \end{multline}
    \begin{figure}[h]
        \centering
        \includegraphics[width=\linewidth]{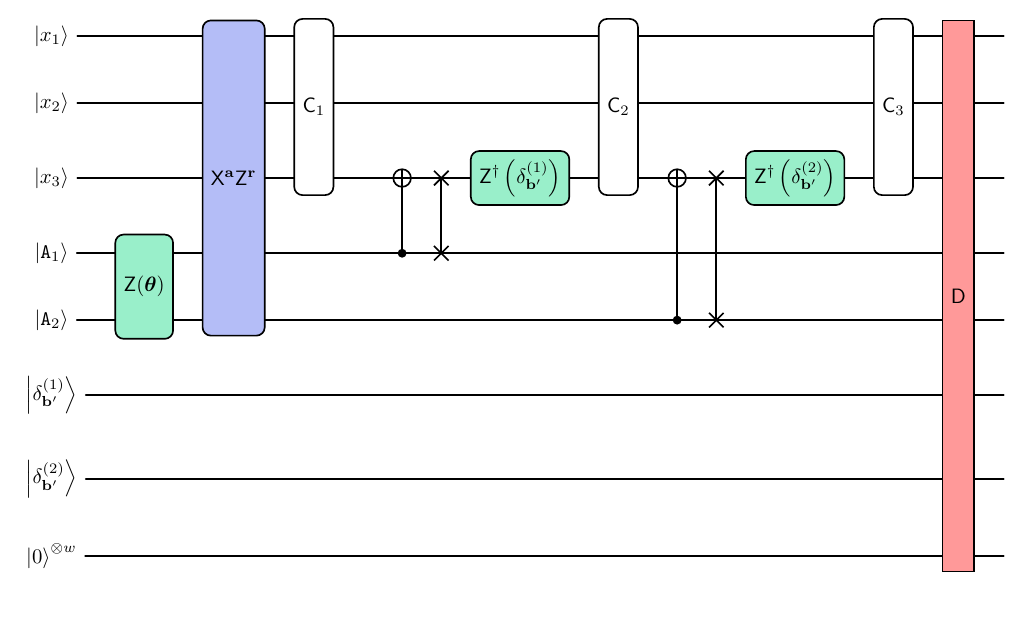}
        \caption{Representation of the state after interaction with a malicious Server, alongside the computation branch $\b'$, on a computation where $n=3, t=2$. The $\ketbra{\b'}{\b}$ operation is to be applied on the output wires.
        The protocol was unitarized and any malicious behavior can be written as a unitary operator $\D$ (red box) performed by the Server after executing the protocol honestly.
        This unitarization has transformed the classically-controlled rotations of Equation~\ref{eq:unitarization:avant} into simple $\Z^\dagger(\delta)$ rotations (green boxes in the circuit and see Equation~\ref{eq:unitarization:fin}) while the deviation is applied on the angles sent by the Client at the end of the honest execution.
        In this figure, the Client has sent the input and the ancillas $\rho\otimes \rho_{\A_1}\otimes... \otimes \rho_{\A_t}$ encrypted with a Pauli operator (blue box) and a pre-rotation on the ancillas (green box on the left), which later gets compensated during the $\Z^\dagger(\delta)$ rotations (green boxes in the middle).}
        \label{fig:mb-DQC-secproof-before commutation}
    \end{figure}

    \paragraph{Commuting the encryption.}
    The next step is to commute the initial Pauli encryption through the honest, unitary sequence of Clifford and state-injection layers. Similarly, as in the proof of correctness of Protocol~\ref{protocol:blind-gate}, after each layer the Pauli encryption is mapped to another Pauli encryption (a bijective mapping since the operations are Clifford). Also, the $\Z(\theta)$ pre-rotation can be commuted and gets cancelled in the $\Z(\delta_{\b'})$ rotation that becomes $\Z(\phi_{\b'})$, \emph{i.e.}, the correct rotation alongside branch $\b'$. This can be applied to all layers successively: the state of the total system before deviation remains encrypted.
    If we write the final encryption keys as $\a', \r'$ that, given $\b'$, are a deterministic Clifford mapping from $\a, \r$, the result can be depicted on Figure~\ref{fig:mb-DQC-secproof-after commutation}, and the state can be written as
    \begin{equation}
        \rho_{out, \b'}^{\a, \r, \boldsymbol{\theta}}
        =
        \ketbra{\b'}{\b}
        \circ
        \D
        \circ \Enc{\a'}{\r'}
        \left[
            \rho_{cor, \b', C}
            \otimes
            \ketbra{\boldsymbol\delta_{b'}}
            \otimes
            \ketbra{0}^{\otimes w}
        \right]
    \end{equation}
    where
         $\rho_{cor, \b, C} =
        \C_{t+1}
        \circ
        \Z^\dagger_n(\phi^{(t)}_{\b})
                \circ \F_{t}
                \circ \C_{t}
            \circ
            \cdots
            \circ
            \Z^\dagger_n(\phi^{(1)}_{\b})
                \circ \F_{1}
                \circ \C_{1}
        [\rho\otimes \rho_\A]
        $ is the correct state after honest Clifford and state injection layers for computation $C$ alongside computation branch $\b$ (used for branch $\b'$ in the above, since the Client stores $\b'$).

            \begin{figure}[h]
        \centering
        \includegraphics[width=\linewidth]{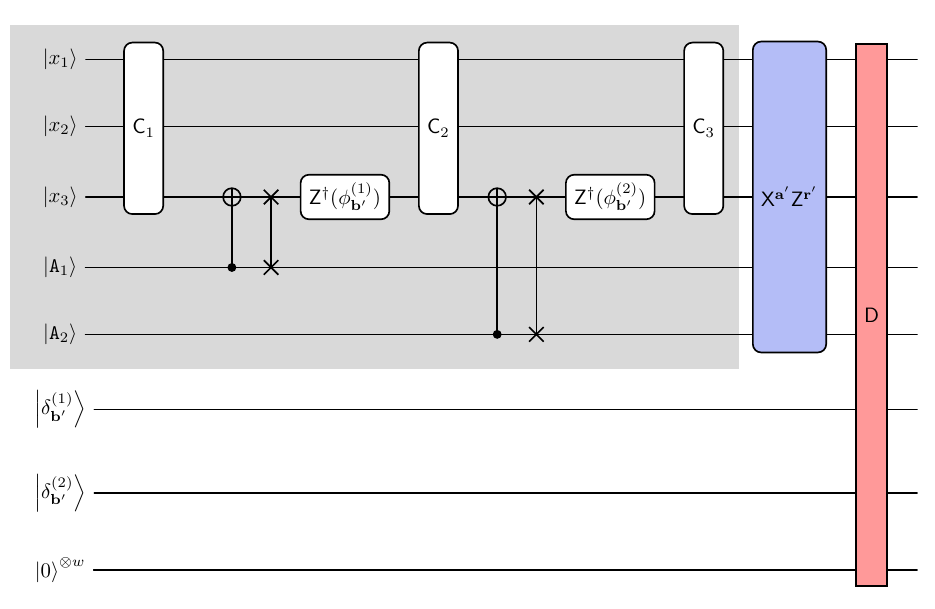}
        \caption{Same as Figure~\ref{fig:mb-DQC-secproof-before commutation}, but with a commuted encryption. On each layer, the $\Z({\theta})$ encryption is absorbed into $\Z^\dagger(\delta)$ to turn it into $\Z(\phi)$, so the green boxes of Figure~\ref{fig:mb-DQC-secproof-before commutation} disappear. The Pauli encryption $\Enc{\a}{\r}$ is commuted through the entire Clifford circuit and results in $\Enc{\a'}{\r'}$ (blue box in the figure). The state in the gray box is $\rho_{cor, \b', C}$, \emph{i.e.}, the correct state after honest execution of the protocol alongside branch $\b'$.}
        \label{fig:mb-DQC-secproof-after commutation}
    \end{figure}

    Then, we show that decoding measurement outcomes can be seen as applying a decryption operator before the measurements, and that is the conjugate of the encryption operator. This is straightforward: the Client's decoding corresponds to storing $\b'=\b\oplus \a'$, since indeed at each step, the Client decrypts the measurement outcomes with the $\X$-factor of the encryption. On the level of the entire computation, it amounts to storing $\b$ as $\b'=\b\oplus \a'$. Hence, it is equivalent to applying a decryption $\X^{\a'}$ before the $\Z$-basis measurement, which is equivalent to a $\Dec{\a'}{\r'}$ decryption since the measurement basis is invariant under $\Z$. Formally, let the following change of variables: $\b'' = \b\oplus \a'$, so that $\b$ becomes $\b''\oplus \a'$, and $\b'$ becomes $\b''$. Then, re-label $\b''$ as $\b$. We thus have:
        \begin{equation}
            \rho_{out, \b}^{\a, \r, \boldsymbol{\theta}}
            =
            \ketbra\b
            \circ
            \Dec{\a'}{\r'}
            \circ
            \D
            \circ
            \Enc{\a'}{\r'}
            \left[
                \rho_{cor, \b, C}
                \otimes
            \ketbra{\boldsymbol\delta_{b}}
            \otimes
            \ketbra{0}^{\otimes w}
            \right]
        \end{equation}

    \paragraph{From the point of view of the Distinguisher.}
    From the point of view of any unbounded Distinguisher, the keys $\a, \r$ are unknown, so the resulting state is a statistical mixture of all the possible values of $\a, \r$ over $\bin^{n+t}$ since they follow a uniform distribution. The updated keys $\a', \r'$ do as well since at each layer they are updated according to a Clifford transformation, which is a bijection. The result is that the sum can be similarly taken over $\a', \r'\in\bin^{n+t}$, which we re-label $\a, \r$ for simplicity. Furthermore, the angles $\theta_i$ for $i\leq t$ are not known neither, so a sum must be taken on the possible values of $\boldsymbol{\theta}$ over $\Theta^{t}$.
        We get:
        \begin{align}
            \rho_{out, \b}&=
            \dfrac{1}{4^{n+t}}
            \sum_{\a, \r\in \bin^{n+t}}
            \dfrac{1}{2^{2t}}
            \sum_{\boldsymbol{\theta}}
            \ketbra\b
            \circ
            \Dec{\a}{\r}
            \circ
            \D
            \circ
            \Enc{\a}{\r}
            \left[
                \rho_{cor, \b, C}
                \otimes
            \ketbra{\boldsymbol\delta_{b}}
            \otimes
            \ketbra{0}^{\otimes w}
            \right]
        \end{align}
        Now we notice that the values of $\theta$ only appears in the angles register. Since their distribution is uniform, they perfectly one-time pad that register and yield the maximally mixed state on $2$ qubits for each layer so $2t$ qubits in total, which can be discarded from the state since it contains no useful information. We organize the above sum to make that clear
        \begin{align}
            \rho_{out, \b}
            &=
            \dfrac{1}{4^{n+t}}
            \ketbra\b
            \left[
                \sum_{\a, \r\in \bin^{n+t}}
                \Dec{\a}{\r}
                \circ
                \D
                \circ
                \Enc{\a}{\r}
                \left[
                    \rho_{cor, \b, C}
                    \otimes
                \left( \dfrac{1}{2^{2t}}
                \sum_{\boldsymbol{\theta}}
                \ketbra{\boldsymbol\delta_{b}} \right)
                \otimes
                \ketbra{0}^{\otimes w}
                \right]
            \right]
            \\
            &=
            \dfrac{1}{4^{n+t}}
            \ketbra\b
            \left[
                \sum_{\a, \r\in \bin^{n+t}}
                \Dec{\a}{\r}
                \circ
                \D
                \circ
                \Enc{\a}{\r}
                \left[
                    \rho_{cor, \b, C}
                    \otimes
                \left( \dfrac{\id_{2t}}{2^{2t}}
                 \right)
                \otimes
                \ketbra{0}^{\otimes w}
                \right]
            \right]
            \\
            \label{eq:paulidev:beforePOV}
            &=
            \dfrac{1}{4^{n+t}}
            \ketbra\b
            \left[
                \sum_{\a, \r\in \bin^{n+t}}
                \Dec{\a}{\r}
                \circ
                \D
                \circ
                \Enc{\a}{\r}
                \left[
                    \rho_{cor, \b, C}
                \otimes
                \ketbra{0}^{\otimes w}
                \right]
            \right]
        \end{align}

        Without loss of generality, we can always decompose $\D$ in the $n+t$-qubit Pauli basis as $\D=\sum_{\E\in\Pauli_{n+t}}\alpha_\E \E\otimes \U_\E$ where $\E$ acts on the $n+t$-qubit system sent by the Client while $\U_\E$ acts on the rest of the system (meaning the working register, since the angle register has been traced out). This choice is convenient because it is on this $n+t$ qubit subsystem that the encryption and decryption occurs, and hence the Pauli Twirling lemma will apply. Indeed, the following simplification can now occur:
        \begin{align}
            \label{eq:twirl begin}
            &\dfrac{1}{4^{n+t}}\sum_{\a, \r\in\bin^{n+t}}
                \Dec{\a}{\r}
                \circ
                \D
                \circ
                \Enc{\a}{\r}
                \left[
                    \rho_{cor, \b, C}
                \otimes
                \ketbra{0}^{\otimes w}
                \right]
                \\
                &=
            \dfrac{1}{4^{n+t}}\sum_{\a, \r\in\bin^{n+t}}
                \Dec{\a}{\r}
                \circ
                \left(
                    \sum_{\E\in\Pauli_{n+t}}\alpha_\E \E\otimes \U_\E
                 \right)
                \circ
                \Enc{\a}{\r}
                \left[
                    \rho_{cor, \b, C}
                \otimes
                \ketbra{0}^{\otimes w}
                \right]
                \\
                &=
                \dfrac{1}{4^{n+t}}\sum_{\E, \E'\in \Pauli_{n+t}}
                \alpha_\E
                \alpha^*_{\E'}
                \sum_{\Q \in \Pauli_{n+t}}
                \Q^\dagger
                \E
                \Q
                \;
                \rho_{cor, b}
                \;
                \Q^\dagger
                \E^{\prime \dagger}
                \Q
                \otimes \U_{\E}\ketbra0^{\otimes w}\U^\dagger_{\E'}
                \\
                &=
                \dfrac{1}{4^{n+t}}\sum_{\E\in \Pauli_{n+t}}
                \abs*{\alpha_\E}^2
                \sum_{\Q \in \Pauli_{n+t}}
                \Q^\dagger
                \E
                \Q
                \;
                \rho_{cor, b}
                \;
                \Q^\dagger
                \E^{\dagger}
                \Q
                \otimes \U_{\E}\ketbra0^{\otimes w}\U^\dagger_{\E}
                \\
                \label{eq:twirl end}
                &=
                \dfrac{1}{4^{n+t}}\sum_{\E\in \Pauli_{n+t}}
                \abs*{\alpha_\E}^2
                \sum_{\Q \in \Pauli_{n+t}}
                \Q^\dagger
                \circ
                \E
                \circ
                \Q[\rho_{cor, b}]
                \otimes \U_\E[
                    \ketbra0^{\otimes w}
                ]
                \\
                \label{eq:encryption commute}
                &=
                \sum_{\E\in \Pauli_{n+t}}
                \abs*{\alpha_\E}^2
                \E
                [\rho_{cor, b}]
                \otimes \U_\E[
                    \ketbra0^{\otimes w}
                ]
        \end{align}
        where in Equations~\ref{eq:twirl begin} to~\ref{eq:twirl end} we simply applied the Pauli Twirl Lemma~\ref{lemma:Pauli Twirl}, and to obtain Equation~\ref{eq:encryption commute} we used the fact that Pauli operators commute up to an irrelevant global phase.
        Hence,
        \begin{equation}
            \rho_{out, \b} =
            \sum_{\E\in \Pauli_{n+t}}
            \abs*{\alpha_\E}^2
            \ketbra \b
            \circ
            \E
            [\rho_{cor, b}]
            \otimes \U_\E[
                \ketbra0^{\otimes w}
            ]
        \end{equation}

        Averaging over the all the computation branches yields the result stated in the initial lemma:
        \begin{align}
            \rho_{out}
        =
            \dfrac{1}{2^{n+t}}
            \sum_{\b\in\bin^{n+t}}
            \rho_{out, \b}
        \end{align}

\end{proof}
 \subsection{Security of Magic-Blind Delegated Quantum Computation}
\label{subsection:security proof MB-DQC}

\begin{proof}
    The security proof of Magic-Blind DQC in the Abstract Cryptography framework amounts to proving indistinguishability between the Real World and the Ideal World. In this proof, we analyze the transcript in the Real World, and present a Simulator that, once plugged into the Resource, allows one to generate the same transcript, concluding the proof. In what follows, let the inputs (computation and input state) chosen by the Distinguisher be $C$ and $\rho$.

    \paragraph{In the Real World.}
    We start with the transcript in the Real World. Here, we can use the work done in the proof of Lemma~\ref{lemma:mblind:Pauli deviations}, \textit{i.e.}, Section~\ref{subsection:proof of reduction to Pauli Dev}, which analyzes the evolution of the state when interacting with a malicious Server in Protocol~\ref{protocol:mblind_DQC}. Namely, after averaging over the possible secrets, using Equation~\ref{eq:paulidev:beforePOV} we can write the transcript in the Real World as
    \begin{equation}
        \rho_{real, out, \b}
        =
        \dfrac{1}{4^{n+t}}
            \ketbra\b
            \left[
                \sum_{\a, \r\in \bin^{n+t}}
                \Dec{\a}{\r}
                \circ
                \D
                \circ
                \Enc{\a}{\r}
                \left[
                    \rho_{cor, \b, C}
                \otimes
                \ketbra{0}^{\otimes w}
                \right]
            \right]\; .
    \end{equation}

    \paragraph{In the Ideal World.}
    We now present a Simulator that interacts with the malicious Server and Resource~\ref{resource:mblind_DQC}, and prove that it generates the same transcript as in the Real World. To do so, we proceed by reduction, as in the security proof of the Blind State Injection Protocol~\ref{protocol:blind-gate}. We present Reduction~\ref{reduction:MB-DQC}, show that it generates the same transcript as Protocol~\ref{protocol:mblind_DQC} (namely $\rho_{real, out, \b}$ for the same computation branch $\b$), and finally present a version using the Simulator and the Resource that performs the same steps and therefore generates the same transcript.

    \subparagraph{Presenting the Reduction.}
    To obtain Reduction~\ref{reduction:MB-DQC}, we replace the initial $n$-qubit encryption by an EPR encryption, and each "Clifford + Blind state injection" layer by its EPR-reduction version with delayed rotations, as already presented and analyzed in Reduction~\ref{protocol:blind-gate-reduction}. Finally, the Client receives $n$ bits from the Server and decodes them in the same way as in the initial protocol.

    \begin{reduction}
    \caption{MB-DQC, EPR-reduction and delayed rotations}
    \label{reduction:MB-DQC}
    \begin{algorithmic}[1]
        \Procedure{Client — EPR-encryption of $\rho$}{}
        \State Prepare $n$ EPR Pairs, and send each half to the Server.
        \State Perform Bell measurements on the other halves and $\rho$, set the $n$-bit outcome strings $\a, \r$ as encryption keys on $n$ bits.
        \EndProcedure

        \Procedure{Client — Perform Clifford and Injection layers}{}

        \For{$i=1, ..., t$}
        \State Prepare an EPR pair and send half to the Server.
        \State EPR-encrypt ancilla $\rho_{\A_i}$, part $1$, get outcome $x_i$.
        \State Set $\a'\gets \a||x_i, \r'\gets \r||0$
        \State Set $\a', \r'\gets \F_i\circ\C_{i}(\a', \r')$
        \State Receive bit $b_i$ from the Server.
        \State Decode $b_i \gets b_i \oplus a'_{n+i}$
        \State Sample $\delta_i \gets \$\Theta$ and sends it to the Server.
        \State Compute $\phi_i(\A_i, b_i)$ according to Equation~\ref{eq:blind:phi}
        \State Compute $\theta_i(\phi_i, a'_n, \delta_i)$ according to Equation~\ref{eq:blinded angle inverted}
        \State EPR-Encrypt ancilla $\rho_{\A_i}$, part $2$, pre-rotation angle $\theta$, get outcome $z_i$
        \State Set $\a\gets\a||x_i, \r\gets \r||z_i$, and re-update $\a', \r'\gets \F_i\circ \C_i(\a, \r)$
        \State Update $a'_n \gets a'_n\oplus b_i$ if $\A_i=\Tlabel$

        \EndFor
        \EndProcedure

        \Procedure{Client — Blind measurements}{}
        \State Truncate $\a, \r$ to the first $n$ bits.
        \State Receive $n$-bit string $\z$ from the Server.
        \State Update $\a, \r\gets \C_{t+1}(\a, \r)$
        \State Decode $\z\gets \z\oplus \a$
        \EndProcedure

        \State Output $\z$ as Client output if $\star$ ; else $\b||\z$.
    \end{algorithmic}
    \end{reduction}

    \subparagraph{Transcript in Reduction~\ref{reduction:MB-DQC}.}
To write the transcript that is generated when a Distinguisher interacts with Reduction~\ref{reduction:MB-DQC} instead of Protocol~\ref{protocol:mblind_DQC}, we proceed with the same logic as the proof of Lemma~\ref{lemma:mblind:Pauli deviations} in Section~\ref{subsection:proof of reduction to Pauli Dev}. We start by writing the state sent to the Server, by fixing the randomness—which here consists of the outcomes of the Bell measurements $\a, \r$ and the angle sampled uniformly by the Client $\boldsymbol\delta$—and the computation branch $\b$. We can note the following: after the EPR-encryption of $\rho$, the Bell measurement outcomes are $n$-bit strings $\a, \r$ and the Server holds $\Enc{\a}{\r}[\rho]$.

    Fixing the randomness $\boldsymbol{\delta}$ fixes the angles register to $\ketbra{\boldsymbol{\delta}}$. Then, fixing the computation branch $\b$ fixes for each $i\leq t$ the angle $\theta_i$. Using the same reasoning as for Reduction~\ref{protocol:blind-gate-reduction}, the qubit that the Server receives for each layer is in state $\Enc{x_i}{z_i}\circ\Z(\theta_i)[\rho_{\A_i}]$ after the Client's operations (EPR-encryption part $1$ and later $2$).
    Note that at each layer, one qubit is added so there is one more bit to the encryption keys: at the end $\a$ and $\r$ are $n+t$-bit encryption keys.
    Thus, without loss of generality, we can write that from the point of view of the Client that knows the outcomes $\a, \r$ and random angle $\boldsymbol{\delta}$, the qubits that are sent to the Server are, after the Client's operations, in the state
    \begin{equation}
        \rho_{red, in, \b}^{\a, \r, \boldsymbol{\delta}}
        =
        \Enc{\a}{\r}
            \circ
            \Z(\boldsymbol\theta)
            \left[\rho
            \otimes
            \rho_\A
            \otimes
            \ketbra{\boldsymbol\delta}
            \otimes
            \ketbra{0}^{\otimes w}\right]\; .
    \end{equation}

    The interaction with the malicious Server can be analyzed the same way as in the Real World. See Section~\ref{subsection:proof of reduction to Pauli Dev} for more details. Hence, after interaction, the state is
    \begin{equation}
            \rho_{red, out, \b}^{\a, \r, \boldsymbol{\delta}}
            =
            \ketbra\b
            \circ
            \Dec{\a'}{\r'}
            \circ
            \D
            \circ
            \Enc{\a'}{\r'}
            \left[
                \rho_{cor, \b, C}
                \otimes
            \ketbra{\boldsymbol\delta}
            \otimes
            \ketbra{0}^{\otimes w}
            \right]\; .
        \end{equation}
        From the point of view of the Distinguisher, the secrets $\a, \r, \boldsymbol{\delta}$ are unknown so we average over all the possible secret values (see Section~\ref{subsection:proof of reduction to Pauli Dev} for more details). Again, we clearly see that since $\boldsymbol{\delta}$ only appears in the angle register, it implements a one-time pad on $2t$ qubits and can safely be discarded. Thus,
            \begin{equation}
        \rho_{red, out, \b}
        =
        \dfrac{1}{4^{n+t}}
            \ketbra\b
            \left[
                \sum_{\a, \r\in \bin^{n+t}}
                \Dec{\a}{\r}
                \circ
                \D
                \circ
                \Enc{\a}{\r}
                \left[
                    \rho_{cor, \b, C}
                \otimes
                \ketbra{0}^{\otimes w}
                \right]
            \right]
            \; .
    \end{equation}
    We thus have $\rho_{red, out, \b} = \rho_{real, out, \b}$.

    \paragraph{Introducing the Simulator and concluding the proof.}
    We now present Simulator~\ref{simulator:MB-DQC}: in total, the operations performed by the Simulator and the Resource (via the CPTP map instructed by the Simulator) are the same as in Reduction~\ref{reduction:MB-DQC}

    \begin{simulatorr}
    \caption{MB-DQC}
    \label{simulator:MB-DQC}
    \begin{algorithmic}[1]
        \Procedure{Simulator - Emulates Client interaction with Server}{}
        \State Prepare $n$ EPR Pairs, and send each half to the Server.
        \For{$i=1, ..., t$}
        \State Prepare an EPR pair and send half to the Server.
        \State Receive bit $b_i$ from the Server.
        \State Sample $\delta_i \gets \$\Theta$ and send it to the Server.
        \EndFor
        \State Receive $n$-bit string $\z$ from the Server.
        \EndProcedure

        \Procedure{Simulator - Make Resource reproduce the Server behavior}{}
        \State Send the remaining $n+t$ EPR halves, the $n+t$-bit string $\b||\z$, and the angles $\boldsymbol\delta=(\delta_1, ..., \delta_t)$ to Ideal Resource
        \State Send the below instructions as a CPTP map to the Resource to perform.
        \EndProcedure

        \State \textbf{Resource Inputs:} $\C_1, ..., \C_{t+1}$ (public info.), $\rho$ and $\A_1, ..., \A_t$ from the Client, and $n+t$ EPR-halves, $\b||\z, \boldsymbol{\delta}$ from the Simulator.
        \Procedure{Resource - EPR-encryption of $\rho$}{}
        \State Perform Bell measurements on the received EPR pairs and $\rho$, and set the $n$-bit strings $\a, \r$ of measurement outcomes as encryption keys.
        \EndProcedure

        \Procedure{Resource - Perform Clifford and Injection layers}{}
        \For{$i=1, ..., t$}
        \State EPR-Encrypt ancilla $\rho_{\A_i}$, part $1$, get outcome $x_i$.
        \State Set $\a'\gets \a||x_i, \r'\gets \r||0$
        \State Set $\a', \r'\gets \F_i\circ\C_{i}(\a', \r')$
        \State Decode $b_i \gets b_i \oplus a'_{n+i}$
        \State Compute $\phi_i(\A_i, b_i)$ according to Equation~\ref{eq:blind:phi}
        \State Compute $\theta_i(\phi_i, a'_n, \delta_i)$ according to Equation~\ref{eq:blinded angle inverted}
        \State EPR-Encrypt ancilla $\rho_{\A_i}$, part $2$, pre-rotation angle $\theta$, get outcome $z_i$
        \State Set $\a\gets\a||x_i, \r\gets \r||z_i$, and re-update $\a', \r'\gets \F_i\circ \C_i(\a, \r)$
        \State Update $a'_n \gets a'_n\oplus b_i$ if $\A_i=\Tlabel$
        \EndFor
        \EndProcedure

        \Procedure{Resource - Perform Blind Measurements}{}
        \State Truncate $\a, \r$ to the first $n$ bits.
        \State Update $\a, \r\gets \C_{t+1}(\a, \r)$
        \State Decode $\z\gets \z\oplus \a$
        \EndProcedure

        \State Output $\z$ as Client output if $\star$ ; else $\b||\z$.
    \end{algorithmic}
    \end{simulatorr}

    Therefore, for the same randomness we have $\rho_{red, out, \b}^{\a, \r, \boldsymbol{\delta}} = \rho_{ideal, out, \b}^{\a, \r, \boldsymbol{\delta}}$, where $\rho_{ideal, out, \b}^{\a, \r, \boldsymbol{\delta}}$ is the transcript in the Ideal World.
    Hence, using the same reasoning, we conclude that $\rho_{ideal, out, \b} = \rho_{real, out, \b}$, and thus the protocol is perfectly secure.

\end{proof}

\section{Proofs for Verification}
\label{section: verification proofs}

First, we remind useful tools from probability theory.
\begin{lemma}[Hoeffding bound for the binomial distribution]
    \label{lemma:hoeffding:binomial}
    Let $X\sim \mathrm{Binomial}(n,p)$.
    For any $k\le np$,
    \[
    \Pr[X\le k]
    \;\le\;
    \exp\!\left(-\frac{2(np-k)^2}{n}\right),
    \]
    and for any $k\ge np$,
    \[
    \Pr[X\ge k]
    \;\le\;
    \exp\!\left(-\frac{2(np-k)^2}{n}\right).
    \]
    \end{lemma}

    \begin{lemma}[Hoeffding bounds for the hypergeometric distribution]
        \label{lemma:hoeffding:hypergeometric}
        Let $X \sim \mathrm{Hypergeometric}(N,K,n)$ be the number of marked items when drawing $n$ samples without replacement from a population of size $N$ containing $K$ marked elements.
        Then $\mathbb{E}[X]=\frac{K}{N}n$.
        For any $\chi\ge 0$ such that
        \[
        \frac{K}{N}-\chi \ge 0
        \quad\text{and}\quad
        \frac{K}{N}+\chi \le 1,
        \]
        we have
        \[
        \Pr\!\left[
        X \le \Big(\frac{K}{N}-\chi\Big)n
        \right]
        \;\le\;
        \exp(-2\chi^2 n),
        \]
        and
        \[
        \Pr\!\left[
        X \ge \Big(\frac{K}{N}+\chi\Big)n
        \right]
        \;\le\;
        \exp(-2\chi^2 n).
        \]
        \end{lemma}

Now, we prove the main Lemma useful in the security proof of the Verification protocol, namely Lemma~\ref{lemma:security error}, that we hereby re-state.
\securityerror*

\begin{proof}
In the following, we let $\E\in\Pauli_{N\times(n+t)}$ be a fixed Pauli deviation on $N$ rounds, where $N=d+s$, and let $\sigma$ be a partition of $[N]$ drawn uniformly at random.
Then, we can define the following random variables:
\begin{itemize}
    \item $Z$ describing the number of computation rounds affected, meaning the number of $i$ ($\leq d$) yielding a decision bit $\b^{(i)}_1 = y^{(i)}\neq z^\star$ (either by the Server deviation or by inherent failure of the algorithm);
    \item $Y$ describing the number of failed test rounds, meaning the number of $i$ ($>d$) yielding measurement outcomes $\b^{(i)}$ such that $b^{(i)}_{q_i}\neq0$.
\end{itemize}
The statement is thus about bounding the quantity $\max_\E \Pr[Y < w \wedge Z > d/2]$.

\begin{enumerate}[label=\arabic*)]
    \item
    In what follows, we will parametrize the Server cheating strategy by the number of rounds on which it decides to act with a non-trivial deviation.
    We use $\E^{(i)}\in\Pauli_{n+t}$ for the ``per-round Pauli deviation'', meaning such that $\E=\bigotimes_{i\leq N}\E^{(i)}$.
    \item
    With this notation, we can parametrize the deviation on the number of rounds on which it is harmful (see Definition \ref{verification:def:harmful}).
    Let us note $m$ that number. Formally, $$
    m=
    \Big|\left\{ i\leq N \; \text{s.t}\;  W_{\X,\Y}\left( \E^{(i)} \right)\cap Q\neq \emptyset  \right\}\Big|
    .$$ The distinguishing advantage parametrized by $m$ writes
    \begin{equation}
        \label{eq:pd}
        p_d = \max_{m\leq N} \mathrm{Pr}\left[Z > \frac{d}{2} \wedge Y<w\right].
    \end{equation}
    \item
    Now, we focus on the values of $m$ around which the highest values of $p_d$ are met.
    Note that we can always let $m_0\leq N$ and decompose Equation~\ref{eq:pd} into two regimes, $m\leq m_0$ and $m>m_0$.
    \begin{eqnarray}
        \label{eq:pd_dec}
        p_d
        &=& \max
        \left\{
            \max_{0< m\leq m_0}\mathrm{Pr}\left[Z > \frac{d}{2} \wedge Y<w\right],
            \max_{m_0 < m \leq N}\mathrm{Pr}\left[Z > \frac{d}{2} \wedge Y<w\right]
         \right\}
         \\
        &=& \max
        \left\{
            \max_{0< m\leq m_0}\mathrm{Pr}\left[Z > \frac{d}{2}\right] \mathrm{Pr} [Y<w],
            \max_{m_0 < m \leq N}\mathrm{Pr}\left[Z > \frac{d}{2}\right] \mathrm{Pr} [Y<w]
         \right\}
         \\
        &\leq& \max
        \left\{
            \max_{m\leq m_0}\mathrm{Pr}\left[Z > \frac{d}{2}\right] ,
            \max_{m> m_0} \mathrm{Pr} [Y<w]
         \right\}
    \end{eqnarray}
    The reason we did that simplification is the following: when $m$ is upper-bounded, the least probable event is $Z>\frac{d}{2}$, as this event captures the probability of \textit{many computation rounds} to be detected.
    Likewise, when $m$ is lower-bounded, it means that at least $m_0$ rounds are being attacked: in this case the least probable event is to be undetected by those deviations.
    We will in~\ref{item: m_0} develop an intuition on where to place the barrier $m_0$.
    \item
    For the moment, let us notice that in both terms, the maximal value is obtained at $m=m_0$.
    The distinguishing advantage can therefore be written as
    \begin{equation}
        p_d = \max(\epsilon, \nu)
    \end{equation}
    where
    \begin{equation}
        \nu =
        \mathrm{Pr}\left[Z > \frac{d}{2} \;|\; m_0 \; \text{rounds attacked}\right]
        ,\qquad
        \epsilon =
        \mathrm{Pr} [Y<w\;|\; m_0 \; \text{rounds attacked}]
    \end{equation}
    \item With that being said, let us first calculate $\epsilon$.
    First, let $X$ denote a random variable describing the number of test rounds affected by a non-trivial deviation.
    Then, let $x>0$. We can always write
    \begin{eqnarray}
        \epsilon &=&
        \mathrm{Pr} [Y<w\;|\; X>x] \quad \mathrm{Pr}[X>x]
        +
        \mathrm{Pr} [Y<w\;|\; X\leq x] \quad \mathrm{Pr}[X\leq x]
        \\
        &\leq&
        \mathrm{Pr}[X\leq x]
        +
        \mathrm{Pr} [Y<w\;|\; X> x]
    \end{eqnarray}
    From there, we can notice that $X$ is naturally upper-bounded in the stochastic order by a hypergeometrically-distributed random variable $\tilde X \sim H(N, m_0, s)$.
    We use this fact to compute the first term of the sum.
    By using the tail bounds for such the hypergeometric variable, we have for $\chi_\epsilon\geq 0$:
    \begin{equation}
        \mathrm{Pr}\left[
            \tilde X \leq \left(
                \dfrac{m_0}{N} - \chi_\epsilon
             \right)s
         \right] \leq \exp(-2\chi_\epsilon^2 s)
    \end{equation}
    Using this tail bounds, the parametrization on $\chi_\epsilon$ requires to set $x = ({m_0}/{N} - \chi_\epsilon)s$ and we get
    \begin{equation}
        \mathrm{Pr}[X\leq x]
        \leq
        \mathrm{Pr}\left[
            \tilde X \leq x
         \right]
         \leq \exp(-2\chi_\epsilon^2 s)
    \end{equation}
    Now we compute an upper-bound on the second term of the sum, noting that $$\mathrm{Pr} [Y<w\;|\;  X> x] \leq \mathrm{Pr} [Y<w\;|\; X= x]$$ so that $X$ is fixed and parametrized according to $\chi_\epsilon$ as previously.
    Then, we notice that $Y$ conditioned on $X=x$ is upper-bounded in the stochastic order by a binomially-distributed variable $\tilde Y \sim (x, 1/k)$.
    Indeed: the number of affected test rounds has been fixed to $x$, there are $k=|\mathcal Q|$ types of traps, and any non-trivial deviation is detected by at least one type of trap.
    For Protocol~\ref{protocol:verification}, this is guaranteed by Lemma~\ref{lemma:harmful deviations are detected}, where the set of traps $\mathcal Q$ defined by \ref{eq: set of traps} was used. Consequently, $k=1+t$.
    The probability that a given test round detects the non-trivial deviation of the Server is thus at least $1/k$.
    The worst case is when the deviation is detected by only one type of trap, so we can fix the detection rate to $1/k$.
    Using the tail bounds for the binomial variable, we have for
    \begin{equation}
        \mathrm{Pr}\left[
            \tilde Y < w
         \right]
         \leq
         \exp\left[
            -2\dfrac{(x/k-w)^2}{x} \; \right]
    \end{equation}
    for $w < x/k$.
    In the parametrization on $\chi_\epsilon$, this means $\chi_\epsilon < \frac{m_0}{N}-\frac{w}{s}k$.
    In short, we can upper-bound $\epsilon$ as follows:
    \begin{equation}
        \label{eq:epsilon}
        \epsilon \leq
        \min_{\chi_\epsilon\in[0, \frac{m_0}{N}-\frac{w}{s}k]}
        \exp(-2\chi_\epsilon^2 s) + \exp\left[
            -2\dfrac{(
                (\frac{m_0}{N}-\chi_\epsilon)
                \frac{1}{k}-\frac{w}{s}
                 )^2}{
                \frac{m_0}{N}-\chi_\epsilon
            } \; s
          \right]
    \end{equation}

    \item Then, let us calculate $\nu$.
    To this end, let us decompose the random variable $Z$ into $Z_1$ and $Z_2$, where $Z_1$ describes the number of computation rounds affected by a non-trivial deviation while $Z_2$ describes the number of computation rounds non-affected by the Server deviation but that nevertheless yield an incorrect outcome due to the probabilistic nature of the BQP computation.
    Then, for $z_1>0$, $\nu$ can be written as
    \begin{eqnarray}
        \nu &=& \mathrm{Pr}\left[Z_1 + Z_2 > \frac{d}{2}\right] \\
        &\leq& \mathrm{Pr}\left[Z_1 + Z_2 > \frac{d}{2} \;\; | Z_1> z_1\right] \;\; \mathrm{Pr}[Z_1>z_1]
            + \mathrm{Pr}\left[Z_1 + Z_2 > \frac{d}{2} \;\; | Z_1\leq  z_1\right] \;\; \mathrm{Pr}[Z_1\leq z_1] \\
        &\leq& \mathrm{Pr}[Z_1>z_1]
            + \mathrm{Pr}\left[Z_1 + Z_2 > \frac{d}{2} \;\; | Z_1\leq  z_1\right] \\
    \end{eqnarray}
    Let us now study both terms separately.
    The first term can be upper-bounded as follows. $Z_1$ is lower-bounded in the usual stochastic order by a $H(N, m_0, d)$ hypergeometrically distributed random variable $\tilde Z_1$.
    Therefore, using the tails bound for such a distribution, we get for $\chi_\nu>0$, the desired bound if we re-write $z_1 = (m_0/N+\chi_\nu)d$:
    \begin{equation}
        \mathrm{Pr}\left[ Z_1 > \left( \dfrac{m_0}{N} + \chi_\nu \right)d \right] \leq \exp(-2\chi_\nu^2 d)\; .
    \end{equation}
    For the first term, let us note that
    \begin{equation}
        \mathrm{Pr}\left[Z_1 + Z_2 > \frac{d}{2} \;\; | Z_1\leq  z_1\right] \leq \mathrm{Pr}\left[Z_2 > \frac{d}{2} - z_1 \;\; | Z_1\leq  z_1\right] \leq \mathrm{Pr}\left[Z_2 > \frac{d}{2} - z_1 \;\; | Z_1= z_1\right]
    \end{equation}
    Now, with $Z_1$ fixed to $z_1$, $Z_2$ becomes upper-bounded in the usual stochastic order by a $B(d-z_1, \bqp)$ binomially-distributed random variable $\tilde Z_2$, as among the $d$ computation rounds, only $d-z_1$ are non-affected by a deviation, and the failure probability is $\bqp$ for each.
    The tails bound for $X \sim B(n,\bqp)$ states that, for $n\bqp \le k$
    \begin{equation}
        \mathrm{Pr}[X> k] \leq \exp\left[
            -2\dfrac{(n\bqp-k)^2}{n}
         \right]
    \end{equation}
    Here, replacing $n$ by $d-z_1 = d(1-\frac{m_0}{N}-\chi_\nu)$, and $k$ by $\frac{d}{2}-z_1 = d(\frac{1}{2}-\frac{m_0}{N}-\chi_\nu)$, we get for $(d-z_1)\bqp < \frac{d}{2}-z_1$, or equivalently $\chi_\nu < \frac{1-2\bqp}{2-2\bqp}-\frac{m_0}{N}$ :
    \begin{equation}
        \mathrm{Pr}\left[Z_2 > \frac{d}{2}-z_1\right] \leq \exp\left[
            -2
            \dfrac{
                [\left( 1-\frac{m_0}{N}-\chi_\nu \right)d\bqp - \left( \frac{1}{2}-\frac{m_0}{N}-\chi_\nu \right)d]^2
            }{
                \left( 1-\frac{m_0}{N}-\chi_\nu \right)d
            }
        \right]
    \end{equation}
    Finally, after simplifying and putting back the pieces together, we get the following
    \begin{equation}
        \label{eq: nu}
        \nu \leq \min_{
            \chi_\nu \in [0, \frac{1-2\bqp}{2-2\bqp}-\frac{m_0}{N}]
        }
        \exp[ -2\chi_\nu^2 d ] + \exp\left[
            -2
            \dfrac{
                \left( (1-\frac{m_0}{N}-\chi_\nu)(1-\bqp)-\frac{1}{2} \right)^2
            }{1-\frac{m_0}{N}-\chi_\nu} d
         \right]
    \end{equation}

    \item \label{item: m_0}
    Lastly, we need to develop an intuition on where to place $m_0$.
    For this barrier to be meaningful, let the following scenario, in which a fraction $f$ of the total rounds are attacked by the Server with a non-trivial deviation.
    Note that to obtain the classical output, the Client performs a majority vote over the obtained results.
    Therefore, we are interested in quantifying the critical fraction that corrupts more than half of the outcomes.
    However, we know that there is an inherent probability $\bqp$ of failure for a BQP computation.
    Hence, our criteria becomes the following: if $f$ is the fraction of rounds attacked by a Server deviation, then this must satisfy $(1-\bqp)(1-f)>1/2$, in the sense the fraction $1-f$ of non-attacked rounds that succeed (with probability $1-\bqp$) are more than one half of the total rounds.
    This becomes $f<(1-2\bqp)/(2-2\bqp)$. We then let $\varphi>0$ and parametrize $m_0$ according to how much the Server is far or close to attacking the critical fraction.
    \begin{equation}
        m_0 = \left( \dfrac{1-2\bqp}{2-2\bqp}-\varphi\right)N
    \end{equation}
    We can thus replace $m_0$ by this expression parametrized by $\varphi$ in Equations~\ref{eq:epsilon} and~\ref{eq: nu}.
    \item Then, it means that for different values of $\varphi$, different values of $\epsilon, \nu$ are obtained: these have become function of $\varphi$.
    \begin{equation}
        \label{eq: epsilon_varphi}
        \epsilon(\varphi) \leq
        \min_{\chi_\epsilon\in[0, (\frac{1-2\bqp}{2-2\bqp}-\frac{w}{s}k)-\varphi ]}
        \exp(-2\chi_\epsilon^2 s) + \exp\left[
            -2\dfrac{( (\frac{m_0}{N}-\chi_\epsilon) \frac{1}{k}-\frac{w}{s} )^2}{ \frac{m_0}{N}-\chi_\epsilon } \; s
          \right]
    \end{equation}

    \begin{equation}
        \label{eq: nu_varphi}
        \nu(\varphi) \leq \min_{ \chi_\nu \in [0, \varphi] }
        \exp[ -2\chi_\nu^2 d ] + \exp\left[
            -2 \dfrac{ \left( (1-\frac{m_0}{N}-\chi_\nu)(1-\bqp)-\frac{1}{2} \right)^2 }{1-\frac{m_0}{N}-\chi_\nu} d
         \right]
    \end{equation}

    However, the distinguishing advantage can always be upper-bounded by the minimal one.
    Furthermore, note that for the bounds on $\chi_\epsilon$ in Equation~\ref{eq: epsilon_varphi} to be well-defined, we need to have $\varphi < \frac{1-2\bqp}{2-2\bqp}-\frac{w}{s}k$.
    Finally, the distinguishing advantage can be written as follows:
    \begin{equation}
        p_d \leq \min_{\varphi \in [0, \frac{1-2\bqp}{2-2\bqp}-\frac{w}{s}k]} (\max(\nu(\varphi), \epsilon(\varphi)))
    \end{equation}
    where again, as a reminder, $k=|\mathcal Q|$.
    For Protocol \ref{protocol:verification}, $k=1+t$.

    \item Now, to set a concrete bound, we can set the buffers $\chi_\epsilon, \chi_\nu, \varphi$ to concrete values.
    Indeed, one can re-define the following constants:
    \begin{itemize}
        \item $\alpha = (1-2\bqp)/(2-2\bqp)$
        \item $\Delta = \alpha-wk/s$
        \item $m_0/N = \alpha-\varphi$
    \end{itemize}
    With this, the range of $\varphi$ is $[0, \Delta]$, so we can set $\varphi = \Delta/2$, in the middle of the range.
    Then, the range for $\chi_\epsilon$ becomes $[0, \Delta/2]$, thus we can set $\chi_\epsilon = \chi_0:=\Delta/4$, again in the middle of the range.
    Finally, for $\chi_\nu$ in the expression of $\nu$, the range is $[0, \varphi]=[0, \Delta/2]$ so we can again set $\chi_\nu=\chi_0$.
    When we replace these buffer values in the above expressions, we get values of $\epsilon, \nu$ that upper-bound the actual values suggested by the minimisation problem above, but nevertheless negligible in the number of rounds $d, s$:
    \begin{align}
        \epsilon &\leq \exp( -2(\Delta/4)^2 s ) + \exp\left( -2\dfrac{ ((\alpha-3\Delta/4)/k-\frac{w}{s})^2 }{\alpha-3\Delta/4} s \right)
        \\
        \nu&\leq \exp( -2(\Delta/4)^2 d ) + \exp\left( -2\dfrac{ (((1-\alpha)+\Delta/4)(1-\bqp)-1/2)^2 }{(1-\alpha)+\Delta/4} d \right)
    \end{align}

\end{enumerate}
\end{proof}

\end{document}